\documentclass[nospthms]{svjour3} % onecolumn (ditto)
\smartqed % flush right qed marks, e.g. at end of proof

\usepackage{amsmath, mathrsfs, amscd, amsfonts, amssymb, graphicx, color}
\usepackage[bookmarksnumbered, colorlinks, plainpages]{hyperref}
\hypersetup{colorlinks=true,linkcolor=blue, anchorcolor=blue, citecolor=blue, urlcolor=blue, filecolor=blue, pdftoolbar=true}
\usepackage{amsmath, amssymb, amsthm, amscd}
\numberwithin{equation}{section}

\usepackage[utf8]{inputenc}
\usepackage{bbm}
\usepackage{mathrsfs,mathtools}
\usepackage[text={5.8in,8.8in},centering]{geometry}
\usepackage{tikz}
\usepackage{caption}
\usepackage{graphicx}
\usepackage{float}
\usepackage{subcaption}
\usepackage{enumitem}
\usepackage{soul}
\usepackage{color}

\newtheorem{theorem}{Theorem}[section]
\newtheorem{proposition}[theorem]{Proposition}
\newtheorem{lemma}[theorem]{Lemma}

\newtheorem{corollary}[theorem]{Corollary}
\newtheorem{remark}[theorem]{Remark}

\renewcommand{\i}{\mathrm{i}}
\newcommand{\e}{\mathrm{e}}

\newcommand{\I}[2]{\mathcal{I}_{#1, #2}}
\newcommand{\K}[2]{\mathcal{K}_{#1, #2}}
\newcommand{\J}[2]{\mathcal{J}_{#1, #2}}

\newcommand{\Hpm}[2]{\mathcal{H}^\pm_{#1, #2}}

\newcommand{\Wr}{\mathscr{W}}
\newcommand{\cH}{\mathcal{H}}

\newcommand{\cI}{\mathcal{I}}
\newcommand{\cK}{\mathcal{K}}

\newcommand{\cX}{\mathcal{X}}
\newcommand{\cJ}{\mathcal{J}}
\newcommand{\cY}{\mathcal{Y}}
\newcommand{\cD}{\Dom}

\newcommand{\pder}{\partial}

\newcommand{\suma}[2]{\sum\limits_{#1}^{#2}}
\newcommand{\loc}{\mathrm{loc}}
\newcommand{\naw}[1]{\left( {#1} \right)}
\newcommand{\abs}[1]{{\left| #1 \right|}}

\renewcommand{\Re}{\mathrm{Re}}
\renewcommand{\Im}{\mathrm{Im}}

\newcommand{\bbC}{\C}
\newcommand{\bbR}{\mathbb{R}}

\def\H{\mathcal H}
\def\z{\mathbf{z}}
\def\U{\mathcal U}

\newcommand{\Dom}{{\mathcal D}}

\newcommand{\N}{\mathrm{N}}

\def\Ka{\mathcal K}
\def\Ia{\mathcal I}
\def\Ha{\mathcal H}
\def\Ja{\mathcal J}

\def\R{\mathbb R}
\def\C{\mathbb C}
\def\N{\mathbb N}
\def\Z{\mathbb Z}

\def\B{\mathcal B}
\def\d{\mathrm d}

\def\z{\mathsf z}
\def\max{\mathrm{max}}
\def\min{\mathrm{min}}

\smartqed

\journalname{AIOT}

\begin{document}

\title{On radial Schr\"odinger operators with a Coulomb potential: \\ General boundary conditions}

\author{
Jan~Derezi\'{n}ski \and
J\'er\'emy~Faupin \and
Quang~Nhat~Nguyen \and Serge~Richard}

\institute{
Jan Derezi\'{n}ski \at
Department of Mathematical Methods in Physics, Faculty of Physics,
University of Warsaw, ul. Pasteura 5, 02-093 Warszawa, Poland\\ \email{jan.derezinski@fuw.edu.pl}
\and
J\'er\'emy Faupin
\at Institut Elie Cartan de Lorraine,  Universit\'e de Lorraine
UFR MIM, 3 rue Augustin Fresnel.   57073 Metz Cedex 03, France \\
\email{jeremy.faupin@univ-lorraine.fr}
\and Quang Nhat Nguyen
\at
Graduate school of mathematics,
Nagoya University,
Chikusa-ku,
Nagoya 464-8602, Japan \\
\email{nguyen.quang.nhat@d.mbox.nagoya-u.ac.jp}
\and
 Serge Richard
\at
Graduate school of mathematics,
Nagoya University,
Chikusa-ku,
Nagoya 464-8602, Japan,
and on
leave of absence from Univ.~Lyon, Universit\'e Claude Bernard Lyon 1, CNRS UMR 5208,
Institut Camille Jordan, 43 blvd.~du 11 novembre 1918, F-69622 Villeurbanne cedex,
France\\
\email{richard@math.nagoya-u.ac.jp}
}

\dedication{This paper is dedicated
 to Prof. Franciszek Hugon Szafraniec}

\date{Received: date / Accepted: date}

\maketitle

\begin{abstract}
This paper presents the spectral analysis of
1-dimensional Schr\"odinger operator on the half-line whose potential
is a linear combination of the {\em Coulomb term} $1/r$
and the {\em centrifugal term} $1/r^2$.
The coupling constants are allowed to be complex,
and all possible boundary conditions at $0$ are considered.
The resulting closed operators are organized in
three holomorphic families.
These operators are closely related to the {\em Whittaker equation}.
Solutions of this equation are thoroughly studied in a large appendix to this paper.
Various special cases of this equation are analyzed, namely the {\em degenerate}, the {\em Laguerre} and the {\em doubly degenerate} cases.
A new solution to the Whittaker equation in the doubly degenerate case is also introduced.
\end{abstract}



\tableofcontents

\section{Introduction}

This paper is devoted to 1-dimensional Schr\"odinger operators with Coulomb and centrifugal potentials. These operators are given by the differential expressions
\begin{equation}\label{whit1.}
L_{\beta,\alpha } :=-\partial_x^2+\Big(\alpha -\frac14\Big)\frac{1}{x^2}-\frac{\beta}{x}.
\end{equation}
The parameters $\alpha$ and $\beta$ are allowed to be complex valued.
We shall study realizations of $L_{\beta,\alpha}$
as closed operators on $L^2(\R_+)$, and consider general boundary conditions.

The operator given in \eqref{whit1.} is one of the most famous and useful exactly
solvable models of Quantum Mechanics. It describes the radial part of the Hydrogen Hamiltonian.
In the mathematical literature, this operator goes back to Whittaker, who studied its eigenvalue equation in \cite{Whi}. For this reason, we call \eqref{whit1.} the {\em Whittaker operator}.

This paper is a continuation of a series of papers  \cite{BDG,DR1,DR2} devoted to an  analysis of exactly solvable 1-dimensional Schr\"odinger operators. We  follow the same philosophy as in \cite{DR1}. We start from a formal differential expression depending on  complex parameters. Then we look for closed realizations of this operator on $L^2(\R_+)$. We do not restrict ourselves to self-adjoint realizations -- we look for realizations  that are {\em well-posed}, that is, possess non-empty resolvent sets. This implies that they satisfy an appropriate boundary condition at $0$, depending on an additional complex parameter. We organize those operators in holomorphic families.

Before describing the holomorphic families introduced in this paper, let us recall the main constructions from the previous papers of this series. In \cite{BDG,DR1} we
considered the operator
\begin{equation}\label{bessel}
L_{\alpha } :=-\partial_x^2+\Big(\alpha -\frac14\Big)\frac{1}{x^2}.
\end{equation}
As is known, it is useful to set $\alpha=m^2$. In \cite{BDG} the following holomorphic family of closed realizations of \eqref{bessel} was introduced:
\begin{align*}
&  H_m,\quad \text{with }-1<\Re(m),\\
& \text{defined by $L_{m^2}$ with boundary conditions} \ \sim x^{\frac12+m}.
\end{align*}
It was proved that for $\Re(m)\geq1$ the operator $H_m$ is the only closed realization of $L_{m^2}$. In the region $-1<\Re(m)<1$ there exist realizations of $L_{m^2}$ with mixed boundary conditions. As described in \cite{DR1}, it is natural to organize them into two holomorphic families:
\begin{align*}
&  H_{m,\kappa},\quad  \text{with }-1<\Re(m)<1,\ m\neq0, \ \kappa\in\C\cup\{\infty\},\\
& \text{defined by $L_{m^2}$ with boundary conditions} \ \sim x^{\frac12+m}+\kappa  x^{\frac12-m},
\end{align*}
and
\begin{align*}
&  H_0^\nu,\quad  \text{with }\nu\in\C\cup\{\infty\},\\
& \text{defined by $L_{0}$ with boundary conditions} \ \sim x^{\frac12}\big(\nu+\ln(x)\big).
\end{align*}
Note that related investigations about these operators have also been performed in \cite{Smi1,Smi2}.

In \cite{DR2} and in the present paper we study closed realizations of the differential operator \eqref{whit1.} on $L^2(\R_+)$. Again, it is useful to set $\alpha=m^2$. In \cite{DR2} we introduced the  family
\begin{align*}
& H_{\beta,m},\quad  \text{with }\beta\in\C,\ -1<\Re(m),\\
&\text{defined by $L_{\beta,m^2}$ with boundary conditions} \ \sim x^{\frac12+m}\Big(1-\frac{\beta}{1+2m} x  \Big).
\end{align*}
It was noted in this reference that this family is holomorphic except for a singularity at $(\beta,m)=\big(0,-\frac12\big)$, which corresponds to the Neumann Laplacian.

For $\Re(m)\geq1$ the operator $H_{\beta,m}$ is also the only closed realization of $L_{\beta,m^2}$. In the region $-1<\Re(m)<1$ there exist other closed realizations of $L_{\beta,m^2}$. The boundary conditions corresponding to $H_{\beta,m}$ are distinguished---we will call them {\em pure}. The goal of the present paper is to describe the most general well-posed realizations of $L_{\beta,m^2}$, with all possible boundary conditions, including the {\em mixed} ones.

We shall show that it is natural to organize all well-posed realizations of $L_{\beta,m^2}$ for  $-1<\Re(m)<1$
in three holomorphic families: The generic family
\begin{align*}
&  H_{\beta,m,\kappa},\quad  \text{with }
\beta\in\C,\ -1<\Re(m)<1,\  m\not \in \big\{-\tfrac12,0,\tfrac12\big\},\ \kappa\in\C\cup\{\infty\},\\
&\text{defined by $L_{\beta,m^2}$ with boundary conditions} \\
&\hspace{15ex}\sim x^{\frac12+m}\Big(1-\frac{\beta}{1+2m} x  \Big)
+\kappa x^{\frac12-m}\Big(1-\frac{\beta}{1-2m} x  \Big),
\end{align*}
the family for $m=0$
\begin{align*}
&  H_{\beta,0}^\nu,\quad  \text{with } \beta\in\C,\ \nu\in\C\cup\{\infty\},\\
&\text{defined by $L_{\beta,0}$ with boundary conditions} \ \sim
x^{\frac12}\big(\nu+\ln(x)\big),
\end{align*}
and the  family for $m=\frac{1}{2}$
\begin{align*}
&  H_{\beta,\frac12}^\nu,\quad  \text{with }
\beta\in\C,\ \nu\in\C\cup\{\infty\}\\
&\text{defined by $L_{\beta,\frac14}$ with boundary conditions} \ \sim
1-\beta x\ln(x)+\nu x.
\end{align*}
The above holomorphic families include all possible well-posed realizations of $L_{\beta,m^2}$ in the region $|\Re(m)|<1$ with one exception: the special case $(\beta,m,\kappa)=\big(0,-\frac12,0\big)$ which corresponds to the Neumann Laplacian $H_{-\frac12}=H_{-\frac12,0}=H_{\frac12,\infty}$, and which is already covered by the families  $H_m$ and $H_{m,\kappa}$.

After having introduced these families and describing a few general results, we provide the spectral analysis of these operators and give the formulas for their resolvents. We also describe the eigenprojections onto  eigenfunctions of these operators. They can be organized into a single family of bounded $1$-dimensional projections $P_{\beta,m} (\lambda)$ such that
$L_{\beta,m}^{\max}P_{\beta,m}(\lambda)=\lambda P_{\beta,m}(\lambda)$. Here $L_{\beta,m}^{\max}$ denotes the maximal operator which is introduced
in Section \ref{sec_max_min}.

There exists a vast literature devoted to Schr\"odinger operators with Coulomb potentials, including various boundary conditions.
Let us mention, for instance, an interesting dispute in Journal of Physics A \cite{FLM,Ku1,Ku2} about self-adjoint extensions of the 1-dimensional Schr\"odinger operator on the real line with a Coulomb potential (without the centrifugal term). Papers \cite{FuLan,KoTe,KuLu}
 discuss generalized Nevanlinna functions naturally appearing in the context of such operators,
especially in the range of parameters $|\Re(m)|\geq1$.
See also \cite{BG,Dol,Gas,Ges,GPT,GTV,Gui,Her,Hum,MOC,Mic,Muk,MZ,Sea} and references therein.
However, essentially all these references are devoted to real parameters $\beta,m$ and self-adjoint realizations of  Whittaker operators.
The philosophy of using holomorphic families of closed operators, which we believe should be one of the standard approaches to the study of special functions, seems to be confined to the series of paper \cite{BDG,DR1,DR2}, which we discussed above.

The main reason why we are able to analyze the operator \eqref{whit1.} so precisely is the fact that it is closely related to an exactly solvable equation,  the so-called {\em Whittaker equation}
\begin{equation*}
\bigg(-\partial_z^2+\Big(m^2-\frac14\Big)\frac{1}{z^2}-\frac{\beta}{z}+\frac{1}{4}\bigg)f(z)=0.
\end{equation*}
Its solutions are called {\em Whittaker functions}, which can be expressed in terms of {\em Kummer's confluent functions}.
The theory of the Whittaker equation is the second subject of the paper. It is extensively developed in a large appendix to this paper. It can be viewed as an extension of the theory of Bessel and Whittaker equation presented in \cite{DR1,DR2}. We discuss in detail various special cases: the {\em degenerate}, the {\em Laguerre} and the {\em doubly degenerate} cases. Besides the well-known Whittaker functions  $\cI_{\beta,m}$ and  $\cK_{\beta,m}$, described for example in \cite{DR2}, we introduce a new kind of Whittaker functions, denoted $\cX_{\beta,m}$. It is needed to fully describe the doubly degenerate case.

The Whittaker equation and its close cousin, the confluent equation, are discussed in many standard monographs, including \cite{AS,confluent,Sl}. Nevertheless, it seems that our treatment  contains a number of  facts about the Whittaker equation, which could not be found in the literature. For example, we have never seen a satisfactory detailed treatment  of the doubly degenerate case. The function $\cX_{\beta,m}$ seems to be our invention. Without this function it would be difficult to analyze the doubly degenerate case. Figures \ref{fig:M1} and \ref{fig:M2}, which illustrate the intricate structure of the degenerate, Laguerre and doubly degenerate cases, apparently appear for the first time in the literature. Another result that seems to be new is a set of explicit formulas for integrals involving
products of solutions of the Whittaker equation. These formulas are related to the eigenprojections of the Whittaker operator.

%%%%%%%%%%%%%%%%%%%%%%%%%%%%%%%%%%%%%%%%%%%%%%%%%%%%%%%%%%%%%%%%%%%%%%%%%%%%%%%%
\section{The Whittaker operator}\label{sec_Whi_op}
\setcounter{equation}{0}
%%%%%%%%%%%%%%%%%%%%%%%%%%%%%%%%%%%%%%%%%%%%%%%%%%%%%%%%%%%%%%%%%%%%%%%%%%%%%%%%

In this section we define the main objects of our paper: the Whittaker operators $H_{\beta,m,\kappa}$, $H_{\beta,\frac12}^\nu$ and $H_{\beta,0}^\nu$ on the Hilbert space $L^2\big(]0,\infty[\big)$.

\subsection{Notations}

We shall use the notations $\R_+=]0,\infty[$, $\N=\{0,1,2,\dots\}$ and $\N^\times = \{ 1 , 2 , \dots \}$. Likewise, we set $\C^\times = \C \setminus \{ 0 \}$ and  $\R^\times = \R \setminus \{ 0 \}$.
We will often consider functions on the Riemann sphere $\C\cup\{\infty\}$ with  the convention $\frac10=\infty$, $\frac1\infty=0$.
 Besides, $\alpha \;\!\infty=\infty$ for any $\alpha\in \C\setminus\{0\}$ and  $\infty + \tau = \infty$.

The Hilbert space $L^2(\R_+)$ is endowed with the scalar product
\begin{equation*}
(h_1|h_2)=\int_0^\infty \overline{h_1(x)}\;\!h_2(x)\;\!\d x.
\end{equation*}
We will also use the bilinear form defined by
\begin{equation*}
\langle h_1|h_2\rangle=\int_0^\infty h_1(x)\;\!h_2(x)\;\!\d x.
\end{equation*}

The Hermitian conjugate of an operator $A$ is denoted by $A^*$. Its transpose is denoted by $A^\#$. If $A$ is bounded, then $A^*$ and $A^\#$ are defined by the relations
\begin{align*}
(h_1|Ah_2)&=(A^*h_1|h_2),\\
\langle h_1|Ah_2\rangle&=\langle A^\#h_1|h_2\rangle.
\end{align*}
The definition of $A^*$ has the well-known generalization to the unbounded case. The definition of $A^\#$ in the unbounded case is analogous.

The following holomorphic functions are understood as their \emph{principal bran\-ches}, that is, their domain is $\C\setminus]-\infty,0]$ and on $]0,\infty[$ they coincide with their usual definitions from real analysis: $\ln(z)$, $\sqrt z$, $z^\lambda$.
We set $\arg (z):=\Im \big(\ln(z)\big)$.
Sometimes it will be convenient to include in the domain of our functions
two copies of $]-\infty,0[$, describing the limits from the
upper and lower half-plane. They correspond to
the limiting cases $\arg(z)=\pm\pi$.

The Wronskian of two continuously differentiable functions $f$ and $g$ on $\R_+$ is denoted by $\Wr(f,g;\cdot)$ and is defined for $x\in \R_+$ by
\begin{equation}\label{wronskian}
\Wr( f,g;x):=f(x)g'(x)-f'(x)g(x).
\end{equation}

\subsection{Zero-energy eigenfunctions of the Whittaker operator}\label{sec_zero-e}

In order to study the realizations of the Whittaker operator $L_{\beta,\alpha}$ one first needs to find out what are the possible boundary conditions at zero. The general theory of 1-dimensional Schr\"odinger operators says that there are two possibilities:
\begin{enumerate}
\item[(i)] there is a 1-parameter family of boundary conditions at zero,
\item[(ii)] there is no need to fix a boundary condition at zero.
\end{enumerate}
One can show that    (i)$\Leftrightarrow$(i') and     (ii)$\Leftrightarrow$(ii'), where
\begin{enumerate}
\item[(i')] for any $\lambda\in\C$ the space of solutions of $(L_{\beta,\alpha}-\lambda)f=0$ which are square integrable around zero is 2-dimensional,
\item[(ii')] for any $\lambda\in\C$ the space of solutions of $(L_{\beta,\alpha}-\lambda)f=0$ which are square integrable around zero is at most 1-dimensional.
\end{enumerate}
We refer to \cite{DG2019} and references therein for more details.

In the above criterion one can choose a convenient $\lambda$. In our case the simplest choice corresponds to $\lambda=0$. Therefore, we first discuss solutions of the zero eigenvalue Whittaker equation
\begin{equation}\label{forsure13}
\bigg(-\partial_x^2+\Big(m^2-\frac14\Big)\frac{1}{x^2} - \frac{\beta}{x}\bigg)f = 0
\end{equation}
for $m$ and $\beta$ in $\C$.
As analyzed in more details in Section \ref{B5}, solutions of \eqref{forsure13} can be constructed from solutions of the Bessel equation. More precisely, for $\beta \neq 0$, let us define the following  function for $x\in \R_+$\;\!:
\begin{equation*}
j_{\beta,m}(x):= \frac{\Gamma(1+2m)}{\sqrt\pi}\beta^{-\frac14-m}x^{1/4}\cJ_{2m}\big(2\sqrt{\beta x}\big) ,
\end{equation*}
where $\cJ_m$ is defined in Section \ref{sec:standard_bessel}. For $\beta=0$ we set
\begin{equation*}
j_{0,m}(x):=x^{m+\frac{1}{2}}.
\end{equation*}
Then, the equation \eqref{forsure13} is solved by the functions $j_{\beta,m}$, see \cite[Sec.~2.8]{DR2} and Section \ref{B5}. For $2m\not\in\Z$, $j_{\beta,m}$ and $j_{\beta,-m}$ span the space of solutions of \eqref{forsure13}. They are square integrable around zero if and only if $|\Re(m)|<1$.

We still need to consider the special cases $m\in\big\{-\frac12,0,\frac12\big\}$. In fact, we shall not consider separately $m=-\frac{1}{2}$ because Equation  \eqref{forsure13} with $m=-\frac12$ coincides with the case $m=\frac12$. As companions to $j_{\beta,0}$ and $j_{\beta,\frac12}$ for $\beta \neq 0$ we introduce
\begin{align*}
y_{\beta,0}(x)&:=\beta^{-\frac14} x^{1/4}\Big[ \sqrt{\pi} \cY_0\big(2\sqrt{\beta x}\big)
-\frac{(\ln(\beta)+2\gamma)}{\sqrt\pi}  \cJ_0\big(2\sqrt{\beta x}\big)\Big], \\
y_{\beta,\frac12}(x)&:=\beta^{\frac14} x^{1/4}\Big[-\sqrt{\pi} \cY_1\big(2\sqrt{\beta x}\big)
+\frac{(\ln(\beta)+2\gamma-1)}{\sqrt\pi} \cJ_1\big(2\sqrt{\beta x}\big)\Big],
\end{align*}
where $\gamma$ is Euler's constant and $\cY_m$ is defined in Section \ref{sec:standard_bessel}. For $\beta =0$ we set
\begin{equation*}
y_{0,0}(x):=x^\frac{1}{2}\ln(x)
\quad \hbox{and} \quad
y_{0,\frac12}(x) := 1.
\end{equation*}
Then $j_{\beta,0}, y_{\beta,0}$ and $j_{\beta,\frac12}, y_{\beta,\frac12}$ span the space of solutions of \eqref{forsure13} for $m=0$ and for $m=\frac12$ respectively. Indeed, a short computation leads to
\begin{equation*}
\Wr(j_{\beta,0},y_{\beta,0};x)=1
\quad \hbox{and}\quad
\Wr(j_{\beta,\frac12},y_{\beta,\frac12};x)=-1.
\end{equation*}
Since the solutions $j_{\beta,0}, y_{\beta,0}$ and  $j_{\beta,\frac12}, y_{\beta,\frac12}$ are also square integrable around zero, for any $m\in \C$ with $|\Re(m)|<1$ the space of solutions of $L_{\beta,\alpha}f=0$ is 2-dimensional.

Let us describe the asymptotics of these solutions near zero. The following results can be computed based on the expressions provided in the appendix of \cite{DR1}. For any $m\in \C$ with $-2m\not \in \N^\times$ one has
\begin{equation}\label{eq_asymp1}
j_{\beta,m}(x)
=x^{\frac{1}{2}+m}
\Big(1-\frac{\beta}{1+2m} x + O\big(x^2\big) \Big).
\end{equation}
In the exceptional cases one has
\begin{align*}
j_{\beta,0}(x)&
=x^{\frac{1}{2}} \big(1-\beta x\big) + O\big(x^{\frac{5}{2}}\big),\\
j_{\beta,\frac12}(x)&
=x \Big(1-\frac{\beta}{2} x \Big) + O\big(x^3\big),
\end{align*}
together with
\begin{align*}
y_{\beta,0}(x) &
= x^{\frac{1}{2}}\ln(x)\big(1 -\beta x\big) +2\beta x^{\frac32}+ O\big(x^{\frac{5}{2}}|\ln(x)|\big),
\\
y_{\beta,\frac{1}{2}}(x)&
= 1 -\beta x\ln(x) + O\big(x^2|\ln(x)|\big).
\end{align*}

\subsection{Maximal and minimal operators}\label{sec_max_min}

For any $\alpha$ and $\beta\in \C$ we consider the differential expression
\begin{equation*}
L_{\beta,\alpha } :=-\partial_x^2+\Big(\alpha -\frac14\Big)\frac{1}{x^2}-\frac{\beta}{x}
\end{equation*}
acting on distributions on $\bbR_+$. The corresponding maximal and minimal operators in $L^2(\R_+)$ are denoted by $L_{\beta,\alpha }^{\max}$ and $L_{\beta,\alpha }^{\min}$, see \cite[Sec.~3.2]{DR2}
for the details. The domain of $L_{\beta,\alpha }^{\max}$ is given by
\begin{equation*}
\Dom(L_{\beta,\alpha }^{\max}) = \Big\{f\in L^2(\R_+) \mid L_{\beta,\alpha } f\in L^2(\R_+)\Big\} ,
\end{equation*}
while $L_{\beta,\alpha}^{\min}$ is
the closure of the restriction of $L_{\beta,\alpha }$ to
$C_{\rm c}^\infty\big(]0,\infty[\big)$, the set of smooth functions with compact supports in $\R_+$.
The operators $L_{\beta,\alpha }^{\min}$ and $L_{\beta,\alpha }^{\max}$ are closed and we have
\begin{equation*}
\big(L_{\beta,\alpha }^{\min}\big)^* = L_{\bar\beta,\bar \alpha }^{\max}\quad \hbox{ and } \quad
\big(L_{\beta,\alpha }^{\min}\big)^\# = L_{\beta, \alpha }^{\max}.
\end{equation*}

We say that $f\in \Dom(L_{\beta,\alpha }^{\min})$ around $0$, (or, by an abuse of notation, $f(x)\in \Dom(L_{\beta,\alpha }^{\min})$ around $0$) if there exists $\zeta\in C_{\rm c}^\infty\big([0,\infty[\big)$ with $\zeta=1$ around $0$ such that $f\zeta\in \Dom(L_{\beta,\alpha }^{\min})$. The following result follows from the theory of one-dimensional Schr\"odinger operators.

\begin{proposition}\label{lem_properties}
Let $\alpha, \beta, m \in \C$.
\begin{enumerate}
\item[(i)] If $f\in \Dom(L_{\beta,\alpha}^{\max})$, then $f$ and $f'$ are continuous functions on $\R_+$ and converge to $0$ at infinity.
\item[(ii)] If $f\in \Dom(L_{\beta,\alpha}^{\min})$, then near $0$ one has:
\begin{enumerate}
\item $f(x) = o\big(x^{\frac{3}{2}}|\ln(x)|\big)$ and $f'(x)=o\big(x^{\frac{1}{2}}|\ln(x)|\big)$ if $\alpha=0$,
\item $f(x)=o\big(x^{\frac{3}{2}}\big)$ and $f'(x)=o\big(x^{\frac{1}{2}}\big)$ if $\alpha \neq0$.
\end{enumerate}
\item[(iii.a)] If $|\Re(m)|<1$ with $m\not \in \big\{-\frac12,0,\frac12\big\}$, then for any $f\in \Dom(L_{\beta,m^2 }^{\max})$
there exists a unique pair $a,b\in \bbC$
such that
\begin{equation*}
f -a\;\!j_{\beta,m} - b\;\!j_{\beta,-m} \in\Dom(L_{\beta,m^2 }^{\min})\hbox{
around }0.
\end{equation*}
\item[(iii.b)] If $f\in \Dom(L_{\beta,0 }^{\max})$, then there exists a unique pair
$a,b\in \bbC$ such that
\begin{equation*}
f -a\;\!j_{\beta,0} - b\;\!y_{\beta,0} \in\Dom(L_{\beta,0 }^{\min})\hbox{
around }0.
\end{equation*}
\item[(iii.c)] If $f\in \Dom(L_{\beta,\frac14 }^{\max})$, then there exists
a unique pair $a,b\in \bbC$ such that
\begin{equation*}
f -a\;\!j_{\beta,\frac12} - b\;\!y_{\beta,\frac12} \in\Dom(L_{\beta,\frac14 }^{\min})\hbox{
around }0.
\end{equation*}
\item [(iv)] If $|\Re(m)| < 1$, then
\begin{align*}
\Dom ( L^{\min}_{\beta,m^2} ) &= \Big \{ f \in \Dom ( L^{\max}_{\beta,m^2} ) \mid \Wr( f , g ; 0 ) = 0 \text{ for all } g \in \Dom ( L^{\max}_{\beta,m^2} ) \Big \} \\
&= \Big \{ f \in \Dom ( L^{\max}_{\beta,m^2} ) \mid f(x) = o\big( x^{\frac12+|\Re(m)|} \big) \text{ near } 0 \Big \}.
\end{align*}
\item[(v)] If $|\Re(m)|\geqslant 1$, then $\Dom(L_{\beta,m^2  }^{\min})=\Dom(L_{\beta,m^2 }^{\max})$.
\end{enumerate}
\end{proposition}

\begin{proof}
The statements $(i)$--$(iii)$ and $(v)$ are a reformulation of \cite[Prop.~3.1]{DR2} with the current notations. Only $(iv)$ requires elaboration.
The first equality in $(iv)$ follows from \cite[Thm.~3.4]{DG2019}, given that $\Wr( f , g ; \infty ) = 0$ for all $f,g \in \Dom ( L^{\max}_{\beta,m^2} )$ by $(i)$.

The inclusion $\Dom ( L^{\min}_{\beta,m^2} ) \subset \big\{ f \in \Dom ( L^{\max}_{\beta,m^2} ) \mid f(x) = o\big( x^{\frac12+|\Re(m)|} \big) \text{ near } 0 \big\}$ is a consequence of $(ii)$. To prove the converse inclusion, let $f \in \Dom ( L^{\max}_{\beta,m^2} )$. Assuming for instance that $m \notin \big\{ -\frac{1}{2} , 0 , \frac{1}{2}\big \}$ and applying $(iii.a)$, one can write
\begin{equation*}
f \zeta = a j_{\beta,m} \zeta + b j_{\beta,-m} \zeta + f_{\min} ,
\end{equation*}
for some $\zeta \in C_{\rm c}^\infty\big([0,\infty[\big)$ such that $\zeta=1$ around $0$, $a, b \in \C$ and $f_{\min} \in \Dom ( L^{\min}_{\beta,m^2} )$. From \eqref{eq_asymp1} and $(ii)$, we deduce that if $f(x) = o\big( x^{\frac12+|\Re(m)|}\big)$ near $0$ then, necessarily, $a=b=0$. Hence we have proved that $\big\{ f \in \Dom ( L^{\max}_{\beta,m^2} ) \mid f(x) = o\big( x^{\frac12+|\Re(m)|}\big) \text{ near } 0 \big\}  \subset \Dom ( L^{\min}_{\beta,m^2} ) $ in the case where $m \notin \big\{ -\frac{1}{2} , 0 , \frac{1}{2} \big\}$. The same argument applies if $m=\pm\frac12$ or $m=0$, using $(iii.b)$ or $(iii.c)$ instead of $(iii.a)$.
\end{proof}

\subsection{Families of Whittaker operators}

We can now provide the definition of three  families of Whittaker operators. The first family covers the generic case. The Whittaker operator  $H_{\beta, m,\kappa}$ is defined for any $\beta\in \C$, for any $m\in \C$ with $|\Re(m)|<1$ and $m\not \in \big\{-\frac12,0,\frac12\big\}$, and for any $\kappa\in \C\cup \{\infty\}$:
\begin{align*}
\Dom(H_{\beta, m,\kappa}) & = \Big\{f\in \Dom(L_{\beta,m^2}^{\max})\mid
\hbox{ for some }  c \in \C,\\
&\qquad f- c\big(j_{\beta,m} + \kappa\;\!j_{\beta,-m}
\big)\in\Dom(L_{\beta,m^2}^{\min})\hbox{ around }
0\Big\},\qquad\kappa\neq\infty, \nonumber \\
\Dom(H_{\beta, m,\infty}) & = \Big\{f\in \Dom(L_{\beta,m^2}^{\max})\mid
\hbox{ for some }  c \in \C,\\
&\qquad f- c\;\!j_{\beta,-m}
\in\Dom(L_{\beta,m^2}^{\min})\hbox{ around } 0\Big\}.\nonumber
\end{align*}
The second family corresponds to $m=0$:
\begin{align*}
\Dom(H_{\beta,0}^{\nu}) & = \Big\{f\in \Dom(L_{\beta,0}^{\max})\mid
\hbox{ for some }  c \in \C,\\
&\qquad f- c\big(y_{\beta,0} + \nu\;\!j_{\beta,0}
\big)\in\Dom(L_{\beta,0}^{\min})\hbox{ around }
0\Big\},\qquad\nu\in\C, \nonumber \\
\Dom(H_{\beta, 0}^{\infty}) & = \Big\{f\in \Dom(L_{\beta,0}^{\max})\mid
\hbox{ for some }  c \in \C,\\
&\qquad f- c\;\!j_{\beta,0}
\in\Dom(L_{\beta,0}^{\min})\hbox{ around } 0\Big\}.\nonumber
\end{align*}
Finally, in the special case $m=\frac12$ we have the third family:
\begin{align*}
\Dom(H_{\beta,\frac12}^{\nu}) & = \Big\{f\in \Dom(L_{\beta,\frac14}^{\max})\mid
\hbox{ for some }  c \in \C,\\
&\qquad f- c\big(y_{\beta,\frac12} + \nu\;\!j_{\beta,\frac12}
\big)\in\Dom(L_{\beta,\frac14}^{\min})\hbox{ around }
0\Big\},\qquad\nu\in \C, \nonumber \\
\Dom(H_{\beta, \frac12}^{\infty}) & = \Big\{f\in \Dom(L_{\beta,\frac14}^{\max})\mid
\hbox{ for some }  c \in \C,\\
&\qquad f- c\;\!j_{\beta,\frac12}
\in\Dom(L_{\beta,\frac14}^{\min})\hbox{ around } 0\Big\}.\nonumber
\end{align*}

\begin{remark}
Observe that the above boundary conditions could be described with the help of simpler functions. For example, in the above definitions we can replace
\begin{align*}
j_{\beta,m}(x)&\quad\text{with}\quad x^{\frac12+m}\Big(1-\frac{\beta}{1+2m}x\Big)&\text{ if } -1<\Re(m)\leq 0,\\
j_{\beta,m}(x)&\quad\text{with}\quad x^{\frac12+m}&\text{ if }
\quad 0<\Re(m)<1,\\
y_{\beta,0}(x)&\quad\text{with}\quad
x^{\frac{1}{2}}\ln(x)(1-\beta x)+2\beta x^{\frac32},\\
y_{\beta,\frac{1}{2}}(x)&\quad\text{with}\quad
1-\beta x\ln(x).& 
\end{align*}
Note that this can be seen directly, without passing through Bessel functions.
We describe this approach below, and refer to \cite{DG2019} for the general theory.

The idea is to look for elements of $\cD(L_{\beta,m^2}^{\max})$ with a nontrivial behavior near $0$. First we consider the general case and observe that
\begin{align}
L_{\beta,m^2}x^{\frac12+m}&=-\beta x^{-\frac12+m},\label{tra1}\\
L_{\beta,m^2}
x^{\frac12+m}\Big(1-\frac{\beta}{1+2m}x\Big)
&=\frac{\beta^2}{1+2m}x^{\frac12+m}.\label{tra2}
\end{align}
Clearly, the function in the r.h.s.~of \eqref{tra1} is in $L^2$ near $0$ for $\Re(m)>0$
but not for $\Re(m)\leq 0$. On the other hand,  the r.h.s.~of \eqref{tra2} is in
$L^2$ near $0$ for $\Re(m)>-1$. Thus, for $m\neq\pm\frac12$, we obtain two elements
of the boundary space $\cD(L_{\beta,m^2}^{\max})/\cD(L_{\beta,m^2}^{\min})$.
For $m\neq 0$ these elements are linearly independent since
\begin{align*}
&  \lim_{x\searrow0}   \Wr\bigg(\Big(  x^{\frac12+m}\Big(1-\frac{\beta}{1+2m}x\Big),
x^{\frac12-m}\Big(1-\frac{\beta}{1-2m}x\Big);x\bigg)\\
&= \lim_{x\searrow0}   \Wr\bigg(\Big(  x^{\frac12+m},
x^{\frac12-m}\Big(1-\frac{\beta}{1-2m}x\Big);x\bigg) \\
& =-2m.
\end{align*}

It remains to find a second element of $\cD(L_{\beta,m^2}^{\max})$
when $m=0$ or when $m=\frac12$ (as already mentioned we disregard $m=-\frac12$).
Firstly, we try to find the simplest possible elements of
$\cD(L_{\beta,0}^{\max})$ with
a logarithmic behavior near $0$.  We add more and more terms:
\begin{align}
L_{\beta,0}\ln(x)x^{\frac12}&=-\beta x^{-\frac12}\ln(x),\label{tra3}\\
L_{\beta,0}\ln(x)x^{\frac12}(1-\beta x)&=2\beta x^{-\frac12}+\beta^2 x^{\frac12}\ln(x),\label{tra4}\\
L_{\beta,0}\big(\ln(x)x^{\frac12}(1-\beta x)+2\beta x^{\frac32}\big)&=\beta^2 x^{\frac12}(\ln(x)-2).\label{tra5}
\end{align}
For $\beta\neq0$, the r.h.s.~of \eqref{tra3} and of \eqref{tra4}
are not in $L^2$ near $0$. However the r.h.s.~of \eqref{tra5} is in $L^2$ near $0$.
We have thus obtained two
elements of $\cD(L_{\beta,0}^{\max})/\cD(L_{\beta,0}^{\min})$ which are
linearly independent since
\begin{equation*}
\lim_{x\searrow0}  \Wr\Big(x^{\frac12}(1-\beta x),\big(\ln(x)x^{\frac12}(1-\beta x)+2\beta x^{\frac32}\big); x\Big)=1.
\end{equation*}

Finally, let us look for the simplest possible
elements of $\cD(L_{\beta,\frac14}^{\max})$ with a
logarithmic behavior near $0$:
\begin{align}
L_{\beta,\frac14}1&=-\beta x^{-1},\label{tra6}\\
L_{\beta,\frac14}\big(1-\beta x\ln(x)\big)&=\beta^2\ln(x).\label{tra7}
\end{align}
For $\beta\neq0$, the r.h.s.~of \eqref{tra6} is not in $L^2$ near $0$, but
the r.h.s.~of \eqref{tra7} is in $L^2$ near $0$.
We have thus obtained two
elements of $\cD(L_{\beta,\frac14}^{\max})/\cD(L_{\beta,\frac14}^{\min})$
which are linearly independent since
\begin{equation*}
\lim_{x\searrow0}   \Wr\Big(x,\big(1-\beta x\ln(x)\big); x\Big)=-1.
\end{equation*}
\end{remark}

The three families $H_{\beta,m,\kappa}$, $H_{\beta,\frac12}^\nu$ and $H_{\beta,0}^\nu$ cover all possible well-posed extensions of $L_{\beta,m^2}$ with $|\Re(m)|<1$. As already mentioned, we do not introduce a special family for $m=-\frac{1}{2}$, since it is covered by the family corresponding to $m=\frac{1}{2}$. For convenience, we also extend the definition of the first family to the exceptional cases by setting for $\beta\in\C$ and any $\kappa\in\C\cup\{\infty\}$
\begin{equation*}
H_{\beta,-\frac12,\kappa}:=H_{\beta,\frac12}^\infty,
\quad H_{\beta,0,\kappa}:=H_{\beta,0}^\infty,
\quad \hbox{and} \quad H_{\beta,\frac12,\kappa}:=H_{\beta,\frac12}^\infty.
\end{equation*}

An invariance property follows directly from the definition:

\begin{proposition}\label{propo1}
For any $\beta \in \C$,  $|\Re(m)|<1$
and  $\kappa\in \C\cup \{\infty\}$
the following relation holds
\begin{equation*}
H_{\beta, m, \kappa}=H_{\beta, -m,\kappa^{-1}}.
\end{equation*}
\end{proposition}

It is also convenient to introduce another two-parameter family of operators, which cover only special boundary conditions, which we call {\em pure}:
\begin{equation}\label{eq_pure}
H_{\beta,m}:=  H_{\beta,m,0}=H_{\beta,-m,\infty}.
\end{equation}
With this notation, for any $\beta\in\C$, one has
\begin{equation*}
H_{\beta,-\frac12}=H_{\beta,\frac12}^\infty,
\quad H_{\beta,0}=H_{\beta,0}^\infty,
\quad \hbox{and} \quad H_{\beta,\frac12}=H_{\beta,\frac12}^\infty.
\end{equation*}

\begin{remark}
The family  $H_{\beta,m}$ is essentially identical to the family denoted by the same symbol  introduced and studied in \cite{DR2}. The only difference with that reference is that the operator corresponding to $(\beta,m)=\big(0,-\frac12\big)$ was left
undefined in \cite{DR2}. This point
corresponds to a singularity, neverthelss in the current paper we have decided to set $H_{0,-\frac12}:=H_{0,\frac12}$.
\end{remark}

Here is a comparison of the above families with the families $H_{m,\kappa}$, $H_0^\nu$ introduced in \cite{DR1} when $\beta=0$. In the first column we put one of the newly introduced family, in the second column we put the families from \cite{DR1,DR2}.
\begin{align*}
H_{0,m,\kappa}&=H_{m,\kappa} &|\Re(m)|<1, \ m\not \in \big\{-\tfrac12,\tfrac12\big\},\quad \kappa\in\C\cup\{\infty\},\\
H_{0,0}^\nu&=H_0^\nu
& \nu\in\C\cup\{\infty\},\\
H_{0,\frac12}^\nu&=H_{-\frac12,\nu}=H_{\frac12,\frac1\nu}
& \nu\in\C\cup\{\infty\}.
\end{align*}
For completeness, let us also mention two special operators which are included in these families (for clarity, the indices are emphasized). The Dirichlet Laplacian on $\R_+$ is given by
\begin{equation*}
H_{\beta=0,m=-\frac12} =  H_{\beta=0,m=\frac12} = H_{0,\frac12}^\infty
= H_{m=\frac12,\kappa=0} = H_{m=-\frac12,\kappa=\infty}
\end{equation*}
while the Neumann Laplacian is given by
\begin{equation*}
H_{\beta=0,m=\frac12}^0 = H_{m=-\frac12,\kappa=0} = H_{m=\frac12,\kappa=\infty}.
\end{equation*}
Note that the former operator was also described in \cite{DR1} by $H_{m=\frac12}$ while the latter operator was described by $H_{m=-\frac12}$.

We now gather some easy properties of the operators $H_{\beta,m,\kappa}$.

\begin{proposition}\label{propo2}
For $m\in \C$ with $|\Re(m)|<1$ one has
\begin{align*}
\big(H_{\beta,m,\kappa}\big)^*=H_{\bar\beta,\bar m,\bar\kappa}&\qquad
\big(H_{\beta,m,\kappa}\big)^\#=H_{\beta,m,\kappa}
&
\kappa\in\C\cup\{\infty\}, \\
\big(H_{\beta,0}^\nu\big)^*=H_{\bar\beta,0}^{\bar\nu}&\qquad
\big(H_{\beta,0}^\nu\big)^\#=H_{\beta,0}^{\nu},
& \nu\in\C\cup\{\infty\},\\
\big(H_{\beta,\frac12}^\nu\big)^*=H_{\bar\beta,\frac12}^{\bar\nu}&\qquad
\big(H_{\beta,\frac12}^\nu\big)^\#=H_{\beta,\frac12}^{\nu}
& \nu\in\C\cup\{\infty\}.
\end{align*}
\end{proposition}

\begin{proof} Let us prove the first statement, the other ones can be obtained similarly.
Recall from Proposition \ref{lem_properties} (see also \cite[Prop.~A.2]{BDG}) that for any $f\in \Dom(L_{\beta,m^2}^{\max})$ and $g\in \Dom(L_{\bar \beta,\bar m^2}^{\max})$, the functions $f,f',g,g'$ are continuous on $\R_+$. In addition, the Wronskian of $\bar f$ and $g$, as introduced in \eqref{wronskian}, possesses a limit at zero, and we have the equality
\begin{equation*}
(L_{\beta,m^2}^{\max}f|g) - (f|L_{\bar \beta,\bar m^2}^{\max}g) = -\Wr(\bar f,g;0).
\end{equation*}
In particular, if $f\in \Dom(H_{\beta,m,\kappa})$ one infers that
\begin{equation*}
(H_{\beta,m,\kappa}f|g) = (f|L_{\bar \beta, \bar m^2}^{\max}g) -\Wr(\bar f,g;0).
\end{equation*}
Thus, $g\in \Dom\big((H_{\beta,m,\kappa})^*\big)$ if and only if $\Wr(\bar f,g;0)=0$, and then $(H_{\beta,m,\kappa})^*g=L_{\bar \beta,\bar m^2}^{\max}g$. By taking into account the explicit description of $\Dom(H_{\beta,m,\kappa})$, straightforward computations show that $\Wr(\bar f,g;0)=0$ if and only if $g\in \Dom(H_{\bar \beta,\bar m,\bar \kappa})$. One then deduces that $(H_{\beta,m,\kappa})^*= H_{\bar \beta,\bar m,\bar \kappa}$. The property for the transpose of $H_{\beta,m,\kappa}$ can be proved similarly.
\end{proof}

By combining Propositions \ref{propo1} and \ref{propo2} one easily deduces the following characterization of  self-adjoint operators contained in our families:

\begin{corollary}\label{corol_s_a}
The operator $H_{\beta,m,\kappa}$ is self-adjoint if and only if one of the following sets of conditions is satisfied:
\begin{enumerate}
\item[(i)] $\beta\in \R$, $m\in ]-1,1[$ and $\kappa\in \R\cup\{\infty\}$,
\item[(ii)] $\beta \in \R$, $m\in \i \R^\times$ and $|\kappa|=1$.
\end{enumerate}
The operators  $H_{\beta,0}^\nu$ and $H_{\beta,\frac{1}{2}}^\nu$ are self-adjoint if and only if $\beta \in \R$ and $\nu \in \R\cup \{\infty\}$.
\end{corollary}

Let us finally mention some equalities about the action of the dilation group. For that purpose, we recall that the unitary group $\{U_\tau\}_{\tau\in\R}$ of dilations acts on $f\in L^2(\R_+)$ as $\big(U_\tau f\big)(x) = \e^{\tau/2}f(\e^\tau x)$. The proof of the following lemma consists in an easy computation.

\begin{proposition}
For $m\in \C$ with $|\Re(m)|<1$ one has
\begin{align*}
U_\tau H_{\beta,m,\kappa}U_{-\tau}
&=\e^{-2\tau}H_{\e^\tau\beta,m,\e^{ -2\tau m}\kappa}
&
\kappa\in\C\cup\{\infty\},\\
U_\tau H_{\beta,0}^\nu U_{-\tau}& =\e^{-2\tau}H_{\e^\tau\beta,0}^{\nu+\tau}
& \nu\in\C\cup\{\infty\},\\
U_\tau H_{\beta,\frac12}^\nu U_{-\tau}
&=\e^{-2\tau}H_{\e^\tau\beta,\frac12}^{\e^\tau(\nu-\beta\tau)}
& \nu\in\C\cup\{\infty\}.
\end{align*}
with the conventions $\alpha \;\!\infty=\infty$ for any $\alpha\in \C\setminus\{0\}$ and  $\infty + \tau = \infty$.
\end{proposition}

%%%%%%%%%%%%%%%%%%%%%%%%%%%%%%%%%%%%%%%%%%%%%%%%%%%%%%%%%%%%%%%%%%%%%%%
\section{Spectral theory}\label{sec_spectral}
\setcounter{equation}{0}
%%%%%%%%%%%%%%%%%%%%%%%%%%%%%%%%%%%%%%%%%%%%%%%%%%%%%%%%%%%%%%%%%%%%%%%%

In this section we investigate the spectral properties of the Whittaker operators.

\subsection{Point spectrum}

The point spectrum is obtained by looking at general solutions of the equation
\begin{equation*}
L_{\beta,m^2}f= -k^2 f
\end{equation*}
for $k\in \C$ with $\Re(k)\geq 0$, and by considering only the solutions
which are in the domain of the operators $H_{\beta,m,\kappa}$, $H^\nu_{\beta,\frac{1}{2}}$, or $H^\nu_{\beta,0}$.

In the following statement,
$\Gamma$ stands for the usual gamma function, $\psi$ is the digamma function defined by $\psi(z) = \Gamma'(z) / \Gamma( z )$ and $\gamma = - \psi( 1 )$. Since the special case $\beta=0$ has already been considered in \cite{DR1}, we assume that $\beta\neq 0$ in the following statement, and recall in Theorem \ref{thm_beta_0} the results obtained for $\beta=0$. It is also useful to note that the condition $\beta\not\in[0,\infty[$ guarantees that either $+\Im(\sqrt{\beta})>0$ or $-\Im(\sqrt{\beta})>0$, due to our definition of the square root.

\begin{theorem}\label{spectrum}
\begin{enumerate}
\item
Let $\beta \in \C^\times$, $|\Re(m)|<1$ with $m\not \in \big\{-\frac12,0,\frac12\big\}$, and let $\kappa\in\C\cup\{\infty\}$.
Then the operator $H_{\beta,m,\kappa}$ possesses an eigenvalue $\lambda\in\C$ in the following cases:
\begin{enumerate}
\item[(i)] $\lambda=-k^2$, $\Re(k)>0$, $\frac{\beta}{2k} + m - \frac12 \notin \N$ and
\begin{equation}\label{eq_Ci1}
\kappa = (2k)^{-2m}\frac{\Gamma(2m)}{\Gamma(-2m)}
\frac{\Gamma\big(\frac{1}{2}-m - \frac{\beta}{2k}\big)}{\Gamma\big(\frac{1}{2}+m- \frac{\beta}{2k}\big)},
\end{equation}
\item[(ii)] $\lambda=\mu^2$, $0<\mu<\pm\Im(\beta)$ and
\begin{equation*}
\kappa =  \e^{\pm \i\pi m} (2\mu)^{-2m}\frac{\Gamma(2m)}{\Gamma(-2m)}
\frac{\Gamma\big(\frac{1}{2}-m\mp \i\frac{\beta}{2\mu}\big)}{\Gamma\big(\frac{1}{2}+m\mp \i\frac{\beta}{2\mu}\big)},
\end{equation*}
\item[(iii)] $\lambda=0$, $\beta \not \in [0,\infty[$, and
\begin{equation*}
\kappa = \frac{\Gamma(2m)}{\Gamma(-2m)\;\!(-\beta)^{2m}}.
\end{equation*}
\end{enumerate}
\item Let $\beta \in \C^\times$ and $\nu\in\C\cup\{\infty\}$. Then $H_{\beta,\frac12}^\nu$ possesses an eigenvalue $\lambda$ in the following cases:
\begin{enumerate}
\item[(i)] $\lambda=-k^2$, $\Re(k)>0$, $\frac{\beta}{2k} \notin \N$ and
\begin{equation*}
\nu=-\beta\left(\frac12\psi\Big(1-\frac{\beta}{2k}\Big)+\frac12\psi\Big(-\frac{\beta}{2k}\Big)+2\gamma-1+\ln(2k) \right),
\end{equation*}
\item[(ii)] $\lambda=\mu^2$, $0<\mu<\pm\Im(\beta)$, and
\begin{equation*}
\nu=-\beta\left(\frac12\psi\Big(1\mp \i\frac{\beta}{2\mu}\Big) + \frac12\psi\Big(\mp \i\frac{\beta}{2\mu}\Big)+2\gamma-1 +\ln(2\mu) \mp \i \frac{\pi}{2} \right),
\end{equation*}
\item[(iii)] $\lambda=0$, $\pm \Im(\sqrt{\beta})>0$, and
\begin{equation*}
\nu =  - \beta \big(\ln( \beta ) + 2\gamma - 1 \mp \i \pi \big).
\end{equation*}
\end{enumerate}
\item Let $\beta \in \C^\times$ and $\nu\in\C\cup\{\infty\}$. Then $H_{\beta,0}^\nu$ possesses an eigenvalue $\lambda$ in the following cases:
\begin{enumerate}
\item[(i)] $\lambda=-k^2$, $\Re(k)>0$, $\frac{\beta}{2k} - \frac12 \notin \N$ and
\begin{equation*}
\nu=\psi\Big(\frac{1}{2}-\frac{\beta}{2k}\Big)+2\gamma+\ln(2k),
\end{equation*}
\item[(ii)] $\lambda=\mu^2$, $0<\mu<\pm\Im(\beta)$, and
\begin{equation*}
\nu=\psi\Big(\frac{1}{2}\mp \i\frac{\beta}{2\mu}\Big) \mp \i\frac{\pi}{2}+2\gamma+ \ln(2\mu),
\end{equation*}
\item[(iii)] $\lambda=0$, $\pm \Im(\sqrt{\beta})>0$, and
\begin{equation*}
\nu = \ln( \beta ) + 2\gamma + 2\ln(2) \mp \i \pi.
\end{equation*}
\end{enumerate}
\end{enumerate}
\end{theorem}

\begin{proof}
We start with the special case $\lambda=-k^2 = 0$. The two solutions of the equation $L_{\beta,m^2}f=0$ are provided by the functions
\begin{equation}\label{hankel}
x\mapsto h_{\beta,m}^\pm(x)
:=x^{1/4}\Ha^{\pm}_{2m}\big(2\sqrt{\beta x}\big),
\end{equation}
with $\Ha^\pm_m$ the Hankel function for dimension 1, see \cite[App.~A.5]{DR1}. We then infer from \cite[App.~A.5]{DR1} that for any $z$ with $-\pi<\arg(z)\leq \pi$, one has as $z\to 0$
\begin{equation*}
\Ha_m^\pm(z)=\left\{ \begin{array}{lcl}
\pm \i \frac{\sqrt{2}}{\sqrt{\pi}} z^{\frac{1}{2}} \big(\ln(z)+\gamma\mp \i \frac{\pi}{2}\big) + O\big(|z|^{\frac52}\ln(|z|)\big)
& {\rm if} & m=0,\\
\mp \i \frac{1}{\sqrt{\pi}}\big(\frac{z}{2}\big)^{-\frac{1}{2}} \pm \i  \frac{2}{\sqrt{\pi}}\Big(\ln\big(\frac{z}{2}\big)+\gamma-\frac{1}{2}\mp \i \frac{\pi}{2}\Big)\big(\frac{z}{2}\big)^{\frac{3}{2}}
+ O\big(|z|^{\frac72}\ln(|z|)\big)
& {\rm if} & m=1,\\
\mp \i \frac{\sqrt{\pi}}{\sin(\pi m)} \left(\frac{z}{2}\right)^{\frac12}
\Big(\frac{1}{\Gamma(1-m)}\big(\frac{z}{2}\big)^{-m}-\frac{\e^{\mp \i \pi m}}{\Gamma(1+m)}\big(\frac{z}{2}\big)^m\Big) + O(|z|^{\frac{5}{2} - | \Re( m ) |})
& \rm if & m \not \in \Z.
\end{array} \right.
\end{equation*}
For $|\Re(m)|<1$, this implies that the two functions $ h_{\beta,m}^\pm$ belong to $L^2(\R_+)$ near $0$. On the other hand, for large $z$ and $|\arg(\mp\i z)|<\pi-\varepsilon$, $\varepsilon>0$, one has
\begin{equation*}
\Ha_m^\pm(z) = \e^{\pm\i (z-\frac{1}{2}\pi m-\frac{1}{4}\pi)}\big(1+O(|z|^{-1})\big).
\end{equation*}
Since $| \arg( 2\sqrt{\beta x} ) | \leq  \pi / 2$,
it follows that
\begin{equation*}
h_{\beta,m}^\pm(x) = x^{1/4}\e^{\pm\i (2\sqrt{\beta x}- \pi m-\frac{1}{4}\pi)}\big(1+O(|x|^{-\frac{1}{2}})\big),
\end{equation*}
Hence  if $\Im(\sqrt{\beta})=0$, then $h_{\beta,m}^\pm$ do not belong to $L^2$ near infinity, while
if $\pm\Im(\sqrt{\beta})>0$, then
$h^\pm_{\beta,m}$ belongs to $L^2$ near infinity, and
$h^\mp_{\beta,m}$ does not.
For $\pm \Im(\sqrt{\beta})>0$, we thus have that $h_{\beta,m}^\pm \in L^2( \R_+ )$ and hence, since in addition $L_{\beta,m^2} h_{\beta,m}^\pm=0$, we deduce that $h_{\beta,m}^\pm \in \Dom(L_{\beta,m^2 }^{\max})$.
It only remains to check in which domain of the operators $H_{\beta,m,\kappa}$, $H^\nu_{\beta,\frac{1}{2}}$, or $H^\nu_{\beta,0}$ does $h_{\beta,m}^\pm$ belong to.
By Proposition \ref{lem_properties}, it suffices to determine the asymptotic expansion near $0$ of $h_{\beta,m}^\pm$ up to remainder terms of order $o(x^{\frac12+|\Re(m)]})$.
This can easily be obtained from the expansion provided above, and yields to the statements $1.(iii)$, $2.(iii)$ and $3.(iii)$.

Let us now prove the statements $1.(ii)$, $2.(ii)$ and $3.(ii)$. We consider the equation $L_{\beta,m^2}f=\mu^2 f$ for some $\mu>0$. Two linearly independent solutions are provided by the functions $x\mapsto \Hpm{\frac{\beta}{2\mu}}{m}(2\mu x)$ introduced in  \cite[Sec.~2.7]{DR2}, see also \eqref{Hpm-definition}. From the asymptotic expansion near infinity given by
\begin{equation}\label{eq_near_infty}
\Hpm{\frac{\beta}{2\mu}}{m}(2\mu x) = \e^{\mp\i\frac{\pi}{2}\naw{\frac{1}{2}+m}}\e^{\frac{\pi\beta}{4\mu}}(2\mu x)^{\pm\i\frac{\beta}{2\mu}}
\;\!\e^{\pm\i\mu x}\big(1 + O(x^{-1})\big) ,
\end{equation}
one infers that at most one of these functions is in $L^2$ near infinity, depending on the sign of $\Im(\beta)$. More precisely, for $\Im(\beta)>0$, the map $x\mapsto \Ha^+_{\frac{\beta}{2\mu},m}(2\mu x)$ belongs to $L^2$ near infinity if $\mu<\Im(\beta)$ and does not belong to $L^2$ near infinity otherwise. Under the same condition $\Im(\beta)>0$, the map $x\mapsto \Ha^-_{\frac{\beta}{2\mu},m}(2\mu x)$ never belongs to $L^2$ near infinity. Conversely, for $\Im(\beta)<0$, the map $x\mapsto \Ha^-_{\frac{\beta}{2\mu},m}(2\mu x)$ belongs to $L^2$ near infinity if $\mu<-\Im(\beta)$ and does not belong to $L^2$ near infinity otherwise. Under the same condition $\Im(\beta)<0$, the map $x\mapsto \Ha^+_{\frac{\beta}{2\mu},m}(2\mu x)$ never belongs to $L^2$ near infinity. Finally, for $\Im(\beta)=0$, none of these functions belongs to $L^2$ near infinity.

For the asymptotic expansion near $0$, the information on $\Ha^\pm_{\delta,m}$ provided in \cite[Eq.~(2.31)]{DR2} is not sufficient. However, the appendix of the current paper contains all the necessary information on these special functions. By taking into account the Taylor expansion of $\cI_{\delta,m}$ near $0$ provided in \eqref{Taylor_2} and the equality $\Gamma(\alpha)\Gamma(1-\alpha)=\frac{\pi}{\sin(\pi\alpha)}$ one infers that for $|\Re(m)|<1$ and $m\not \in \big\{-\frac{1}{2},0,\frac{1}{2}\big\}$ one has
\begin{equation}\label{eq_a2}
\cI_{\delta,m}(z)=\frac{z^{\frac{1}{2}+m}}{\Gamma(1+2m)} \Big(1 -\frac{\delta}{1+2m}z+ O(z^2)\Big)
\end{equation}
and
\begin{align*}
\Hpm{\delta}{m}(z)  = &  \mp \i \e^{\mp i\pi m} \frac{\Gamma(-2m)}{\Gamma\big(\frac{1}{2}-m\mp \i\delta\big)}z^{\frac{1}{2}+m}
\Big(1-\frac{\delta}{1+2m}z\Big) \\
& \quad \mp \i \frac{\Gamma(2m)}{\Gamma\big(\frac{1}{2}+m\mp \i\delta\big)}z^{\frac{1}{2}-m}
\Big(1-\frac{\delta}{1-2m}z\Big) + o\big(z^{\frac{3}{2}}\big).
\end{align*}

For $2m \in \Z$ one has to consider the expression for $\Ka_{\delta,\frac{1}{2}}$ and $\Ka_{\delta,0}$ provided in \eqref{eq_a12} and \eqref{eq_a0} respectively. Then, by considering the Taylor expansion near $0$ of these functions one gets
\begin{align}\notag
\cK_{\delta, \frac{1}{2}}(z)
=& \frac{1}{\Gamma(1-\delta)}+\frac{1}{\Gamma(-\delta)}z\ln(z)\\
&+\frac{1}{\Gamma(-\delta)} \Big(\frac12\psi(1-\delta)+\frac12\psi(-\delta)+2\gamma-1\Big)z
+o\big(z^{\frac{3}{2}}\big),\label{eq_Taylor_12}\\
\label{eq_Taylor_0}
\nonumber \cK_{\delta,0}(z) = &  -\frac{1}{\Gamma\big(\frac{1}{2}-\delta\big)}\Big[
z^\frac{1}{2} \ln(z) + \Big(\psi\Big(\frac{1}{2}-\delta\Big)+2\gamma\Big)z^\frac{1}{2}
-\delta z^\frac{3}{2}\ln(z) \\
& \quad-\delta\Big(
\psi\Big(\frac{1}{2}-\delta\Big)+2\gamma-2
\Big)z^{\frac{3}{2}}\Big]
+o\big(z^{\frac{3}{2}}\big).
\end{align}
From Equation \eqref{Hpm-definition} one finally deduces the relations
\begin{align*}
\Hpm{\delta}{\frac{1}{2}}(z)
= & \mp \i\frac{1}{\Gamma(1\mp \i\delta)} - \frac{1}{\Gamma(\mp \i\delta)}z\ln(z)  \\
& \quad -\frac{1}{\Gamma(\mp \i\delta)} \Big(\frac12\psi(1\mp \i\delta)+\frac12\psi(\mp \i\delta)+2\gamma-1\mp \i \frac{\pi}{2}\Big)z
+o\big(z^{\frac{3}{2}}\big)\\
\Hpm{\delta}{0}(z)
= &\pm \i\frac{1}{\Gamma\big(\frac{1}{2}\mp \i\delta\big)}\Big[
z^\frac{1}{2} \ln(z) + \Big(\psi\Big(\frac{1}{2}\mp \i\delta\Big) \mp \i\frac{\pi}{2}+2\gamma\Big)z^\frac{1}{2}
-\delta z^\frac{3}{2}\ln(z) \Big]+O\big(z^{\frac{3}{2}}\big).
\end{align*}

To show $1.(ii)$ we consider the function $x\mapsto \Ha_{\frac{\beta}{2\mu},m}^+(2\mu x)$ if $\Im(\beta)>0$ and $x\mapsto \Ha_{\frac{\beta}{2\mu},m}^-(2\mu x)$ if $\Im(\beta)<0$, and check for which $\kappa$ these functions belong to $\Dom(H_{\beta,m \kappa})$. For $|\Re(m)|<1$ and $m\not \in \big\{-\frac{1}{2},0,\frac{1}{2}\big\}$ one has
\begin{align*}
\Ha_{\frac{\beta}{2\mu},m}^\pm(2\mu x)  = & \mp \i \e^{\mp \i\pi m} \frac{\Gamma(-2m)}{\Gamma\big(\frac{1}{2}-m\mp \i\frac{\beta}{2\mu}\big)}
(2\mu x)^{\frac{1}{2}+m}
\Big(1-\frac{\beta}{1+2m}x\Big) \\
& \quad \mp \i \frac{\Gamma(2m)}{\Gamma\big(\frac{1}{2}+m\mp \i\frac{\beta}{2\mu}\big)}(2\mu x)^{\frac{1}{2}-m}
\Big(1-\frac{\beta}{1-2m}x\Big) + o\big(x^{\frac{3}{2}}\big) \\
= & \mp \i c\big(j_{\beta,m} + \kappa j_{\beta,-m}(x) \big) + o\big(x^{\frac{3}{2}}\big)
\end{align*}
with $c:=  \e^{\mp \i\pi m} \frac{\Gamma(-2m)}{\Gamma(\frac{1}{2}-m\mp \i\frac{\beta}{2\mu})} (2\mu)^{\frac{1}{2}+m}$
and
\begin{equation*}
\kappa:=\frac{1}{c} \frac{\Gamma(2m)}{\Gamma\big(\frac{1}{2}+m\mp \i\frac{\beta}{2\mu}\big)}(2\mu)^{\frac{1}{2}-m}
=  \e^{\pm \i\pi m} (2\mu)^{-2m}\frac{\Gamma(2m)}{\Gamma(-2m)}
\frac{\Gamma\big(\frac{1}{2}-m\mp \i\frac{\beta}{2\mu}\big)}{\Gamma\big(\frac{1}{2}+m\mp \i\frac{\beta}{2\mu}\big)}.
\end{equation*}
Note that the conditions $\pm \Im(\beta)>0$, $|\Re(m)|<1$, and $\mu<\pm\Im(\beta)$ imply that $\pm \i\frac{\beta}{2\mu}+m-\frac{1}{2}\not \in \N$.

The proof of $2.(ii)$ and $3.(ii)$ can be obtained similarly once the following expressions are taken into account:
\begin{align*}
\Hpm{\frac{\beta}{2\mu}}{\frac{1}{2}}(2\mu x)
=& \frac{2\mu}{\beta} \frac{1}{\Gamma\big(\mp \i\frac{\beta}{2\mu}\big)} \big(1-\beta x\ln(x)\big)  \\
&\hspace{-9ex} -\frac{2\mu}{\Gamma\big(\mp \i\frac{\beta}{2\mu}\big)} \Big[\frac12\psi\Big(1\mp \i\frac{\beta}{2\mu}\Big)+\frac12\psi\Big(\mp \i\frac{\beta}{2\mu}\Big)+2\gamma-1 +\ln(2\mu) \mp \i \frac{\pi}{2} \Big]
x +o\big(x^{\frac{3}{2}}\big),\\
\Hpm{\frac{\beta}{2\mu}}{0}(2\mu x)
= &\pm \i\frac{(2\mu)^{\frac{1}{2}}}{\Gamma\big(\frac{1}{2}\mp \i\frac{\beta}{2\mu}\big)}\Big(
x^\frac{1}{2} \ln(x)  \\
& + \Big[\psi\Big(\frac{1}{2}\mp \i\frac{\beta}{2\mu}\Big) \mp \i\frac{\pi}{2}+2\gamma+ \ln(2\mu)\Big)x^\frac{1}{2}
-\beta x^\frac{3}{2}\ln(x) \Big]+O\big(x^{\frac{3}{2}}\big).
\end{align*}

We shall now turn to the generic case (statements $1.(i)$, $2.(i)$ and $3.(i)$), namely the equation $L_{\beta,m^2}f=-k^2 f$ for some $k\in \C$ with $\Re(k)>0$. In the non-degenerate case, solutions of this equation are provided by the functions
\begin{equation}\label{eq_3sol}
x \mapsto \Ka_{\frac{\beta}{2k},m}(2kx)
\qquad \hbox{and} \qquad
x \mapsto \Ia_{\frac{\beta}{2k},\pm m}(2kx).
\end{equation}
We refer again to the appendix for an introduction to these functions. The behavior for large $z$ of the function $\cK_{\delta,m}(z)$ has been provided in \eqref{Kbm-around-infinity}, from which one infers that the first function in \eqref{eq_3sol} is always in $L^2$ near infinity. On the other hand, since for $|\arg(z)|<\frac{\pi}{2}$ one has
\begin{equation*}
\cI_{\delta,\pm m}(z) = \frac{1}{\Gamma\big(\frac{1}{2}\pm m-\delta\big)}z^{-\delta}\;\!\e^{\frac{z}{2}}\big(1+O(z^{-1})\big)
\end{equation*}
it follows that the remaining two functions in \eqref{eq_3sol} do not belong to $L^2$ near infinity as long as $\frac{\beta}{2k} \mp m -\frac{1}{2}\not \in \N$. Still in the non-degenerate case and when the condition $\frac{\beta}{2k} + m -\frac{1}{2} \in \N$ holds, it follows from relation \eqref{Wr} that the functions $\Ka_{\frac{\beta}{2k},m}(2k\cdot)$ and $\Ia_{\frac{\beta}{2k},-m}(2k\cdot)$ are linearly dependent, but still $\Ia_{\frac{\beta}{2k},m}(2k\cdot)$ does not belong to $L^2$ near infinity. Similarly, when $\frac{\beta}{2k} - m -\frac{1}{2} \in \N$ it is the function $\Ia_{\frac{\beta}{2k},-m}(2k\cdot)$ which does not belong to $L^2$ near infinity.

Let us now turn to the degenerate case, when $m\in \big\{-\frac{1}{2},0,\frac{1}{2}\big\}$. In this situation the two functions $\cI_{\delta, m}$ and $\cI_{\delta,-m}$ are no longer independent, as a consequence of \eqref{wro1}. In the non-doubly degenerate case (see the appendix for more details), which means for $\big(\frac{\beta}{2k},m\big)\not \in \big(\Z,\pm\frac{1}{2}\big)$ or for $\big(\frac{\beta}{2k},m\big)\not \in \big(\Z+\frac{1}{2},0\big)$, the above arguments can be mimicked, and one gets that only the function $\Ka_{\frac{\beta}{2k},m}(2k\cdot)$ belongs to $L^2$ near infinity. In the doubly degenerate case, the function $\cX_{\delta,m}$, introduced in \eqref{eq_def_X}, has to be used. This function is independent of the function $\cK_{\delta,m}$, as shown in \eqref{eq_Wr_2d}. However, this function explodes exponentially near infinity, which means that $\cX_{\frac{\beta}{2k},m}(2k\cdot)$ does not belong to $L^2$ near infinity. Once again, only the function $\Ka_{\frac{\beta}{2k},m}(2k\cdot)$ plays a role.

As a consequence of these observations, it will be sufficient to concentrate on the function $\Ka_{\frac{\beta}{2k},m}(2k\cdot)$ and to check for which $\kappa$ or $\nu$ does this function belong to the domain of the operators $H_{\beta,m,\kappa}$, $H^\nu_{\beta,\frac{1}{2}}$, or $H^\nu_{\beta,0}$ respectively. For the behavior of this function near $0$ one infers from \eqref{generic} and \eqref{eq_a2} that for $m\not \in \big\{-\frac{1}{2},0,\frac{1}{2}\big\}$
\begin{align*}
\cK_{\frac{\beta}{2k},m}(2kx)
= &  -\frac{\pi}{\sin(2\pi m)}\left(\frac{\cI_{\frac{\beta}{2k},m}(2kx)}{\Gamma\big(\frac{1}{2}-m-\frac{\beta}{2k}\big)}
- \frac{\I{\frac{\beta}{2k}}{-m}(2kx)}{\Gamma\big(\frac{1}{2}+m-\frac{\beta}{2k}\big)}\right) \\
= & (2k)^{\frac{1}{2}+m}\frac{\Gamma(-2m)}{\Gamma\big(\frac{1}{2}-m-\frac{\beta}{2k}\big)} x^{\frac{1}{2}+m}
\Big(1 -\frac{\beta}{1+2m}x\Big) \\
& \quad + (2k)^{\frac{1}{2}-m}\frac{\Gamma(2m)}{\Gamma\big(\frac{1}{2}+m-\frac{\beta}{2k}\big)} x^{\frac{1}{2}-m}
\Big(1 -\frac{\beta}{1-2m}x\Big) + o(x^\frac{3}{2}).
\end{align*}
Similarly, it follows from \eqref{eq_a12} and  \eqref{eq_a0} that
\begin{align}\notag
 \cK_{\frac{\beta}{2k},\frac{1}{2}}(2kx)
= & -\frac{1}{\beta}\frac{2k}{\Gamma\big(-\frac{\beta}{2k}\big)} \big(1-\beta x\ln(x)\big)
  +\frac{2k}{\Gamma\big(-\frac{\beta}{2k}\big)} \Big[\frac12\psi\Big(1-\frac{\beta}{2k}\Big)+\frac12\psi\Big(-\frac{\beta}{2k}\Big)\\&+2\gamma-1+\ln(2k)\Big]x
+o\big(x^{\frac{3}{2}}\big),\label{eq_K1/2}\\
\nonumber \cK_{\frac{\beta}{2k},0}(2kx)=&  -\frac{(2kx)^{\frac{1}{2}}}{\Gamma\big(\frac{1}{2}-\frac{\beta}{2k}\big)}\Big[
(1-\beta x)\ln(x) + \Big(\psi\Big(\frac{1}{2}-\frac{\beta}{2k}\Big)+2\gamma+\ln(2k)\Big)  \\
&\quad -\beta\Big(
\psi\Big(\frac{1}{2}-\frac{\beta}{2k}\Big)+2\gamma-2 +\ln(2k)\Big)x\Big]
+o\big(x^{\frac{3}{2}}\big).\label{eq_K0}
\end{align}

The statements $1.(i)$, $2.(i)$, and $3.(i)$ follow then straightforwardly.
\end{proof}

\begin{remark}
A special feature of positive eigenvalues described in Theorem \ref{spectrum} is that the corresponding eigenfunctions have an inverse polynomial decay at infinity, and not an exponential decay at infinity, as it is often expected. This property can be directly inferred from the asymptotic expansion provided in \eqref{eq_near_infty}.
\end{remark}

\begin{remark}
Self-adjoint operators that are included in the families $H_{\beta,m,\kappa}$, $H_{\beta,\frac12}^\nu$ and $H_{\beta,0}^\nu$ do not have eigenvalues in $]0,\infty[$. Indeed, in Theorem \ref{spectrum} a necessary condition for the existence of strictly positive eigenvalues is that $\Im(\beta)\neq0$. This automatically prevents these operators to be self-adjoint, as a consequence of Corollary \ref{corol_s_a}.
\end{remark}

For completeness let us recall the results already obtained in \cite[Sec.~5]{DR1} for $\beta=0$.

\begin{theorem}\label{thm_beta_0}
\begin{enumerate}
\item[(i)]  If $|\Re(m)|<1$, $m\not \in \big\{-\frac{1}{2},0,\frac{1}{2}\big\}$ and $\kappa\in\C\cup\{\infty\}$, the eigenvalues of the operator $H_{0,m,\kappa}$
are of the form $-k^2$ with $\Re(k)>0$, where
\begin{equation*}
\kappa=\Big(\frac{k}{2}\Big)^{-2m}\frac{\Gamma(m)}{\Gamma(-m)},
\end{equation*}
\item[(ii)] If $\nu\in\C\cup\{\infty\}$, the eigenvalues of the operator $H_{0,\frac{1}{2}}^\nu$ are of the form $-k^2$ with $\Re(k)>0$, where $\nu=-k$,
\item[(iii)] If $\nu\in\C\cup\{\infty\}$, the eigenvalues of the operator $H_{0,0}^\nu$ are of the form $-k^2$ with $\Re(k)>0$, where
\begin{equation*}
\nu=\gamma+\ln\Big(\frac{k}{2}\Big).
\end{equation*}
\end{enumerate}
\end{theorem}

\begin{remark}
Note that Theorem \ref{thm_beta_0} can be derived from Theorem \ref{spectrum}. Indeed, for $m\not \in \big\{-\frac12,0,\frac12\big\}$ we infer from the Legendre duplication formula
\begin{equation*}
\Gamma(z)\Gamma\Big(\frac{1}{2}+z\Big)=2^{1-2z}\sqrt{\pi}\;\!\Gamma(2z),
\end{equation*}
that
\begin{equation*}
(2k)^{-2m}\frac{\Gamma(2m)}{\Gamma(-2m)}
\frac{\Gamma\big(\frac{1}{2}-m - \frac{\beta}{2k}\big)}{\Gamma\big(\frac{1}{2}+m- \frac{\beta}{2k}\big)}\Big|_{\beta=0}
=\Big(\frac{k}{2}\Big)^{-2m}\frac{\Gamma(m)}{\Gamma(-m)}.
\end{equation*}
For $m=\frac{1}{2}$, we first note that $\Gamma\big(\frac{1}{2}\big)=\sqrt{\pi}$ and $\Gamma\big(-\frac{1}{2}\big)=-2\sqrt{\pi}$. Then we use the relations $\psi(1+z)=\psi(z)+\frac{1}{z}$ and $\psi(1)=-\gamma$, and infer that
\begin{equation*}
\lim_{\beta\to0}
-\beta\Big(\frac12\psi\Big(1-\frac{\beta}{2k}\Big)+\frac12\psi\Big(-\frac{\beta}{2k}\Big)+2\gamma-1+\ln(2k) \Big)
=\Big(\frac{k}{2}\Big) \frac{\Gamma\big(-\frac12\big)}{\Gamma\big(\frac12\big)}
=-k.
\end{equation*}
Finally for $m=0$, from the equality $\psi\big(\frac{1}{2}\big)=-2\ln(2)-\gamma$ one gets
\begin{equation*}
\psi\Big(\frac{1}{2}-\frac{\beta}{2k}\Big)+2\gamma+\ln(2k)
\Big|_{\beta=0}
=\gamma+\ln\Big(\frac{k}{2}\Big).
\end{equation*}
\end{remark}

As a consequence of the expressions provided in Theorem \ref{spectrum}, the discreteness of the spectra of all operators can be inferred in $\C \setminus [0,\infty[$.

\subsection{Green's functions}

Let us now turn our attention to the continuous spectrum. We shall first look for an expression for  Green's function. We will  use the well-known theory of 1-dimensional Schr\"odinger operators, as presented for example in the appendix of \cite{BDG} or in \cite{DG2019}. We begin by recalling a result on which we shall rely.

Let $AC(\R_+)$ denote the set of absolutely continuous functions from $\R_+$ to $\C$, that is functions whose distributional derivative belongs to $L^1_\loc(\R_+)$. Let also $AC^1(\R_+)$ be the set of functions from $\R_+$ to $\C$ whose distributional derivatives belong to $AC(\R_+)$. If $V \in L^1_\loc(\R_+)$, it is not difficult to check that the operator $-\partial_x^2 + V$ can be interpreted as a linear map from $AC^1(\R_+)$ to $L^1_\loc(\R_+)$. The maximal operator associated to $-\partial_x^2 + V$ is then defined as
\begin{align*}
&\Dom( L^\max) := \Big \{ f \in L^2(\R_+) \, \cap \, AC^1(\R_+)
\mid \big( -\partial_x^2 + V \big ) f \in L^2(\R_+) \Big \} \\
& L^\max f := \big ( -\partial_x^2 + V \big ) f , \quad f \in \Dom( L^\max ).
\end{align*}
The minimal operator $L^{\min}$ is the closure of $L^\max$ restricted to compactly supported functions. Note that $L^\max=(L^\min)^\#$.

As before, we say that a function $f : \R_+ \to \C$ belongs to $L^2$ around $0$ (respectively around $\infty$) if there exists $\zeta \in C_{\rm c}^\infty\big([0,\infty[\big)$ with $\zeta = 1$ around $0$ such that $f \zeta \in L^2(\R_+)$ (respectively $f(1-\zeta) \in L^2(\R_+)$).

The following statement contains several results proved in \cite{DG2019}.

\begin{proposition}\label{red}
Let $V\in L_\loc^1(\R_+)$. Let $k \in \C$ and suppose that $u(k,\cdot),v(k,\cdot)\in AC^1(\R_+)$ solve
\begin{align*}
\big(-\partial_x^2+V\big)u(k,\cdot)&=-k^2u(k,\cdot),\\
\big(-\partial_x^2+V\big)v(k,\cdot)&=-k^2v(k,\cdot).
\end{align*}
Assume that $u( k , \cdot)$, $v( k , \cdot )$ are linearly independent and that $u( k , \cdot ) \in L^2$ around $0$, $v ( k , \cdot ) \in L^2$ around $\infty$. Let $\Wr(k):=\Wr\big(u(k,\cdot),v(k,\cdot) ; x\big)$ be the Wronskian of these two solutions. Set
\begin{equation*}
R(-k^2 ; x,y):= \frac{1}{\Wr(k)}
\left \{\begin{array}{ll}
u(k,x)\;\!v(k,y) & \text{ for }0<x<y,\\
u(k,y)\;\!v(k,x) & \text{ for }0<y<x,
\end{array}
\right.
\end{equation*}
and assume that $R(-k^2 ; x ,y)$ is the integral kernel of a bounded operator $R(-k^2)$. Then there exists a unique closed realization $H$ of $-\partial_x^2+V$ with the boundary condition at $0$ given by $u(k,\cdot)$ and at $\infty$ given by $v(k,\cdot)$ in the sense that
\begin{align*}
  &  \Dom( H ) = \big \{ f \in \Dom ( L^\max ) , \, f - u( k , \cdot ) \in \Dom(  L^\min) \hbox{ around } 0 \big \} , \\
&  \phantom{ \Dom( H ) } = \big \{ f \in \Dom ( L^\max) , \, f - v( k , \cdot ) \in \Dom( L^{\min}) \hbox{ around } \infty \big \} , \\
& H f = \big ( -\partial_x^2+V \big ) f,  \quad f \in \Dom( H ).
\end{align*}
Moreover $- k^2$ belongs to the resolvent set of $H$ and
$R(-k^2)=(H+k^2)^{-1}$.
\end{proposition}

By using such a statement, it has been proved in \cite{DR2} that, for $k \in \C$ such that $\Re( k ) > 0$ and $\frac{\beta}{2k} - \frac{1}{2} - m \notin \N$, we have that $- k^2 \notin \sigma( H_{\beta,m} )$ and $R_{\beta,m}(-k^2):=(k^2+H_{\beta,m})^{-1}$ has the integral kernel
\begin{align*}
\nonumber &R_{\beta,m}(-k^2;x,y)\\
& = \tfrac{1}{2k}\Gamma\naw{\tfrac{1}{2}+m-\tfrac{\beta}{2k}} \begin{cases} \I{\frac{\beta}{2k}}{m}(2k x)\K{\frac{\beta}{2k}}{m}(2k y) & \mbox{ for }0<x<y,\\
\I{\frac{\beta}{2k}}{m}(2k y)\K{\frac{\beta}{2k}}{m}(2k x) & \mbox{ for }0<y<x .
\end{cases}
\end{align*}

Let us now describe the integral kernel of the resolvent of all operators under investigation.  We recall that our parameters are $\beta\in \C$, $\kappa \in \C\cup \{\infty\}$, $\nu \in \C\cup \{\infty\}$, and $m\in \C$ satisfying $-1<\Re(m)<1$.

\begin{theorem}\label{thm_hard}
Let $k \in \C$  with $\Re( k ) > 0$. We have the following properties.
\begin{enumerate}[label=(\roman*)]
\item For  $\kappa\neq\infty$ and $m\not \in \big\{-\frac{1}{2},0,\frac{1}{2}\big\}$ set
\begin{align}
\nonumber\gamma_{\beta,m}(k)
&:=\frac{(2k)^{-m}}{\Gamma\big(\frac12+m-\frac{\beta}{2k}\big)\Gamma(1-2m)},\\
\label{def_omegabm} \omega_{\beta,m,\kappa}(k) &:=
\frac{\gamma_{\beta,m}(k)+\kappa \gamma_{\beta, -m}(k)}{\kappa \gamma_{\beta,-m}(k)}
\frac{\pi}{\sin(2\pi m)}.
\end{align}
If $\gamma_{\beta,m}(k)+\kappa \gamma_{\beta, -m}(k)\neq0$,
then $-k^2\not\in\sigma(H_{\beta,m,\kappa})$ and the integral kernel of $R_{\beta,m,\kappa}( - k^2 ) := ( H_{\beta,m,\kappa} + k^2 )^{-1}$ is given by
\begin{align}\label{res-gen}
\nonumber & R_{\beta,m,\kappa}(-k^2 ; x,y) \\
\nonumber & =\frac{1}{\gamma_{\beta,m}(k)+\kappa \gamma_{\beta, -m}(k)}
\Big(\gamma_{\beta,m}(k)R_{\beta,m}(-k^2 ; x,y)
+\kappa \gamma_{\beta,-m}(k)R_{\beta,-m}(-k^2 ; x,y)\Big) \\
&= R_{\beta,m}(-k^2;x,y)
+ \frac{\Gamma\big(\frac12+m-\frac{\beta}{2k}\big)\Gamma\big(\frac12-m-\frac{\beta}{2k}\big)}
{2k\;\!\omega_{\beta,m,\kappa}(k)}
\cK_{\frac{\beta}{2k},m}(2k y)  \cK_{\frac{\beta}{2k},m}(2k x).
\end{align}

If $\kappa = \infty$ and $\frac{\beta}{2k}+m-\frac{1}{2}\not \in \N$, then
$-k^2\not\in\sigma(H_{\beta,m,\infty})$ and $R_{\beta,m,\infty}(-k^2)=R_{\beta,-m}(-k^2)$.

\item  For   $\nu \neq \infty$, $m = \frac{1}{2}$ and $\frac{\beta}{2k}\not \in \N^\times$ set
\begin{equation*}
\omega_{\beta,\frac12}^\nu(k):= -\frac12\psi\Big(1- \frac{\beta}{2k} \Big)- \frac12\psi\Big(- \frac{\beta}{2k} \Big) -2\gamma - \ln(2k) + 1-\frac{\nu}{\beta}.
\end{equation*}
If $\omega_{\beta,\frac{1}{2}}^\nu( k )  \neq 0$, then $-k^2\not\in\sigma(H_{\beta,\frac{1}{2}}^\nu)$ and the integral kernel of
$R_{\beta,\frac{1}{2}}^\nu(-k^2) := ( H_{\beta,\frac{1}{2}}^\nu + k^2 )^{-1}$ is given by
\begin{align}\label{res-1/2}
\nonumber & R_{\beta,\frac12}^\nu(-k^2;x,y) \\
& = R_{ \beta, \frac12 }( -k^2 ; x , y )
+ \frac{\Gamma\big(-\frac{\beta}{2k}\big) \Gamma\big(1-\frac{\beta}{2k}\big)}{2k\;\!\omega_{\beta,\frac12}^\nu(k)} \cK_{\frac{\beta}{2k},\frac12}(2kx) \cK_{\frac{\beta}{2k},\frac12}(2ky) .
\end{align}
If $\nu = \infty$ and $\frac{\beta}{2k}\not \in \N^\times$, then $-k^2\not\in\sigma(H_{\beta,\frac{1}{2}}^\infty)$ and $R_{\beta,\frac12}^\infty(-k^2)=R_{\beta,\frac12}(-k^2)$.

\item For  $\nu \neq \infty$, $m = 0$ and $\frac{\beta}{2k}-\frac{1}{2}\not \in \N$ set
\begin{equation*}
\omega_{\beta,0}^\nu(k):=\psi\Big(\frac12-\frac{\beta}{2k}\Big)+2\gamma+\ln(2k)-\nu .
\end{equation*}
If $\omega_{\beta,0}^\nu(k) \neq 0$, then $-k^2\not\in\sigma(H_{\beta,0}^\nu)$ and the integral kernel of $R_{\beta,0}^\nu( - k^2 ) :={ ( H_{\beta,0}^\nu + k^2 )^{-1}}$ is given by
\begin{align}\label{res-zero}
\nonumber & R_{\beta,0}^\nu( - k^2 ; x , y ) \\
& = R_{\beta,0}(-k^2;x,y) +\frac{\Gamma\big(\frac12-\frac{\beta}{2k}\big)^2} {2k\;\! \omega_{\beta,0}^\nu( k ) } \cK_{\frac{\beta}{2k},0}(2k x) \cK_{\frac{\beta}{2k},0}(2k y).
\end{align}
If $\nu = \infty$ and  $\frac{\beta}{2k}-\frac{1}{2}\not \in \N$, then $-k^2\notin\sigma(H_{\beta,0}^\infty)$ and
$R_{\beta,0}^\infty(-k^2)=R_{\beta,0}(-k^2)$.
\end{enumerate}
\label{green0}
\end{theorem}

For the proof of this theorem, we shall mainly rely on a similar statement which was proved in \cite[Sec.~3.4]{DR2}. The context was less general, but some of the estimates turn out to be still useful.

\begin{proof}[Proof of Theorem \ref{green0}.]
The proof consists in checking that all conditions of Proposition \ref{red} are satisfied.

For $(i)$ we need to show that the integral kernel $R_{\beta,m,\kappa}(-k^2; x , y)$ defines a bounded operator on $L^2(\R_+)$. This follows from \eqref{res-gen}, because all numerical factors are harmless and because by \cite[Thm.~3.5]{DR2} $R_{\beta,m} ( - k^2 ; x , y )$ and $R_{\beta,-m} ( - k^2 ; x , y )$ are the kernels defining bounded operators.

Moreover, we can write
\begin{align}\label{resolvent_generic}
\nonumber & R_{\beta,m,\kappa}(-k^2 ; x,y)
=\frac{1}{2k\;\! \big(\gamma_{\beta,m}(k)+\kappa \gamma_{\beta, -m}(k)\big)} \\
&\times
\begin{cases}
\Big( \frac{(2k)^{-m}}{\Gamma(1-2m)} \cI_{\frac{\beta}{2k},m}(2k x) + \kappa \frac{(2k)^m}{\Gamma(1+2m)}\cI_{\frac{\beta}{2k},-m}(2k x)\Big) \cK_{\frac{\beta}{2k},m}(2k y) & \mbox{ for }0<x<y,\\[2ex]
\Big(\frac{(2k)^{-m}}{\Gamma(1-2m)} \cI_{\frac{\beta}{2k},m}(2k y) +\kappa\frac{(2k)^{m}}{\Gamma(1+2m)}\cI_{\frac{\beta}{2k},-m}(2k y)\Big) \cK_{\frac{\beta}{2k},m}(2k x) &  \mbox{ for }0<y<x.
\end{cases}
\end{align}
Since $\cK_{\frac{\beta}{2k},m}(2k \cdot )$ belongs to $L^2(\R_+)$, this solution is $L^2$ around $\infty$.  For the other solution, one verifies by \eqref{eq_a2} that
\begin{align*}
& \frac{(2k)^{-m}}{\Gamma(1-2m)} \cI_{\frac{\beta}{2k},m}(2k x) +\kappa \frac{(2k)^m}{\Gamma(1+2m)}\cI_{\frac{\beta}{2k},-m}(2k x) \\
& = \frac{ (2k)^{\frac12} }{ \Gamma( 1 + 2m ) \Gamma( 1 - 2m ) } \Big[ x^{\frac12+m} \Big ( 1 - \frac{ \beta }{ 1 + 2m }x \Big ) + \kappa x^{\frac12 - m} \Big ( 1 - \frac{ \beta }{ 1 - 2m }x \Big ) \Big] \\
& \quad + O \big ( x^{\frac{5}{2}-|\Re(m)|} \big ).
\end{align*}
Therefore, this function belongs to $L^2$ around $0$ and satisfies the same boundary condition at $0$ as $j_{\beta,m,} + \kappa j_{\beta,-m}$. By Proposition \ref{red}, this proves $(i)$ when $\kappa \neq \infty$. Note that in the special case $\kappa = \infty$, it is enough to observe that $H_{\beta,m,\infty} =H_{\beta,-m,0}$ and to apply the previous result.

To prove $(ii)$, consider first $\nu \neq \infty$ and $\frac{\beta}{2k}\not \in \N^\times$. It has been proved in \cite[Thm.~3.5]{DR2} that the first kernel of \eqref{res-1/2} defines a bounded operator. The second kernel corresponds to a constant multiplied by a rank one operator defined by the function $\Ka_{\frac{\beta}{2k},m}(2k \cdot)\in L^2(\R_+)$ and therefore this operator is also bounded. Next we write
\begin{align}\label{resolvent_1/2}
&R_{\beta,\frac12}^\nu(-k^2;x,y) =\frac{\Gamma\big(-\frac{\beta}{2k}\big)\Gamma\big(1-\frac{\beta}{2k}\big)}{2k\;\!\omega_{\beta,\frac12}^\nu(k)}  \\
&\times
\begin{cases}
\Big(\frac{\omega_{\beta,\frac12}^\nu(k)}{\Gamma(-\frac{\beta}{2k})}\cI_{\frac{\beta}{2k},\frac12}(2kx)+\cK_{\frac{\beta}{2k},\frac12}(2kx)\Big)\cK_{\frac{\beta}{2k},\frac12}(2ky) & \mbox{ for }0<x<y,\\[2ex]
\Big(\frac{\omega_{\beta,\frac12}^\nu(k)}{\Gamma(-\frac{\beta}{2k})}\cI_{\frac{\beta}{2k},\frac12}(2ky)+\cK_{\frac{\beta}{2k},\frac12}(2ky)\Big)\cK_{\frac{\beta}{2k},\frac12}(2kx) & \mbox{ for }0<y<x.
\end{cases}\notag
\end{align}
We deduce from \eqref{eq_a2} and \eqref{eq_K1/2} that
\begin{align*}
& \frac{\omega_{\beta,\frac12}^\nu(k)}{\Gamma\big(-\frac{\beta}{2k}\big)}\cI_{\frac{\beta}{2k},\frac12}(2kx)+\cK_{\frac{\beta}{2k},\frac12}(2kx) \\
& = \frac{1}{ \Gamma\big( 1 - \frac{\beta}{2k} \big) } \big ( 1- \beta x \ln(x) + \nu x \big )
+ o\big(x^\frac32\big) ,
\end{align*}
which belongs to $L^2$ around $0$ and corresponds to the boundary condition defining $H_{\beta,\frac{1}{2}}^\nu$.

The proof of $(iii)$ is analogous. We use first \eqref{res-zero} for the boundedness. Then we rewrite Green's function as
\begin{align} \label{resolvent_0}
\nonumber & R_{\beta,0}^\nu(-k^2;x,y) = \frac{\Gamma\big(\frac12-\frac{\beta}{2k}\big)^2}{2k\;\! \omega_{\beta,0}^\nu(k)} \\
&\times
\begin{cases}
\Big(\frac{\omega_{\beta,0}^\nu(k)}{\Gamma(\frac12-\frac{\beta}{2k})}\cI_{\frac{\beta}{2k},0}(2kx)+\cK_{\frac{\beta}{2k},0}(2kx)\Big)\cK_{\frac{\beta}{2k},0}(2ky) & \mbox{ for }0<x<y,\\[2ex]
\Big(\frac{\omega_{\beta,0}^\nu(k)}{\Gamma(\frac12-\frac{\beta}{2k})}\cI_{\frac{\beta}{2k},0}(2ky)+\cK_{\frac{\beta}{2k},0}(2ky)\Big)\cK_{\frac{\beta}{2k},0}(2kx) & \mbox{ for }0<y<x.
\end{cases}
\end{align}
We check that
\begin{align*}
&\frac{\omega_{\beta,0}^\nu(k)}{\Gamma\big(\frac12-\frac{\beta}{2k}\big)}\cI_{\frac{\beta}{2k},0}(2kx)+\cK_{\frac{\beta}{2k},0}(2kx)\\ =& - \frac{ (2k)^{\frac12} }{ \Gamma\big( \frac12 - \frac{\beta}{2k} \big) } \big ( x^{\frac12} (1-\beta x)\ln( x ) + 2\beta x^{\frac32}+\nu x^{\frac12}(1-\beta x) \big ) + O \big( x^{\frac52}|\ln(x)| \big) ,
\end{align*}
by \eqref{eq_a2} and \eqref{eq_K0}, see also \eqref{eq_a0}.
\end{proof}

Strictly speaking, the formulas of Thm \ref{green0} are not valid in doubly degenerate points, when the functions $\cK_{\beta,m}$ and $\cI_{\beta,m}$ are proportional to one another, and the operator $H_{\beta,m}$ has an eigenvalue. To obtain well defined formulas one needs to use the function $\cX_{\beta,m}$
defined in \eqref{eq_def_X}, as described in the following proposition:

\begin{proposition}\label{green1}
Let $k \in \C$  with $\Re( k ) > 0$.
We have the following properties.
\begin{enumerate}[label=(\roman*)]
\item[(ii')]
For $m = \frac{1}{2}$, $\nu \neq \infty$ and $\frac{\beta}{2k} \in \N^\times$, set
\begin{equation*}
\xi_{\beta,\frac12}^\nu(k):= \frac12 \psi \Big ( 1 + \frac{ \beta }{ 2 k } \Big )
+ \frac12 \psi \Big ( \frac{ \beta }{ 2 k } \Big )
+ 2\gamma + \ln(2k) - 1  + \frac{ \nu}{\beta}.
\end{equation*}
Then $-k^2\not\in\sigma(H_{\beta,\frac12}^\nu)$ and the integral kernel of $R_{\beta,\frac12}^\nu( - k^2 )$ is given by
\begin{align*}
& R_{\beta,\frac12}^\nu(-k^2;x,y)\\
= &\frac{1}{2k}
\begin{cases}
\Big((-1)^{\frac{\beta}{2k}}\cX_{\frac{\beta}{2k},\frac12}(2kx)
+ \frac{\xi_{\beta,\frac12}^\nu(k) }
{\Gamma(\frac{\beta}{2k})\Gamma(1+\frac{\beta}{2k})}\cK_{\frac{\beta}{2k},\frac12}(2kx)\Big) \cK_{\frac{\beta}{2k},\frac12}(2ky), & \mbox{ for }0<x<y,\\[2ex]
\Big((-1)^{\frac{\beta}{2k}}\cX_{\frac{\beta}{2k},\frac12}(2ky)
+ \frac{\xi_{\beta,\frac12}^\nu(k)}{\Gamma(\frac{\beta}{2k})\Gamma(1+\frac{\beta}{2k})}  \cK_{\frac{\beta}{2k},\frac12}(2ky)\Big)\cK_{\frac{\beta}{2k},\frac12}(2kx), & \mbox{ for }0<y<x.
\end{cases}\notag
\end{align*}

\item[(iii')]
For $m = 0$, $\nu \neq \infty$, and $\frac{\beta}{2k}-\frac{1}{2} \in \N$, set
\begin{equation*}
\xi_{\beta,0}^\nu(k):=  -\psi \Big ( \frac12 + \frac{ \beta }{ 2 k } \Big )- 2\gamma -
\ln( 2 k ) + \nu .
\end{equation*}
Then $-k^2\not\in\sigma(H_{\beta,0}^\nu)$ and the integral kernel of $R_{\beta,0}^\nu( - k^2 )$ is given by
\begin{align*}
& R_{\beta,0}^\nu(-k^2;x,y)\\ = &\frac{1}{2k}
\begin{cases}\Big((-1)^{\frac{\beta}{2k}+\frac12}\cX_{\frac{\beta}{2k},0}(2kx)
+ \frac{\xi_{\beta,0}^\nu(k) }
{\Gamma(\frac12+\frac{\beta}{2k})^2}\cK_{\frac{\beta}{2k},0}(2kx)\Big) \cK_{\frac{\beta}{2k},0}(2ky), & \mbox{ for }0<x<y,\\[2ex]
\Big((-1)^{\frac{\beta}{2k}+\frac12}\cX_{\frac{\beta}{2k},0}(2ky)
+ \frac{\xi_{\beta,0}^\nu(k)}{\Gamma(\frac12+\frac{\beta}{2k})^2}  \cK_{\frac{\beta}{2k},0}(2ky)\Big)\cK_{\frac{\beta}{2k},0}(2kx), & \mbox{ for }0<y<x.
\end{cases}
\end{align*}
\end{enumerate}
\end{proposition}

\begin{proof} $(ii')$ is proved similarly as $(ii)$ of Theorem \ref{green0}, by using for $m = \frac{1}{2}$, $\nu \neq \infty$ and $\frac{\beta}{2k} \in \N^\times$ that
\begin{align*}
& (-1)^{\frac{\beta}{2k}}\cX_{\frac{\beta}{2k},\frac12}(2kx)
+ \frac{\xi_{\beta,\frac12}^\nu(k) }
{\Gamma(\frac{\beta}{2k})\Gamma(1+\frac{\beta}{2k})}\cK_{\frac{\beta}{2k},\frac12}(2kx)
\\
& = \frac{(-1)^{\frac{\beta}{2k}+1}}{\Gamma\big(1+\frac{\beta}{2k}\big)}
\big( 1 - \beta x \ln x + \nu x + o ( x )\big).
\end{align*}
This follows from \eqref{eq_Wr_2d}, \eqref{exp_X12},  and \eqref{eq_Taylor_12}.

$(iii')$ is proved similarly as $(iii)$ of Theorem \ref{green0}. In particular, using  \eqref{eq_Wr_2d}, \eqref{exp_X0}, and \eqref{eq_Taylor_0} one verifies that
\begin{align*}
& (-1)^{\frac{\beta}{2k}+\frac12}\cX_{\frac{\beta}{2k},0}(2kx)
+ \frac{\xi_{\beta,0}^\nu(k) }
{\Gamma(\frac12+\frac{\beta}{2k})^2}\cK_{\frac{\beta}{2k},0}(2kx) \\
& =(-1)^{\frac{\beta}{2k}-\frac{1}{2}}
\frac{(2k)^{\frac{1}{2}}}{\Gamma\big(\frac{1}{2}+\frac{\beta}{2k}\big)}
x^{\frac{1}{2}}
\big((1-\beta x)\ln(x)+2\beta x + \nu (1-\beta x)\big)
+ o \big( x^{\frac32} \big).
\end{align*}
\end{proof}

\subsection{Holomorphic families of closed operators}

In this section we show that the families of operators introduced before are holomorphic for suitable values of the parameters. A general definition of a {\em holomorphic family of closed operators} can be found in \cite{Kato}, see also \cite{DW}. Actually, we will not need its most general definition. For us it is enough to recall  this concept in the special case where the operators possess a nonempty resolvent set.

Let $\H$ be a complex Banach space. Let $\{ H( \z ) \}_{ \z \in \Theta }$ be a family of closed operators on $\H$ with nonempty resolvent set, where $\Theta$ is an open subset of $\C^d$.  $\{ H( \z ) \}_{ \z \in \Theta }$ is called holomorphic on $\Theta$ if for any $\z_0 \in \Theta$, there exist $\lambda \in \C$ and a neighborhood $\Theta_0 \subset \Theta$ of $\z_0$ such that, for all $\z \in \Theta_0$, $\lambda$ belongs to the resolvent set of $H( \z )$ and the map $\Theta_0 \ni \z \mapsto ( H( \z ) - \lambda )^{-1} \in \B( \H )$ is holomorphic on $\Theta_0$.
Note that if $\Theta_0 \ni \z \mapsto ( H( \z ) - \lambda )^{-1} \in \B( \H )$ is locally bounded on $\Theta_0$ and if there exists a dense subset $\Dom \subset \H$ such that, for all $f , g \in \Dom$, the map $\Theta_0 \ni \z \mapsto ( f | ( H( \z ) - \lambda )^{-1} g )$ is holomorphic on $\Theta_0$, then $\Theta_0 \ni \z \mapsto ( H( \z ) - \lambda )^{-1} \in \B( \H )$ is holomorphic on $\Theta_0$.  Besides, by Hartog's theorem, $\z \mapsto ( f | ( H( \z ) - \lambda )^{-1} g )$ is holomorphic if and only if it is separately analytic in each variable.

This definition naturally generalizes to families of operators defined on $( \C \cup \{\infty\} )^d$ instead of $\C^d$, recalling that a map $\varphi : \C \cup \{\infty\} \to \C$ is called holomorphic in a neighborhood of $\infty$ if the map $\psi : \C \to \C$, defined by $\psi( z ) = \phi ( 1 / z )$ if $z \neq 0$ and $\psi( 0 ) = \phi( \infty )$, is holomorphic in a neighborhood of $0$.

Recall that the family $H_{\beta,m}$ has been defined on $\C\times\{m \in \C \mid \Re(m)>-1\}$ in \cite{DR2}, see also \eqref{eq_pure}. However, it is not holomorphic on the whole domain. The following has been proved in \cite{DR2}.

\begin{theorem}
The family of closed operators
$(\beta,m)\mapsto H_{\beta,m}$ is holomorphic on
\begin{equation*}
\C\times\{m \in \C \mid \Re(m)>-1\}\backslash\big\{\big(0,-\tfrac12\big)\big\}.
\end{equation*}
However, it cannot be extended by continuity
to include the point $\big(0,-\frac12\big)$.
\end{theorem}

Let us sketch what happens at $\big(0,-\frac12\big)$.
Recall that in \cite{BDG,DR1} a holomorphic family
$\big\{m \in \C \mid \Re(m)>-1\big\}\ni m\mapsto H_m$ has been
introduced, and satisfies $H_m=H_{0,m}$ for  $m\neq -\frac12 $.
Note also that for any $\beta$ we have $H_{\beta,-\frac12}=H_{\beta,\frac12}$.
It then turns out that
\begin{equation*}
\lim_{\beta\to 0}H_{\beta,-\frac12}=H_{\frac12}\neq
H_{-\frac12}=\lim_{m\to-\frac12}H_{0,m},
\end{equation*}
where these limits have to be understood as weak resolvent limits.
Note that in the sequel and in particular in \eqref{discountinuity2},  \eqref{discountinuity3}, and \eqref{discountinuity4}, the limits should be understood in such a sense.

Let us consider now the families of operators involving mixed boundary conditions.
To this end, it will be convenient to introduce the notation $$
\Pi:=\{m\in \C\mid -1<\Re(m)<1\}.
$$ Recall that
 $(\beta,m,\kappa)\mapsto \{H_{\beta,m,\kappa}\}$
has been defined
on $\C\times\Pi\times(\C\cup\{\infty\})$. However, it is not holomorphic on this whole set:

\begin{theorem}\label{thm:holomorphy}
\begin{enumerate}[label=(\roman*)]
\item  The family of closed operators
$\{H_{\beta,m,\kappa}\}$ is holomorphic on
$\C\times\Pi\times \big(\C\cup\{\infty\}\big)$ except for
\begin{equation}\label{except}
\big(0,-\tfrac12\big){\times}\big(\C\cup\{\infty\}\big)\,\cup\,
\big(0,\tfrac12\big){\times}\big(\C\cup\{\infty\}\big)\,\cup\,
\C\times(0,-1).
\end{equation}
\item The family of closed operators
$\{H_{\beta,0}^\nu\}$ is holomorphic on
$\C \times \big(\C\cup\{\infty\}\big)$.
\item The family of closed operators
$\big\{H_{\beta,\frac{1}{2}}^\nu\big\}$ is holomorphic on $\C \times  \big(\C\cup\{\infty\}\big)$.
\end{enumerate}
\end{theorem}

\begin{proof}
For shortness, let us set
\begin{equation}\label{def_eta}
\eta_{\beta,m,\kappa}(k)
:=\gamma_{\beta,m}(k)+\kappa \gamma_{\beta, -m}(k).
\end{equation}
This expression appears in the numerator of \eqref{def_omegabm}
and plays an important role in the expression \eqref{resolvent_generic} for the resolvent of
$H_{\beta,m,\kappa}$.

$(i)$ Let $(\beta_0,m_0,\kappa_0)$ belong to the domain $\C\times\Pi\times \big(\C\cup\{\infty\}\big)$.
First assume that $m_0 \notin \big\{ -\frac{1}{2} , 0 , \frac{1}{2} \big\}$ and that $\kappa_0 \in \C$.
Let $k \in \C$ with $\Re( k ) > 0$ such that
$\eta_{\beta_0,m_0,\kappa_0}(k)\neq0$, where $\eta_{\beta,m,\kappa}(k)$
is defined in \eqref{def_eta}.
By continuity of the map $( \beta , m , \kappa ) \mapsto \eta_{\beta,m,\kappa}(k)$, there exists a neighborhood $\U_0$ of $( \beta_0 , m_0 , \kappa_0 )$ such that for all $( \beta , m , \kappa )$ in this neighborhood, we have $\eta_{\beta,m,\kappa}( k ) \neq 0$. Hence, by Theorem \ref{green0}, we infer that $-k^2 \notin \sigma ( H_{\beta,m,\kappa} )$, and the resolvent $( H_{\beta,m,\kappa} + k^2 )^{-1} \in \B \big( L^2 ( \R_+ ) \big)$ is the operator whose kernel is given by \eqref{resolvent_generic}. It then easily follows from the analyticity properties of the maps $( \beta , m , \kappa ) \mapsto \cI_{\frac{\beta}{2k},\pm m}( 2kx )$ and $( \beta , m , \kappa ) \mapsto \cK_{\frac{\beta}{2k},m}( 2kx )$
(for fixed $x > 0$ and $k$) that, for all $f , g \in L^2(\R_+)$, the map $( \beta , m , \kappa ) \mapsto ( f | ( H_{\beta,m,\kappa} + k^2 )^{-1} g )$ is holomorphic on $\U_0$. Hence $\{H_{\beta,m,\kappa}\}$ is holomorphic on $\U_0$.

If $m_0 \notin  \big\{ -\frac{1}{2} , 0 , \frac{1}{2} \big\}$ and $\kappa_0 = \infty$, the statement directly follows from the equality $H_{\beta,m,\infty} = H_{\beta,-m,0}$.

Suppose now that $m_0 = 0$ and that $\kappa_0 \in \C \setminus \{ - 1 \}$.
We extend by continuity the definition of $\eta_{\beta,m,\kappa}(k)$ in \eqref{def_eta} for $m=0$ by setting
\begin{equation*}
\eta_{\beta,0,\kappa}( k ) := \frac{ 1 + \kappa }{ \Gamma\big( \frac12 - \frac{ \beta }{ 2 k } \big) }.
\end{equation*}
We also choose $k \in \C$ with $\Re( k ) > 0$ such that $\frac{\beta_0}{2k} -\frac{1}{2} \not \in \N$. This latter requirement implies that $\eta_{\beta_0,m_0,\kappa_0}( k ) \neq 0$, and by continuity of the map $( \beta , m , \kappa ) \mapsto \eta_{\beta,m,\kappa}( k ) $,  there exists a neighborhood $\U_0$ of $( \beta_0 , 0 , \kappa_0 )$ such that for all $( \beta , m , \kappa )$ in this neighborhood,  $\eta_{\beta,m,\kappa}( k ) \neq 0$. In particular, by Theorem \ref{green0}, one verifies that, for all $f , g \in L^2(\R_+)$, the map $( \beta , m , \kappa ) \mapsto ( f | ( H_{\beta,m,\kappa} + k^2 )^{-1} g )$ is well-defined and holomorphic on $\U_0$ provided that \eqref{resolvent_generic} is extended to $\U_0 \cap \{ (\beta,0,\kappa) \mid \beta \in \C , \kappa \in \C \}$ by
\begin{align*}
&R_{\beta,0,\kappa}(-k^2 ; x,y) =\frac{ \Gamma\big( \frac12 - \frac{ \beta }{ 2 k } \big) }{2k }
\begin{cases}
 \cI_{\frac{\beta}{2k},0}(2k x) \cK_{\frac{\beta}{2k},0}(2k y) & \mbox{ for }0<x<y,\\[2ex]
 \cI_{\frac{\beta}{2k},0}(2k y) \cK_{\frac{\beta}{2k},0}(2k x) &  \mbox{ for }0<y<x.
\end{cases}
\end{align*}
Note that this corresponds to the integral kernel of $(H_{\beta,0,0} + k^2 )^{-1}
 = (H_{\beta,0}^\infty+k^2)^{-1}$. This shows that $\{H_{\beta,m,\kappa}\}$ is holomorphic on $\U_0$ (provided that $\U_0$ is chosen small enough so that $(\beta,0,-1)\not \in \U_0$).

If $m_0 = 0$ and $\kappa_0 = \infty$, the argument is similar once it is observed that
\begin{equation*}
 (H_{\beta,0,\infty} + k^2 )^{-1}
= (H_{\beta,0}^\infty+k^2)^{-1}
= (H_{\beta,0,0} + k^2 )^{-1}.
\end{equation*}

It remains to consider the cases $m_0 = \pm \frac{1}{2}$ and $\beta_0 \neq 0$. Assume for instance that $m_0 = - \frac{1}{2}$, $\beta_0 \neq 0$, and $\kappa_0 \in \C$. We extend by continuity the definition of $\eta_{\beta,m,\kappa}(k)$ in \eqref{def_eta} for $m=-\frac{1}{2}$ by setting
\begin{equation*}
\eta_{\beta,-\frac12,\kappa}( k ) := \frac{ (2k)^{\frac12} }{ \Gamma\big( - \frac{ \beta }{ 2 k } \big) }.
\end{equation*}
We also choose $k \in \C$ with $\Re( k ) > 0$ such that $\frac{\beta_0}{2k} \notin \N$. Then we have $\eta_{\beta_0, -\frac{1}{2} , \kappa_0 } (k) \neq 0$,  and by continuity of $( \beta , m , \kappa ) \mapsto \eta_{\beta,m,\kappa}( k ) $ there exists a neighborhood $\U_0$ of $( \beta_0 , -\frac{1}{2} , \kappa_0 )$ such that  $\eta_{\beta,m,\kappa}(k) \neq 0$ for all $( \beta , m , \kappa )$ in $\U_0$. By Theorem \ref{green0}, one then verifies that for all $f , g \in L^2(\R_+)$, the map $( \beta , m , \kappa ) \mapsto ( f | ( H_{\beta,m,\kappa} + k^2 )^{-1} g )$ is well-defined and holomorphic on $\U_0$ provided that \eqref{resolvent_generic} is extended to $\U_0 \cap \big\{\big(\beta,-\frac{1}{2},\kappa\big) \mid \beta \in \C , \kappa \in \C \big\}$ by
\begin{align*}
R_{\beta,-\frac{1}{2},\kappa}(-k^2 ; x,y) &=\frac{1}{2k \;\!\eta_{\beta,-\frac12,\kappa}(k) }
\begin{cases}
 (2k)^{\frac12} \cI_{\frac{\beta}{2k},-\frac12}(2k x) \cK_{\frac{\beta}{2k},-\frac12}(2k y) & \mbox{ for }0<x<y,\\[2ex]
(2k)^{\frac12} \cI_{\frac{\beta}{2k},-\frac12}(2k y) \cK_{\frac{\beta}{2k},-\frac12}(2k x) &  \mbox{ for }0<y<x ,
\end{cases} \\
&=\frac{ \Gamma \big( 1 - \frac{ \beta }{ 2 k } \big) }{2k }
\begin{cases}
 \cI_{\frac{\beta}{2k},\frac12}(2k x) \cK_{\frac{\beta}{2k},\frac12}(2k y) & \mbox{ for }0<x<y,\\[2ex]
 \cI_{\frac{\beta}{2k},\frac12}(2k y) \cK_{\frac{\beta}{2k},\frac12}(2k x) &  \mbox{ for }0<y<x .
\end{cases}
\end{align*}
Note that this corresponds to the integral kernel of $\big(H_{\beta,\frac{1}{2},0} + k^2 \big)^{-1}=\big(H_{\beta,\frac{1}{2}}^\infty+k^2\big)^{-1}$. This shows that $\{H_{\beta,m,\kappa}\}$ is holomorphic on $\U_0$. The argument easily adapts to the case  $m_0 = \frac{1}{2}$ and $\beta_0 \neq 0$.

As before, if $m_0=\pm \frac{1}{2}$, $\beta_0 \neq 0$, and $\kappa_0 = \infty$, the statement follows from the equalities
\begin{equation*}
\big(H_{\beta,\pm\frac{1}{2},\infty} + k^2 \big)^{-1}
= \big(H_{\beta,\frac{1}{2}}^\infty+k^2\big)^{-1}
= \big(H_{\beta,\pm\frac{1}{2},0} + k^2 \big)^{-1}.
\end{equation*}

The second part of the statement $(i)$ follows directly from \cite[Thm.~3.5]{DR2}. To prove $(ii)$ and $(iii)$, the argument is analogous and simpler: it suffices to use the formulas \eqref{resolvent_1/2} to prove $(ii)$ and \eqref{resolvent_0} to prove $(iii)$.
\end{proof}

The following statement shows that the domains of holomorphy obtained in Theorem \ref{thm:holomorphy} are maximal  for $m\in \Pi $. In particular, we will prove that \eqref{except} are sets of non-removable singularities of the family $(\beta,m,\kappa)\mapsto \{H_{\beta,m,\kappa}\}$.

\begin{proposition}
\begin{enumerate}[label=(\roman*)]
\item For any fixed $\kappa \in \C^\times$, the family of closed operators
$(\beta,m)\mapsto H_{\beta,m,\kappa}$
defined on $\C\times \Pi \setminus \{ (0,-\frac{1}{2} ) , (0,\frac{1}{2}) \}$ cannot be extended by continuity at $( 0 , -\frac{1}{2} )$ and $( 0 , \frac{1}{2} )$.
If $\kappa = 0$, the family $(\beta,m)\mapsto H_{\beta,m,0}$ defined on $\C\times \Pi  \setminus \{ (0,-\frac{1}{2} ) \}$ cannot be extended by continuity at $( 0 , -\frac{1}{2} )$,
and for $\kappa = \infty$ the family $(\beta,m)\mapsto H_{\beta,m,\infty}$ defined on
$\C\times \Pi \setminus \{ (0,\frac{1}{2} ) \}$ cannot be extended by continuity at $( 0 , \frac{1}{2} )$.
\item For any fixed $\beta \in \C$, the family
$(m,\kappa)\mapsto H_{\beta,m,\kappa}$
defined on $\Pi \times  \big(\C \cup \{\infty\}\big) \setminus \{ (0,-1 ) \}$ cannot be extended by continuity at $( 0 , -1 )$.
 \end{enumerate}
\end{proposition}

\begin{proof}
$(i)$ Let us first consider $\beta=0$. Recall that in \cite{DR1}
the family of closed operators $\Pi\times(\C\cup\{\infty\})\ni(m,\kappa)\mapsto H_{m,\kappa}$ has been introduced, and that this family is holomorphic on $\Pi\times(\C\cup\{\infty\})\setminus \{0\}\times (\C\cup\{\infty\})$. Here is its relationship to the families from the present article:
\begin{equation*}
 H_{m, \kappa}:=\left\{
\begin{array}{ll}
H_{0,m,\kappa} & \text{if } m \notin \{ -\frac12 , \frac12 \} \\
H_{0,\frac12}^{\kappa^{-1}}  & \text{if } m = \frac12 \\
H_{0,\frac12}^{\kappa}  & \text{if } m = - \frac12
\end{array}
\right.
\end{equation*}
Let us now focus on $m=-\frac12$ and on $m=\frac12$. We have for any $\kappa \in \C\cup\{\infty\}$
\begin{equation*}
H_{\beta,-\frac12,\kappa}=  H_{\beta,\frac12,\kappa}=H_{\beta,\frac12}=H_{\beta,\frac12}^\infty.
\end{equation*}
Therefore, for $\kappa\neq0$,
\begin{equation}\label{discountinuity2}
\lim_{\beta\to0}H_{\beta,\frac12,\kappa}=H_{0,\frac12}^\infty
\neq H_{0,\frac12}^{\kappa^{-1}}=\lim_{m\to\frac12}H_{0,m,\kappa}.
\end{equation}
Similarly, for $\kappa\neq\infty$,
\begin{equation}\label{discountinuity3}
\lim_{\beta\to0}H_{\beta,-\frac12,\kappa}=H_{0,\frac12}^\infty
\neq H_{0,\frac12}^{\kappa}=\lim_{m\to-\frac12}H_{0,m,\kappa}.
\end{equation}
This proves (i) when $\kappa \not \in \{0,\infty\}$. The proof in these special cases is similar.

$(ii)$ Let us first consider a fixed parameter $\beta \in \C$ and $m=0$. By definition we have
\begin{equation*}
H_{\beta,0,\kappa}=H_{\beta,0}=H_{\beta,0}^\infty,
\end{equation*}
independently of $\kappa\in\C\cup\{\infty\}$. We now consider a fixed parameter $\beta \in \C$ and $\kappa=-1$. Choosing $k \in \C$ with $\Re( k ) > 0$ such that $\frac{\beta}{2k}-\frac{1}{2}\not\in \N$, it follows from \eqref{resolvent_generic} that for any $m\neq0$ in a complex neighborhood of $0$, the integral kernel of the resolvent of $H_{\beta,m,-1}$ is given by
\begin{align*}
&R_{\beta,m,-1}(-k^2 ; x,y) = \frac{1}{2k \;\!\eta_{\beta,m,-1}(k) }  \\
&\times
\begin{cases}
\Big( \frac{(2k)^{-m}}{\Gamma(1-2m)} \cI_{\frac{\beta}{2k},m}(2k x) - \frac{(2k)^m}{\Gamma(1+2m)}\cI_{\frac{\beta}{2k},-m}(2k x)\Big) \cK_{\frac{\beta}{2k},m}(2k y) & \mbox{ for }0<x<y,\\[2ex]
\Big(\frac{(2k)^{-m}}{\Gamma(1-2m)} \cI_{\frac{\beta}{2k},m}(2k y) - \frac{(2k)^{m}}{\Gamma(1+2m)}\cI_{\frac{\beta}{2k},-m}(2k y)\Big) \cK_{\frac{\beta}{2k},m}(2k x) &  \mbox{ for }0<y<x ,
\end{cases}\notag
\end{align*}
where $\eta_{\beta,m,-1}(k)$ is defined in \eqref{def_eta}. One then infers that
\begin{align*}
g_{\beta,k,x}(m):=& \frac{ 1 }{\eta_{\beta,m,-1}(k) } \Big( \frac{(2k)^{-m}}{\Gamma(1-2m)} \cI_{\frac{\beta}{2k},m}(2k x) - \frac{(2k)^m}{\Gamma(1+2m)}\cI_{\frac{\beta}{2k},-m}(2k x)\Big) \\
=& \frac{ \frac{ (2k)^{\frac12} }{ \Gamma( 1 - 2m ) \Gamma ( 1 + 2m ) } \big ( x^{\frac12+m} - x^{\frac12-m} \big ) }{ \frac{(2k)^{-m}}{\Gamma(\frac12+m-\frac{\beta}{2k})\Gamma(1-2m)} - \frac{(2k)^{m}}{\Gamma(\frac12-m-\frac{\beta}{2k})\Gamma(1+2m)} } + O( x^{\frac32 - |\Re(m)| } ) , \quad x \to 0.
\end{align*}
By using this expression, one can verify that the map $m\mapsto g_{\beta,k,x}(m)$, defined in a punctured complex neighborhood of $0$, can be analytically extended at $0$ with
\begin{align*}
g_{\beta,k,x}(0)= - \frac{ (2k)^{\frac12} \Gamma\big( \frac12 - \frac{\beta}{2k} \big) }{ \ln(2k) + \psi\big( \frac12 - \frac{\beta}{2k} \big) + 2\gamma  }  x^{\frac12} \ln( x ) + o( x^{\frac12} ) , \quad x \to 0 .
\end{align*}
Thus, the family of operators $\{\tilde H_{\beta,m,-1}\}$ defined by
\begin{align*}
\tilde H_{\beta,m,-1} =
\left \{
\begin{array}{ll}
H_{\beta,m,-1} & \text{if } m \neq 0 \\
H_{\beta,0}^{0} & \text{if } m = 0 ,
\end{array}
\right.
\end{align*}
is holomorphic for $m\in \Pi$.
It thus follows that
\begin{equation}\label{discountinuity4}
\lim_{\kappa\to-1}H_{\beta,0,\kappa}=H_{\beta,0}^\infty
\neq H_{\beta,0}^0=\lim_{m\to0}H_{\beta,m,-1},
\end{equation}
which concludes the proof.
\end{proof}

\subsection{Blowing up the singularities at $m=0$ and at $m=\pm\frac12$}

As presented above, the boundary conditions for $m=0$ and $m=\pm\frac12$ are described by separate holomorphic families of operators $H_{\beta,0}^\nu$ and $H_{\beta,\frac12}^\nu$. One can however view
these exceptional families as limiting cases of the generic family $H_{\beta,m,\kappa}$.
What is more, after an appropriate change of parameters near the points  $m=0$ and $m=\pm\frac12$
one can holomorphically pass from the generic family to the exceptional families. Such a  procedure is referred to as blowing up a singularity.

More precisely, let us define two new families of operators:
\begin{align}
H_{\beta,m}^{(0),\nu}&:=\begin{cases}
H_{\beta,m,\kappa},&\quad m\neq0,\qquad \kappa=\kappa^{(0)}(m,\nu)
:=-\frac{1}{(1+2m\nu)},\\
H_{\beta,0}^\nu,&\quad m=0;
\end{cases}\\
H_{\beta,m}^{(\frac12),\nu}&:=\begin{cases}
H_{\beta,m,\kappa},&\quad m\neq\frac12,\qquad \kappa=\kappa^{(\frac12)}(\beta,m,\nu):=\frac{1}{\big(-\frac{\beta}{(2m-1)}+\nu\big)},\\
H_{\beta,\frac12}^\nu,&\quad m=\frac12.
\end{cases}
\end{align}
Thus $H_{\beta,m}^{(0),\nu}$ includes both $H_{\beta,0}^{\nu}$ and
$  H_{\beta,m,\kappa}$, and
$H_{\beta,m}^{(\frac12),\nu}$ includes both $H_{\beta,\frac12}^{\nu}$ and
$  H_{\beta,m,\kappa}$.

\begin{theorem}
\begin{enumerate}[label=(\roman*)]
\item
The family $\{H_{\beta,m}^{(0),\nu}\}$ is holomorphic on
$\C\times\Pi\times \big(\C\cup\{\infty\}\big)$ except for
\begin{equation*}
\big(0,-\tfrac12\big){\times}\big(\C\cup\{\infty\}\big)\,\cup\,
\big(0,\tfrac12\big){\times}\big(\C\cup\{\infty\}\big).
\end{equation*}
\item
The family $\{H_{\beta,m}^{(\frac12),\nu}\}$ is holomorphic on
$\C\times\Pi\times \big(\C\cup\{\infty\}\big)$ except for
\begin{equation*}
\big(0,-\tfrac12\big){\times}\big(\C\cup\{\infty\}\big)\,\cup\,
\{(\beta,0,-1-\beta)\ \mid\ \beta\in\C\}.
\end{equation*}
\end{enumerate}
\end{theorem}

\begin{proof}
For any fixed $m\in\Pi$, Theorem \ref{thm:holomorphy} shows that $(\beta,\nu) \mapsto H_{\beta,m}^{(0),\nu}$ is holomorphic on $\C \times (\C \cup \{\infty\})$ if $m\neq \pm\frac12$ and on $(\C\setminus\{0\}) \times (\C \cup \{\infty\})$ if $m=\pm\frac12$. Likewise, $(\beta,\nu) \mapsto H_{\beta,m}^{(\frac12),\nu}$ is holomorphic in $\C \times (\C \cup \{\infty\})$ if $m\not \in \big\{-\frac{1}{2},0\big\}$, on $\C \times (\C \cup \{\infty\}) \setminus \{ \beta , 1 - \beta \, | \, \beta \in \C \}$ if $m=0$, and on $(\C\setminus\{0\}) \times (\C \cup \{\infty\})$ if $m=-\frac12$. It remains to study holomorphy in $m$ for fixed $(\beta,\nu)$.

Recall that in Theorem \ref{thm_hard}
we introduced the functions $\omega_{\beta,m,\kappa}(k)$,
$\omega_{\beta,0}^\nu(k)$, and $\omega_{\beta,\frac12}^\nu(k)$. Let us now define two more functions
\begin{align*}
\omega_{\beta,m}^{(0),\nu}(k)&:=\begin{cases}
\omega_{\beta,m,\kappa}(k),&\quad m\neq0,\qquad \kappa=\kappa^{(0)}(m,\nu),\\
\omega_{\beta,0}^\nu(k),&\quad m=0;
\end{cases}\\
\omega_{\beta,m}^{(\frac12),\nu}(k)&:=\begin{cases}
\omega_{\beta,m,\kappa}(k),&\quad m\neq\frac12,\qquad \kappa=\kappa^{(\frac12)}(\beta,m,\nu),\\
\omega_{\beta,\frac12}^\nu(k),&\quad m=\frac12.
\end{cases}
\end{align*}
Clearly, by Theorem \ref{thm_hard} one has
\begin{align*}
& R_{\beta,m}^{(0),\nu}(-k^2 ; x,y) \\
& = R_{\beta,m}(-k^2;x,y)
+ \frac{\Gamma\big(\frac12+m-\frac{\beta}{2k}\big)\Gamma\big(\frac12-m-\frac{\beta}{2k}\big)}
{2k\;\!\omega_{\beta,m}^{(0),\nu}(k)}
\cK_{\frac{\beta}{2k},m}(2k y)  \cK_{\frac{\beta}{2k},m}(2k x)
\end{align*}
and
\begin{align*}
& R_{\beta,m}^{(\frac12),\nu}(-k^2 ; x,y) \\
&=  R_{\beta,m}(-k^2;x,y)
+ \frac{\Gamma\big(\frac12+m-\frac{\beta}{2k}\big)\Gamma\big(\frac12-m-\frac{\beta}{2k}\big)}
{2k\;\!\omega_{\beta,m}^{(\frac12),\nu}(k)}
\cK_{\frac{\beta}{2k},m}(2k y)  \cK_{\frac{\beta}{2k},m}(2k x).
\end{align*}

Let us show that, for fixed $(\beta,\nu)$ such that $\frac{\beta}{2k}-\frac{1}{2}\not \in \N$, the map
\begin{equation}
m\mapsto \frac{\Gamma\big(\frac12+m-\frac{\beta}{2k}\big)\Gamma\big(\frac12-m-\frac{\beta}{2k}\big)}
{2k\;\!\omega_{\beta,m}^{(\frac12),\nu}(k)}  \label{eq:zlo}
\end{equation}
is holomorphic for $m$ near $0$. It is clearly holomorphic in a punctured neighborhood of $0$. Hence it suffices to show that it is continuous at $m=0$.
Recall from \eqref{def_omegabm} that
\begin{equation}\label{omegg}
\omega_{\beta,m,\kappa}(k)
=\Big(1+\frac{(2k)^{-2m}\Gamma\big(\frac12-m-\frac{\beta}{2k}\big)\Gamma(1+2m)}{\kappa
\Gamma\big(\frac12+m-\frac{\beta}{2k}\big)\Gamma(1-2m)}\Big)\frac{\pi}{\sin(2\pi m)}.
\end{equation}
Then, by inserting $\kappa=\kappa^{(0)}(m,\nu)$ for $m\neq0$ into \eqref{omegg} we obtain
\begin{align*}
\omega_{\beta,m}^{(0),\nu}(k)
& =\pi \frac{\Gamma\big(\frac12+m-\frac{\beta}{2k}\big)\Gamma(1-2m)-
(2k)^{-2m}\Gamma\big(\frac12-m-\frac{\beta}{2k}\big)\Gamma(1+2m)}{    \Gamma\big(\frac12+m-\frac{\beta}{2k}\big)\Gamma(1-2m)\sin(2\pi m)}\\
&\quad -\nu\frac{ (2k)^{-2m}\Gamma\big(\frac12-m-\frac{\beta}{2k}\big)\Gamma(1+2m)2\pi m}{\Gamma\big(\frac12+m-\frac{\beta}{2k}\big)\Gamma(1-2m)\sin(2\pi m)}\\
&\underset{m\to0}\to
\psi\Big(\frac12-\frac{\beta}{2k}\Big)+2\gamma+\ln(2k)-\nu \\
&=\omega_{\beta,0}^{(0),\nu}(k).
\end{align*}
Thus \eqref{eq:zlo} is holomorphic for $m$ near $0$.

Similarly, let us show that, for fixed $(\beta,\nu)$ such that $\frac{\beta}{2k}\not \in \N$, the map
\begin{equation}
m\mapsto \frac{\Gamma\big(\frac12+m-\frac{\beta}{2k}\big)\Gamma\big(\frac12-m-\frac{\beta}{2k}\big)}
{2k\;\!\omega_{\beta,m}^{(\frac12),\nu}(k)}  \label{eq:zlo2}
\end{equation}
is holomorphic for $m$ near $\frac12$. 
By inserting $\kappa=\kappa^{(\frac12)}(\beta,m,\nu)$ for $m\neq\frac12$ into \eqref{omegg} we obtain
\begin{align*}
\omega_{\beta,m}^{(\frac12),\nu}(k)
&= \pi \frac{\Gamma\big(\frac12+m-\frac{\beta}{2k}\big)\Gamma(2-2m)+\beta
(2k)^{-2m}\Gamma\big(\frac12-m-\frac{\beta}{2k}\big)\Gamma(1+2m)}{    \Gamma\big(\frac12+m-\frac{\beta}{2k}\big)\Gamma(2-2m)\sin(2\pi m)}\\
&\quad -\nu\frac{ (2k)^{-2m}\Gamma\big(\frac12-m-\frac{\beta}{2k}\big)\Gamma(1+2m)\pi (2m-1)}{\Gamma\big(\frac12+m-\frac{\beta}{2k}\big)\Gamma(2-2m)\sin(2\pi m)}\\
&\underset{m\to\frac12}\to
-\frac12\psi\Big(1- \frac{\beta}{2k} \Big)- \frac12\psi\Big(- \frac{\beta}{2k} \Big) -2\gamma - \ln(2k) + 1-\frac{\nu}{\beta}\\
&=\omega_{\beta,\frac12}^{(\frac12),\nu}(k),
\end{align*}
which proves that \eqref{eq:zlo2} is holomorphic for $m$ near $\frac12$.

The remaining restrictions on the domain of holomorphy are inferred directly from
Theorem \ref{thm:holomorphy}.
\end{proof}

\subsection{Eigenprojections}\label{sec_eigenp}

Let us now describe a family of projections $\{P_{\beta,m}(\lambda)\}$ which is closely related to the Whittaker operator. We will define it by specifying its integral kernel.

We first introduce a holomorphic function for $m\not \in \{-\frac12,0,\frac12\}$ by
\begin{equation*}
(\beta,m,k)\mapsto  \zeta_{\beta,m}(k)
:=\frac{\pi\Big(2m+\frac{\beta}{2k}\psi\big(\frac12+m-\frac{\beta}{2k}\big)-
\frac{\beta}{2k}\psi\big(\frac12-m-\frac{\beta}{2k}\big)\Big)}{\sin(2\pi m)}.
\end{equation*}
One easily observes that $\zeta_{\beta,m}(k)=  \zeta_{\beta,-m}(k)$.
We can extend this function continuously to $m\in\{-\frac12,0,\frac12\}$ by
\begin{align*}
\zeta_{\beta,0}(k)
&=1+\frac{\beta}{2k}\psi'\Big(\frac12-\frac{\beta}{2k}\Big),\\
\zeta_{\beta,-\frac12}(k)=  \zeta_{\beta,\frac12}(k)
&:= -\bigg(1+\frac{\beta}{4k}\psi'\Big(1-\frac{\beta}{2k}\Big)+\frac{\beta}{4k}\psi'\Big(-\frac{\beta}{2k}\Big)\bigg).
\end{align*}

We now consider $\lambda\in\C\backslash[0,\infty[$, and as usual we write $\lambda=-k^2$ with $\Re (k)>0$.
We then define the integral kernel  $P_{\beta,m}(\lambda;x,y)$:
\begin{align}\label{defo}
P_{\beta,m}(-k^2;x,y)
:=\frac{k
\Gamma\big(\frac12+m-\frac{\beta}{2k}\big) \Gamma\big(\frac12-m-\frac{\beta}{2k}\big)}    {\zeta_{\beta,m}(k)}\cK_{\frac{\beta}{2k},m}(2kx)\cK_{\frac{\beta}{2k},m}(2ky).
\end{align}

The definition \eqref{defo} naturally extends to
$\lambda\in]0,\infty[$, where we distinguish between points coming from the upper and lower half-plane by writing $\lambda\pm\i0=-(\mp\i\mu)^2$ with $\mu>0$.
Thus, let us set $k=\mp\i \mu$ and
\begin{equation*}
\zeta_{\beta,m}(\mp\i \mu):=\frac{\pi\Big(
2m\pm\i\frac{\beta}{2\mu}\psi\big(\frac12+m\mp\i\frac{\beta}{2\mu}\big)\mp\i\frac{\beta}{2\mu}\psi\big(\frac12-m\mp\i\frac{\beta}{2\mu}\big)\Big)}{\sin(2\pi m)}
\end{equation*}
which can be naturally extended to $m\in\{-\frac12,0,\frac12\}$ by
\begin{align*}
\zeta_{\beta,0}(\mp\i\mu)
&=1\pm\i\frac{\beta}{2\mu}\psi'\Big(\frac12\mp\i\frac{\beta}{2\mu}\Big),\\
\zeta_{\beta,-\frac12}(\mp\i\mu)=  \zeta_{\beta,\frac12}(\mp\i\mu)
&:= -\bigg(1\pm\i\frac{\beta}{4\mu}\psi'\Big(\mp\i\frac{\beta}{2\mu}\Big)\pm\i\frac{\beta}{4k}\psi'\Big(1\mp\i\frac{\beta}{2\mu}\Big)\bigg).
\end{align*}
For $k=\mp\i \mu$ we can then rewrite \eqref{defo} as
\begin{equation*}
P_{\beta,m}(\mu^2\pm\i0;x,y) := \frac{\e^{\pm\i\pi m}\mu
\Gamma\big(\frac12+m\mp \i\frac{\beta}{2\mu}\big) \Gamma\big(\frac12-m\mp \i\frac{\beta}{2\mu}\big)}{\zeta_{\beta,m}(\mp\i\mu)}
\cH_{\frac{\beta}{2\mu},m}^\pm(2\mu x)
\cH_{\frac{\beta}{2\mu},m}^\pm(2\mu y).
\end{equation*}

Finally, to handle $\lambda=0$
we shall use the function
$\frac{ \sin(2\pi m)}{m(4m^2-1) }$
extended to $\{-\frac12,0,\frac12\}$ by
\begin{equation*}
\frac{ \sin(2\pi m)}{m(4m^2-1) }\Big|_{m=0}=-2\pi
\quad \hbox{and} \quad
\frac{ \sin(2\pi m)}{m(4m^2-1) }\Big|_{m=\pm\frac12}=-\pi.
\end{equation*}
We set, for $\pm \Im(\sqrt{\beta} ) > 0$,
\begin{equation*}
P_{\beta,m}(0;x,y)
:=3\e^{\pm \i\pi 2m} \beta
\frac{\sin(2\pi m)}{m(4m^2-1) } (\beta x)^{\frac14}\cH_{2m}^{\pm}(2\sqrt{\beta x}) (\beta y)^{\frac14}\cH_{2m}^{\pm}(2\sqrt{\beta y}).
\end{equation*}

The integral kernel
$P_{\beta,m}(-k^2;x,y)$ defines an operator-valued map $(\beta,m,k)\mapsto P_{\beta,m}(-k^2)$ described in the following proposition.

\begin{proposition}
On  the set
\begin{align}\label{eq:domain_proj0}
\begin{split}
&\C\times\Pi\times\{k\in \C\mid\Re(k)>0\}  \\
&\ \cup\{(\beta,m,\mp\i\mu)\mid \beta \in \C , \,  m\in\Pi , \, 0<\mu<\pm\Im(\beta) \} \\
&\ \cup\{(\beta,m,0)\mid \beta \in \C , \,  m\in\Pi , \, 0<\pm\Im(\sqrt{\beta}) \} ,
\end{split}
\end{align} the function $(\beta,m,k)\mapsto P_{\beta,m}(-k^2)$
has values in bounded projections. Moreover, it is continuous on
\begin{align}\label{eq:domain_proj}
\begin{split}
&\C\times\Pi\times\{k\in \C\mid\Re(k)>0\}  \\
&\ \cup\{(\beta,m,\mp\i\mu)\mid \beta \in \C , \,  m\in\Pi , \, 0<\mu<\pm\Im(\beta) \} ,
\end{split}
\end{align}
and holomorphic on  $\C\times\Pi\times\{\Re(k)>0\}$. It satisfies
\begin{align}
P_{\beta,m}(-k^2)&= P_{\beta,-m}(-k^2), \label{eq:proj1} \\
P_{\beta,m}(-k^2)^\#&= P_{\beta,m}(-k^2),\\
P_{\beta,m}(-k^2)^*&= P_{\bar\beta,\bar m}(-\overline {k^2}) , \label{eq:proj3}
\end{align}
for all $(\beta,m,k)$ in the set \eqref{eq:domain_proj0}.
 \end{proposition}

\begin{proof}
The fact that $P_{\beta,m}(-k^2)$ are rank-one projections follows directly from their expressions together with Corollaries \ref{prop-lag2} and \ref{cor:lag2} and Proposition \ref{propB6}. Continuity on the domain \eqref{eq:domain_proj} and holomorphy on $\C\times\Pi\times\{\Re(k)>0\}$, as well as the relations \eqref{eq:proj1}--\eqref{eq:proj3}, follow again from the expressions involved in the definitions of $P_{\beta,m}(-k^2)$.
\end{proof}

We recall from Proposition \ref{propo2} that the operators $H_{\beta,m,\kappa}$,  $H_{\beta,0}^\nu$ and $H_{\beta,\frac12}^\nu$ are self-transposed. Moreover, it follows from Theorem \ref{spectrum} and its proof that all eigenvalues of these operators are simple. If $\lambda$ is a simple eigenvalue of a self-transposed operator $H$ associated to an eigenvector $u$ such that $\langle u | u \rangle = 1$, we define the \emph{self-transposed eigenprojection} associated to $\lambda$ as
\begin{equation*}
P = \langle u | \cdot \rangle u.
\end{equation*}
In the case where $\lambda$ is in addition an isolated point of the spectrum, it is then easy to see that the self-transposed eigenprojection $P$ coincides with the usual Riesz projection corresponding to $\lambda$.

\begin{theorem}
Let $\beta \in \C$, $m\in\Pi\setminus\big\{-\frac12,0,\frac12\big\}$, $\kappa \in \C \cup \{ \infty \}$ and $\nu \in \C \cup \{ \infty \}$. Let $\lambda \in \C$ be an eigenvalue of $H_{\beta,m,\kappa}$,  $H_{\beta,0}^\nu$ or $H_{\beta,\frac12}^\nu$ respectively. Then the self-transposed eigenprojection is $P_{\beta,m}(\lambda)$ for the corresponding value of $m$.
\end{theorem}

\begin{proof}
We prove the theorem in the case where $\lambda = -k^2$ with $\Re(k)>0$ and $m \notin \big\{ -\frac12 , 0 , \frac12 \big\}$.
The other cases are similar.

From the proof of Theorem \ref{spectrum}, we know that if $\lambda$ is an eigenvalue of $H_{\beta,m,\kappa}$, then a corresponding eigenstate is given by $x \mapsto \Ka_{\frac{\beta}{2k},m}(2kx)$. Corollary \ref{prop-lag2} shows that
\begin{equation*}
\big \langle \Ka_{\frac{\beta}{2k},m}(2k \cdot ) \mid \Ka_{\frac{\beta}{2k},m}(2k \cdot) \big \rangle =
\frac{\pi}{\sin(2\pi m)} \frac{2m+\frac{\beta}{2k}\psi\big(\frac12+m-\frac{\beta}{2k}\big)-\frac{\beta}{2k}\psi\big(\frac12-m-\frac{\beta}{2k}\big)}
{k\Gamma\big(\frac12+m-\frac{\beta}{2k}\big) \Gamma\big(\frac12-m-\frac{\beta}{2k}\big)}.
\end{equation*}
This proves that $P_{\beta,m}(-k^2)$ is the self-transposed eigenprojection corresponding to $\lambda$, as claimed.
\end{proof}

The point $k=0$ is rather special for the family $P_{\beta,m}( - k^2 )$, as shown in next proposition.

\begin{proposition}\label{prop:not_continuous}
Let $m \in \Pi$ and $\beta \in \C$ such that $\pm \Im( \sqrt{\beta} ) > 0$. Then the map $k \mapsto P_{\beta,m}( - k^2 )$ is not continuous at $k=0$.
\end{proposition}

\begin{proof}
We consider the case where $m \notin \{ -\frac12 , 0 , \frac12 \}$. The other cases are similar.

First, we claim that for all continuous and compactly supported function $f$,
\begin{equation*}
\lim_{k\to 0} \big \langle f | P_{\beta,m}( - k^2 ) f \big \rangle = \big \langle f | P_{\beta,m}( 0 ) f \big \rangle,
\end{equation*}
where $k\in \C$ is chosen such that $\Re(k)>0$ and $\pm \big(\arg(\beta)-\arg(k)\big)\in ]\varepsilon, \pi-\varepsilon[$ with $\varepsilon > 0$. To shorten the expressions below, we set in this proof
\begin{equation*}
g_{\beta,m,k}( x ) :=\mp \i \frac{\Gamma\big(\frac12+m-\frac{\beta}{2k}\big)}{\sqrt{\pi}}
\Big(\frac{\beta}{2k}\Big)^{\frac12-m}
\cK_{\frac{\beta}{2k},m}(2k x ) ,
\end{equation*}
and
\begin{equation*}
g_{\beta,m,0}( x ) := (\beta x )^{\frac14}\cH_{2m}^\pm(2\sqrt{\beta x}) .
\end{equation*}
We show that $g_{\beta,m,k}$ is uniformly bounded, for $k$ satisfying the conditions above, by a locally integrable function. From the definition \eqref{Taylor_2} of $\cI_{\beta,m}$ and proceeding as in the proof of Proposition \ref{propB1}, we obtain that,
for $k\in \C$ such that $\Re(k)>0$, $|k|<1$, and $\pm \big(\arg(\beta)-\arg(k)\big)\in ]\varepsilon, \pi-\varepsilon[$ with $\varepsilon > 0$,
\begin{align*}
\Big | \Big ( \frac{ \beta }{2k} \Big )^{\frac12+m} \cI_{\frac{\beta}{2k},m}(2kx) \Big |
&=\Big | (\beta x)^{\frac{1}{2}+m}\e^{- k x } \suma{j=0}{\infty}\frac{\big(\frac{1}{2}+m - \frac{\beta}{2k} \big)_j}{\Gamma(1+2m+j)}\;\!\frac{ (2kx)^j}{j!} \Big | \\
&\le  |\beta x|^{\frac{1}{2}+m} \suma{j=0}{\infty} \frac{c^j x^j }{ | \Gamma(1+2m+j) |} ,
\end{align*}
for some constant $c>0$ depending on $\beta$ and $m$ but independent of $k$ and $x$. Using that
\begin{align*}
\nonumber & g_{\beta,m,k}( x ) = \frac{\mp\i\sqrt{\pi}}{\sin(2\pi m)}
\Big(\frac{\beta}{2k}\Big)^{\frac12-m}
\Big(-\frac{\Gamma\big(\frac12+m-\frac{\beta}{2k}\big)}{\Gamma\big(\frac{1}{2}-m-\frac{\beta}{2k}\big)}
\cI_{\frac{\beta}{2k},m}(2kx)
+ \cI_{\frac{\beta}{2k},-m}(2kx)\Big),
\end{align*}
together with Lemma \ref{lem_key}, one then deduces that
\begin{align*}
 & \big | g_{\beta,m,k}( x ) \big | \le c_1 \e^{c_2 x } ,
\end{align*}
for some positive constants $c_1$, $c_2$ independent of $k$ and $x$.

The previous bound together with the dominated convergence theorem and Proposition \ref{propB1} show that
\begin{align*}
\lim_{k\to0} \big \langle g_{\beta,m,k} | f \big \rangle &= \big \langle g_{\beta,m,0} | f \big \rangle ,
\end{align*}
for all continuous and compactly supported function $f$, and for $k$ satisfying the conditions exhibited above.
We then have that
\begin{align*}
& \big \langle f | P_{\beta,m}( - k^2 ) f \big \rangle \\
& = \frac{k\sin(2\pi m) \Gamma\big(\frac12+m-\frac{\beta}{2k}\big) \Gamma\big(\frac12-m-\frac{\beta}{2k}\big)}    {\pi\big[2m+\frac{\beta}{2k}\psi\big(\frac12+m-\frac{\beta}{2k}\big)-\frac{\beta}{2k}\psi\big(\frac12-m-\frac{\beta}{2k}\big)\big]} \big \langle \cK_{\frac{\beta}{2k},m}(2k \cdot ) | f \big \rangle^2 \\
& = - \frac{k\sin(2\pi m)}    {2m+\frac{\beta}{2k}\psi\big(\frac12+m-\frac{\beta}{2k}\big)-\frac{\beta}{2k}\psi\big(\frac12-m-\frac{\beta}{2k}\big)} \frac{ \Gamma\big(\frac12-m-\frac{\beta}{2k}\big)}{ \Gamma\big(\frac12+m-\frac{\beta}{2k}\big) } \Big(\frac{\beta}{2k}\Big)^{2m-1} \big \langle g_{\beta,m,k} | f \big \rangle^2 \\
& = \frac{2k^2\sin(2\pi m)}{\beta\big[\frac{\beta}{2k} \big(\frac{2k}{\beta}\big)^3 \frac{m}{6}( -1 + 4m^2 ) +o(1) \big ] } \frac{ \Gamma\big(\frac12-m-\frac{\beta}{2k}\big)}{ \Gamma\big(\frac12+m-\frac{\beta}{2k}\big) } \Big(\frac{\beta}{2k}\Big)^{2m} \big \langle g_{\beta,m,k} | f \big \rangle^2 \\
& \underset{k\to0}{\to} \frac{3\beta\sin(2\pi m)}{m ( 4m^2 - 1 ) } \e^{\pm\i\pi 2m} \big \langle g_{\beta,m,0} | f \big \rangle^2 \\
& = \big \langle f | P_{\beta,m}( 0 ) f \big \rangle,
\end{align*}
where we used Lemma \ref{lemB5} in the third equality.

Now, we claim that $P_{\beta,m}(-k^2)$ is not continuous at $k=0$ for the strong operator topology. Indeed, using that $P_{\beta,m}(-k^2)$ is a self-transposed projection, we infer that, for $f$ continuous and compactly supported,
\begin{align*}
& \big \langle \big( P_{\beta,m}( -k^2 ) - P_{\beta,m}( 0 ) \big) f | \big( P_{\beta,m}( -k^2 ) - P_{\beta,m}( 0 ) \big) f \big \rangle \\
& = \big \langle P_{\beta,m}( -k^2 ) f | f \big \rangle + \big \langle P_{\beta,m}( 0 ) f | f \big \rangle  - 2 \big \langle P_{\beta,m}( 0 ) f | P_{\beta,m}( -k^2 ) f \big \rangle .
\end{align*}
A similar computation as above gives
\begin{align*}
& \big \langle P_{\beta,m}( 0 ) f | P_{\beta,m}( -k^2 ) f \big \rangle \\
& = \frac{3\beta\sin(2\pi m)}{m \e^{\mp \i\pi 2m} \big(4m^2-1\big) } \frac{2k^2\sin(2\pi m)}{\beta\big[\frac{\beta}{2k} \big(\frac{2k}{\beta}\big)^3 \frac{m}{6}( -1 + 4m^2 ) +o(1) \big ] } \frac{ \Gamma\big(\frac12-m-\frac{\beta}{2k}\big)}{ \Gamma\big(\frac12+m-\frac{\beta}{2k}\big) } \Big(\frac{\beta}{2k}\Big)^{2m} \\
& \quad \times \langle f | g_{\beta,m,k} \rangle \langle g_{\beta,m,k} | g_{\beta,m,0} \rangle \langle g_{\beta,m,0} | f \rangle \\
& \underset{k\to0}{\to} 0 ,
\end{align*}
since $\lim\limits_{k\to 0}\langle g_{\beta,m,k} | g_{\beta,m,0} \rangle =0$ by Remark \ref{curious}, while the other terms converge. Therefore,
\begin{align*}
& \big \langle \big( P_{\beta,m}( -k^2 ) - P_{\beta,m}( 0 ) \big) f | \big( P_{\beta,m}( -k^2 ) - P_{\beta,m}( 0 ) \big) f \big \rangle \underset{k\to0}{\to} 2 \big \langle P_{\beta,m}( 0 ) f | f \big \rangle \neq 0 ,
\end{align*}
for suitably chosen compactly supported functions $f$. This proves that $P_{\beta,m}(-k^2)$ is not continuous at $k=0$.
\end{proof}

\appendix
%%%%%%%%%%%%%%%%%%%%%%%%%%%%%%%%%%%%%%%%%%%%%%%%%%%%%%%%%%%%%%%%%%%%%%%
\section{The Whittaker equation}\label{sec_Whittaker}
\setcounter{equation}{0}
\renewcommand{\theequation}{A.\arabic{equation}}
%%%%%%%%%%%%%%%%%%%%%%%%%%%%%%%%%%%%%%%%%%%%%%%%%%%%%%%%%%%%%%%%%%%%%%%%

\subsection{General theory}\label{subsec:general}

In this section we collect basic information about the Whittaker equation.
This should be considered as a supplement to \cite[Sec.~2]{DR2}.

The Whittaker equation is represented by the equation
\begin{equation}\label{whit}
\Big(L_{\beta,m^2}+\frac14\Big)f := \bigg(-\partial_z^2+\Big(m^2-\frac14\Big)\frac{1}{z^2}-\frac{\beta}{z}+\frac{1}{4}\bigg)f=0.
\end{equation}
We observe that the equation does not change when we replace $m$ with $-m$. It has also another symmetry:
\begin{equation}\label{symme}
\Big(L_{\beta,m^2}+\frac14\Big)f(z)=0\quad\Rightarrow\quad \Big(L_{-\beta,m^2}+\frac14\Big)f(-z)=0.
\end{equation}
Solutions of \eqref{whit} are provided by the functions
$z \mapsto \Ia_{\beta,\pm m}(z)$ which are defined by
\begin{align}\label{Taylor_2}
\cI_{\beta,m}(z)
\nonumber &=z^{\frac{1}{2}+m}\e^{\mp\frac{z}{2}}\frac{{}_{1}F_{1}\big(\frac{1}{2}+m\mp\beta;1+2m;\pm z\big)}{\Gamma(1+2m)} \\
&= z^{\frac{1}{2}+m}\e^{\mp\frac{z}{2}}\suma{k=0}{\infty}\frac{\naw{\frac{1}{2}+m\mp\beta }_k}{\Gamma(1+2m+k)}\;\!\frac{(\pm z)^k}{k!},
\end{align}
where $(a)_k:=a(a+1)\cdots(a+k-1)$ and $(a)_0=1$ are the usual Pochhammer's symbols and ${}_1F_1$ is Kummer's confluent hypergeometric function.
For $\Re(m)>-\frac{1}{2}$ and $\Re\big(m\mp \beta +\frac{1}{2}\big)>0$
the function $\cI_{\beta,m}$ has also an integral representation given by
\begin{equation*}
\cI_{\beta,m}(z)
=\frac{z^{\frac12+m}}{\Gamma\big(\frac12+m+\beta\big)
\Gamma\big(\frac12+m-\beta\big)}
\int_0^1\e^{\pm z(s-\frac12)}s^{m\mp\beta-\frac12}(1-s)^{m\pm \beta-\frac12}\;\!\d s.
\end{equation*}
Based on \eqref{Taylor_2} one easily gets
\begin{equation}\label{wro1}
\Wr\big(\Ia_{\beta,m},\Ia_{\beta,-m};x\big)
=-\frac{\sin(2\pi m)}{\pi}
\end{equation}
as well as the following identity
\begin{equation}\label{eq_miracle}
\cI_{\beta,m}(z)=\e^{\mp\i\pi(\frac12+m)}  \cI_{-\beta,m}\big(\e^{\pm\i\pi}z\big).
\end{equation}

Another solution of \eqref{whit} is provided by the function $z \mapsto \Ka_{\beta,m}(z)$.
For $m\not \in \frac12\Z$ it can be defined by the following relation:
\begin{equation}\label{generic}
\cK_{\beta,m}
= \frac{\pi}{\sin(2\pi m)}\bigg(-\frac{\cI_{\beta,m}}{\Gamma\big(\frac{1}{2}-m-\beta\big)}
+ \frac{\I{\beta}{-m}}{\Gamma\big(\frac{1}{2}+m-\beta\big)}\bigg).
\end{equation}
For the remaining $m$ we can extend the definition of
$\Ka_{\beta,m}$ by continuity, see  Subsect.~\ref{The degenerate case}.
Note that $\Ka_{\beta,-m}=\Ka_{\beta,m}$, and that the function
$\Ka_{\beta,m}$ can also be expressed in
terms of the function ${}_2F_0$, namely:
\begin{equation*}
\cK_{\beta,m}(z)
= z^\beta\e^{-\frac{z}{2}}{}_2F_0\big(\tfrac12+m-\beta,
\tfrac12-m-\beta;-;-z^{-1}\big).
\end{equation*}
An alternative definition of $\Ka_{\beta,m}$  can be provided by an integral representation valid for $\Re\big(-\beta \mp m+\frac{1}{2}\big)>0$ and $\Re(z)>0$:
\begin{equation*}
\cK_{\beta,m}(z)
=\frac{z^{\frac12\mp m}\e^{-\frac{z}{2}}}{\Gamma\big(\frac12-\beta\mp m\big)}
\int_0^\infty\e^{-zs}s^{-\frac12-\beta\mp m}(1+s)^{-\frac12+\beta\mp m}\d s.
\end{equation*}
Note that the function $\Ka_{\beta,m}$ decays exponentially for large $\Re(z)$, more precisely, if $\varepsilon>0$ and $\abs{\arg(z)}<\frac32\pi-\varepsilon$, then one has
\begin{equation}\label{Kbm-around-infinity}
\cK_{\beta,m}(z) = z^\beta\;\!\e^{-\frac{z}{2}} \big(1+O(z^{-1})\big).
\end{equation}
By using the relation \eqref{generic} one also obtains that
\begin{equation}\label{Wr}
\Wr\big(\Ia_{\beta,m},\Ka_{\beta,m};x\big) = -\frac{1}{\Gamma\big(\frac{1}{2}+m-\beta\big)}.
\end{equation}

We would like to treat
$\cI_{\beta,m}$,  $\cI_{\beta,-m}$  and  $\cK_{\beta, m}$ as the principal solutions of the Whittaker equation \eqref{whit}. There are however cases for which this is not sufficient. Therefore, we introduce below a fourth solution, which we denote by $\cX_{\beta,m}$.
To the best of our knowledge, this function has never appeared elsewhere in the literature.

The function $\cK_{\beta,m}$ is distinguished by the fact that it decays exponentially,
while the solutions $\cI_{\beta,\pm m}(z)$
explode exponentially, see \cite[Eq.~(2.14) \& (2.22)]{DR2}.
This is also the case for the analytic continuations of
$\Ka_{-\beta,m}$ by the angles $\pm\pi$, which by the symmetry \eqref{symme} are also solutions of \eqref{whit}. It will be convenient to introduce a name for a solution constructed from these two analytic continuations. There is some arbitrariness for this choice, but we have decided on:
\begin{equation}\label{eq_def_X}
\cX_{\beta,m}(z):=\tfrac12\Big(
\e^{-\i\pi(\frac{1}{2}+m)}\cK_{-\beta,m}\big(\e^{\i\pi}z\big)
+ \e^{\i\pi(\frac{1}{2}+m)}\cK_{-\beta,m}\big(\e^{-\i\pi}z\big)\Big).
\end{equation}

As a consequence of this definition and of \eqref{eq_miracle} one gets the relations
\begin{equation}\label{generic*1}
\cX_{\beta,m}(z)
= -\frac{\pi}{\sin(2\pi m)}\bigg(\frac{\cI_{\beta,m}(z)}{\Gamma\big(\frac{1}{2}-m+\beta\big)}
- \frac{\cos(2\pi m)\cI_{\beta,-m}(z)}{\Gamma\big(\frac{1}{2}+m+\beta\big)}\bigg),
\end{equation}
and
\begin{equation*}
\e^{\mp\i\pi(\frac12+m)}\cK_{-\beta,m}\big(\e^{\pm\i\pi}z\big)
=\cX_{\beta,m}(z)\mp\frac{\i\pi\cI_{\beta,-m}(z)}{\Gamma\big(\frac12+m+\beta\big)}.
\end{equation*}
In addition, by using the equalities
\begin{align}\label{eq_cos}
\nonumber \cos\big(\pi(m-\beta)\big) & = \cos\big(2\pi m - \pi(m+\beta)\big) \\
&= \cos(2\pi m)\cos \big(\pi(m+\beta)\big)+ \sin(2\pi m)\sin \big(\pi(m+\beta)\big),
\end{align}
one infers from \eqref{generic} and \eqref{generic*1} that
\begin{align*}
& \frac{\cos(2\pi m)\cK_{\beta,m}}{\Gamma\big(\frac12+m+\beta\big)}-
\frac{\cX_{\beta,m}}{\Gamma\big(\frac12+m-\beta\big)} \\
&= \frac{1}{\sin(2\pi m)}
\Big(\cos\big(\pi(m-\beta)\big)-\cos(2\pi m)\cos\big(\pi(m+\beta)\big)\Big)
\cI_{\beta,m} \\
& = \sin\big(\pi(m+\beta)\big)\cI_{\beta,m},
\end{align*}
which finally leads to the relation
\begin{equation}\label{eq_lin_comb}
\cI_{\beta,m} =\frac{1}{\sin(\pi(m+\beta))}\bigg(
\frac{\cos(2\pi m)}{\Gamma\big(\frac12+m+\beta\big)} \cK_{\beta,m}
 - \frac{1}{\Gamma\big(\frac12+m-\beta\big)}  \cX_{\beta,m}\bigg).
\end{equation}

By taking formulas \eqref{generic}, \eqref{generic*1}, and  \eqref{eq_cos}
into account, one infers that the Wronskian is provided by
\begin{equation*}
\Wr(\cK_{\beta,m},\cX_{\beta,m};x)=-\sin\big(\pi(m+\beta)\big).
\end{equation*}
Hence for $m+\beta\in\Z$ the solutions $\cK_{\beta,m}$ and $\cX_{\beta,m}$ are proportional to one another.
In fact, for such $\beta,m$, we have
\begin{equation*}
\cX_{\beta,m}(z)=\frac{\Gamma\big(\frac12-m-\beta\big)}{\Gamma\big(\frac12-m+\beta\big)} \cK_{\beta,m}(z).
  \end{equation*}
Note that this corresponds to the lines
$m+\beta=n\in\Z$. However in our applications, we need $\cX_{\beta,m}$ on the lines
$m+\beta-\frac12=n\in\Z$, where
 $\cK_{\beta,m}$ and $\cX_{\beta,m}$ are linearly independent.

\subsection{The Laguerre cases}

Let us now consider two special cases, namely when  $- \frac{1}{2}-m+\beta: = n\in\N$
and when $- \frac{1}{2}-m-\beta: = n\in\N$.
In the former case, observe that the Wronskian of
$\cI_{\beta,m}$  and  $\cK_{\beta,m}$ vanishes, see \eqref{Wr}.
It means that in such a case these two functions are proportional to one another.
In order to deal with this situation we define, for
$p\in\C$ and $n\in \N$, the {\em Laguerre polynomials}
by the formulas
\begin{align*}
L_n^{(p)}(z) &= \frac{z^{-p}\e^z}{n!} \frac{\d^n}{\d z^n}\big(\e^{-z}z^{p+n}\big) \\
& =\sum_{k=0}^n\frac{(p+k+1)_{n-k}(-z)^k}{(n-k)!k!} \\
&=\frac{(p+1)_n}{n!}{}_1F_1(-n;{p+1};z) \\
&={\frac{(-1)^n}{n!}z^n {}_2F_0(-n,-p-n;-;-z^{-1}).}
\end{align*}
Then, by setting $2m=p$, we get
\begin{equation*}
\cI_{\frac{1+p}{2}+n,\frac{p}{2}} = \frac{n!\;\!z^{\frac{1+p}{2}}\e^{-\frac{z}{2}}}{\Gamma(1+p+n)}L_n^{(p)}.
\end{equation*}
Note that this solution can also be expressed in terms of the $\cK_{\beta,m}$ function, namely
\begin{equation}\label{eq_bb}
\cK_{\frac{1+p}{2}+n,\frac{p}{2}} = (-1)^n n!\, z^{\frac{1+p}{2}}\e^{-\frac{z}{2}}L_n^{(p)}.
\end{equation}
We shall call this situation the {\em decaying Laguerre case}.
In this case the relation \eqref{eq_lin_comb} reduces to
\begin{equation}\label{eq_bbb}
\cI_{\frac{1+p}{2}+n,\frac{p}{2}} =
\frac{(-1)^n}{\Gamma(1+p+n)} \cK_{\frac{1+p}{2}+n,\frac{p}{2}},
\end{equation}
and more generally for $\ell\in \Z$ one has
\begin{equation*}
\cI_{\frac{1+p}{2}+\ell,\frac{p}{2}} =
\frac{(-1)^\ell}{\Gamma(1+p+\ell)} \cK_{\frac{1+p}{2}+\ell,\frac{p}{2}}
+ \frac{(-1)^{\ell+1}}{\cos(\pi p)\Gamma(-\ell)}  \cX_{\frac{1+p}{2}+\ell,\frac{p}{2}}.
\end{equation*}

In the special case $- \frac{1}{2}-m-\beta: = n\in\N$ a similar analysis with $p=2m$
leads to
\begin{equation*}
\cI_{-\frac{1+p}{2}-n,\frac{p}{2}}(z) = \frac{n!\;\!z^{\frac{1+p}{2}}\e^{\frac{z}{2}}}{\Gamma(1+p+n)}L_n^{(p)}(-z)
\end{equation*}
and to
\begin{equation}\label{eq_magic_3}
\cX_{-\frac{1+p}{2} -n, \frac{p}{2}}(z)=
\e^{\mp\i\frac{1+p}{2}\pi}\cK_{\frac{1+p}{2} +n,\frac{p}{2}}\big(\e^{\pm\i\pi}z\big)
= (-1)^nn!\;\! z^{\frac{1+p}{2}}\e^{\frac{z}{2}}L_n^{(p)}(-z).
\end{equation}
We shall call this situation the {\em exploding Laguerre case}.
In this case the relation \eqref{eq_lin_comb} reduces to
\begin{equation}\label{eq_aa}
\cI_{-\frac{1+p}{2}-n,\frac{p}{2}} = \frac{(-1)^{n}}{\Gamma(1+p+n)}
\cX_{-\frac{1+p}{2}-n,\frac{p}{2}},
\end{equation}
and more generally for $\ell\in \Z$ one has
\begin{equation*}
\cI_{-\frac{1+p}{2}-\ell,\frac{p}{2}} =
\frac{(-1)^{\ell+1}\cos(\pi p)}{\Gamma(-\ell)} \cK_{-\frac{1+p}{2}-\ell,\frac{p}{2}}
+\frac{(-1)^\ell}{\Gamma(1+p+\ell)}  \cX_{-\frac{1+p}{2}-\ell,\frac{p}{2}}.
\end{equation*}

\subsection{The degenerate case}\label{The degenerate case}

In this section we consider the special case $m\in\frac12\Z$, which will be called {\em the degenerate case}, see Figure \ref{fig:M1}. In this situation the Wronskian of  $\cI_{\beta,m}$ and  $\cI_{\beta,-m}$ vanishes, see \eqref{wro1}. More precisely, for any $p\in \N$ one has the identity
\begin{equation*}
\cI_{\beta,-\frac{p}{2}}= \Big(-\beta - \frac{p-1}{2}\Big)_{p}\, \cI_{\beta,\frac{p}{2}},
\end{equation*}
or equivalently,
\begin{equation*}
\frac{1} {\Gamma\big(\frac{1+p}{2}-\beta\big)}\I{\beta}{-\frac{p}{2}}
= \frac{1} {\Gamma\big(\frac{1-p}{2}-\beta\big)} \cI_{\beta,\frac{p}{2}}.
\end{equation*}

Based on this equality and by a limiting procedure,
one can provide an expression for the functions $\Ka_{\beta,\frac{p}{2}}$
(see \cite[Thm.~2.2]{DR2}), namely
\begin{align}\label{eq_K_p/2}
\begin{split}
\cK_{\beta,\frac{p}{2}}(z)
& = \frac{(-1)^{p+1}\ln(z)\;\!\cI_{\beta,\frac{p}{2}}(z)}
{\Gamma\big(\tfrac{1-p}{2}-\beta\big)} \\
&+\frac{(-1)^{p+1}\e^{-\frac{z}{2}}z^{\frac{1+p}{2}}}{\Gamma\big(\frac{1-p}{2}-\beta\big)} \sum_{k=0}^\infty
\frac{\big(\frac{1+p}{2}-\beta\big)_k\;\!z^{k}}{(p+k)!\;\!k!}\\
&\quad \times \Big(\psi\big(\tfrac{1+p}{2}-\beta+k\big)-\psi(p+1+k)
-\psi(1+k)\Big)\\
&+\frac{(-1)^{p+1}\e^{-\frac{z}{2}}z^{\frac{1+p}{2}}}{\Gamma\big(\frac{1-p}{2}-\beta\big)}
\sum_{j=1}^{p}
\frac{\big(\frac{1+p}{2}-\beta\big)_{-j}(-1)^{j-1}(j-1)!z^{-j}}{(p-j)!},
\end{split}
\end{align}
where $\psi$ is the digamma function defined by $\psi(z)=\frac{\Gamma'(z)}{\Gamma(z)}$.
Note that the equality (or definition) $(a)_j=\frac{\Gamma(a+j)}{\Gamma(a)}$ has also been used for arbitrary $j\in \Z$.
For our applications the most important functions correspond to $m=\frac12$ and $m=0$:
\begin{align}\label{eq_a12}
\cK_{\beta,\frac{1}{2}}(z)
= & \frac{\ln(z)\;\!\cI_{\beta,\frac{1}{2}}(z)}
{\Gamma\big(-\beta\big)} +\frac{\e^{-\frac{z}{2}}}{\Gamma(1-\beta)} \\
&+\frac{\e^{-\frac{z}{2}}}{\Gamma\big(-\beta\big)} \sum_{k=0}^\infty
\frac{\big(1-\beta\big)_k\;\!z^{1+k}}{(1+k)!\;\!k!}  \Big(\psi\big(1-\beta+k\big)-\psi(2+k)
-\psi(1+k)\Big),\notag\\
  \label{eq_a0}
\cK_{\beta,0}(z) = & -\frac{\ln(z)\;\! \cI_{\beta,0}(z)}{\Gamma\big(\tfrac{1}{2}-\beta\big)} \\
& - \frac{\e^{-\frac{z}{2}}}{\Gamma\big(\frac{1}{2}-\beta\big)} \sum_{k=0}^\infty
\frac{\big(\frac12-\beta\big)_k\;\!z^{\frac12+k}}{(k!)^2}
\Big(\psi\big(\tfrac12-\beta+k\big)-2\psi(1+k)\Big).\notag
\end{align}

Let us still provide the expression for the function $\cX_{\beta, \frac{p}{2}}$.
Starting from its definition in \eqref{eq_def_X} and by using the expansion
\eqref{eq_K_p/2} as well as the identity provided in \eqref{eq_miracle}
one gets
\begin{align*}
\cX_{\beta,\frac{p}{2}}(z)
= &\frac{(-1)^{p+1}\ln(z)\;\!\cI_{\beta,\frac{p}{2}}(z)}
{\Gamma\big(\tfrac{1-p}{2}+\beta\big)}\\
&+\frac{(-1)^{p+1}\e^{\frac{z}{2}} z^{\frac{1+p}{2}}}{\Gamma\big(\frac{1-p}{2}+\beta\big)} \sum_{k=0}^\infty
\frac{\big(\frac{1+p}{2}+\beta\big)_k\;\!(-1)^k z^{k}}{(p+k)!\;\!k!}\\
&\quad \times \Big(\psi\big(\tfrac{1+p}{2}+\beta+k\big)-\psi(p+1+k)
-\psi(1+k)\Big)\\
&-\frac{(-1)^{p+1}\e^{\frac{z}{2}}z^{\frac{p+1}{2}}}{\Gamma\big(\frac{1-p}{2}+\beta\big)}\sum_{j=1}^{p}
\frac{\big(\frac{1+p}{2}+\beta\big)_{-j}(j-1)!z^{-j}}{(p-j)!}.
\end{align*}
In particular, the  expansions for  $m=\frac12$ and $m=0$ will be useful:
\begin{align}\label{exp_X12}
\cX_{\beta,\frac{1}{2}}(z)=&
 -\frac{1}{\Gamma(1+\beta)}+\frac{1}{\Gamma(\beta)}z\ln(z)\\
&+ \frac{1}{\Gamma(\beta)}\Big(\frac12\psi(1+\beta)+\frac12\psi(\beta)+2\gamma-1\Big)z + o(z)\notag\\
\label{exp_X0}
 \cX_{\beta,0}(z)  =&
-\frac{z^{\frac{1}{2}}}{\Gamma\big(\frac{1}{2}+\beta\big)}
\Big[(1-\beta z)\ln(z) +\Big(\psi\big(\tfrac{1}{2}+\beta\big)+2\gamma\Big) \\
& \quad -\beta\Big(
\psi\Big(\frac{1}{2}+\beta\Big)+2\gamma-2 \Big)z
\Big]
+ o\big(z^{\frac{3}{2}}\big).\notag
\end{align}

Note also that the following identity holds:
\begin{equation*}
\cX_{\beta,-\frac{p}{2}}=(-1)^p    \cX_{\beta,\frac{p}{2}},
\end{equation*}
as a consequence of \eqref{eq_def_X}.

\begin{figure}[H]
\centering
\begin{tikzpicture}
\setlength{\unitlength}{1cm}

\draw[thick, ->] (-5,0)--(5.3,0);
\draw[thick, ->] (0,-5)--(0,5.3);
\draw[blue, thin] (-5,-4.5)--(4.5,5);
\draw[blue, thin] (-5,-3.5)--(3.5,5);
\draw[blue, thin] (-5,-2.5)--(2.5,5);
\draw[blue, thin] (-5,-1.5)--(1.5,5);
\draw[blue, thin] (-5,-0.5)--(0.5,5);
\draw[blue, thin] (-5,0.5)--(-0.5,5);
\draw[blue, thin] (-5,1.5)--(-1.5,5);
\draw[blue, thin] (-5,2.5)--(-2.5,5);
\draw[blue, thin] (-5,3.5)--(-3.5,5);
\draw[blue, thin] (-5,4.5)--(-4.5,5);

\draw[green, thin] (-5,4.5)--(4.5,-5);
\draw[green, thin] (-5,3.5)--(3.5,-5);
\draw[green, thin] (-5,2.5)--(2.5,-5);
\draw[green, thin] (-5,1.5)--(1.5,-5);
\draw[green, thin] (-5,0.5)--(0.5,-5);
\draw[green, thin] (-5,-0.5)--(-0.5,-5);
\draw[green, thin] (-5,-1.5)--(-1.5,-5);
\draw[green, thin] (-5,-2.5)--(-2.5,-5);
\draw[green, thin] (-5,-3.5)--(-3.5,-5);
\draw[green, thin] (-5,-4.5)--(-4.5,-5);

\draw[red, thin] (-4.5,-5)--(-4.5,5);
\draw[red, thin] (-4,-5)--(-4,5);
\draw[red, thin] (-3.5,-5)--(-3.5,5);
\draw[red, thin] (-3,-5)--(-3,5);
\draw[red, thin] (-2.5,-5)--(-2.5,5);
\draw[red, thin] (-2,-5)--(-2,5);
\draw[red, thin] (-1.5,-5)--(-1.5,5);
\draw[red, thin] (-1,-5)--(-1,5);
\draw[red, thin] (-0.5,-5)--(-0.5,5);
\draw[red, thin] (0,-5)--(0,5);
\draw[red, thin] (5,-5)--(5,5);
\draw[red, thin] (4.5,-5)--(4.5,5);
\draw[red, thin] (4,-5)--(4,5);
\draw[red, thin] (3.5,-5)--(3.5,5);
\draw[red, thin] (3,-5)--(3,5);
\draw[red, thin] (2.5,-5)--(2.5,5);
\draw[red, thin] (2,-5)--(2,5);
\draw[red, thin] (1.5,-5)--(1.5,5);
\draw[red, thin] (1,-5)--(1,5);
\draw[red, thin] (0.5,-5)--(0.5,5);

\put(0.1,5){$\beta$};
\put(5,0.1){$m$};
\put(-4.1,0.1){$\scriptstyle{-}4$};
\put(-3.1,0.1){$\scriptstyle{-}3$};
\put(-2.1,0.1){$\scriptstyle{-}2$};
\put(-1.1,0.1){$\scriptstyle{-}1$};

\put(0.1,0.1){$\scriptstyle0$};
\put(1.1,0.1){$\scriptstyle1$};
\put(2.1,0.1){$\scriptstyle2$};
\put(3.1,0.1){$\scriptstyle3$};
\put(4.1,0.1){$\scriptstyle4$};

\filldraw[gray](-4,0)circle(0.06);
\filldraw[gray](-3,0)circle(0.06);
\filldraw[gray](-2,0)circle(0.06);
\filldraw[gray](-1,0)circle(0.06);
\filldraw[gray](0,0)circle(0.06);
\filldraw[gray](1,0)circle(0.06);
\filldraw[gray](2,0)circle(0.06);
\filldraw[gray](3,0)circle(0.06);
\filldraw[gray](4,0)circle(0.06);

\put(0.1,-4.07){$\scriptstyle{-}4$};
\put(0.1,-3.07){$\scriptstyle{-}3$};
\put(0.1,-2.07){$\scriptstyle{-}2$};
\put(0.1,-1.07){$\scriptstyle{-}1$};
\put(0.1,3.92){$\scriptstyle4$};
\put(0.1,2.92){$\scriptstyle3$};
\put(0.1,1.92){$\scriptstyle2$};
\put(0.1,0.92){$\scriptstyle1$};

\put(1.1,0.1){$\scriptstyle1$};
\put(2.1,0.1){$\scriptstyle2$};
\put(3.1,0.1){$\scriptstyle3$};
\put(4.1,0.1){$\scriptstyle4$};

\filldraw[gray](0,-4)circle(0.06);
\filldraw[gray](0,-3)circle(0.06);
\filldraw[gray](0,-2)circle(0.06);
\filldraw[gray](0,-1)circle(0.06);
\filldraw[gray](0,0)circle(0.06);
\filldraw[gray](0,1)circle(0.06);
\filldraw[gray](0,2)circle(0.06);
\filldraw[gray](0,3)circle(0.06);
\filldraw[gray](0,4)circle(0.06);

\end{tikzpicture}
\caption{The vertical lines correspond to the degenerate cases,
the lines with slope $1$ to the decaying Laguerre case,
the lines with slope $-1$ with the exploding Laguerre case.}
\label{fig:M1}
\end{figure}
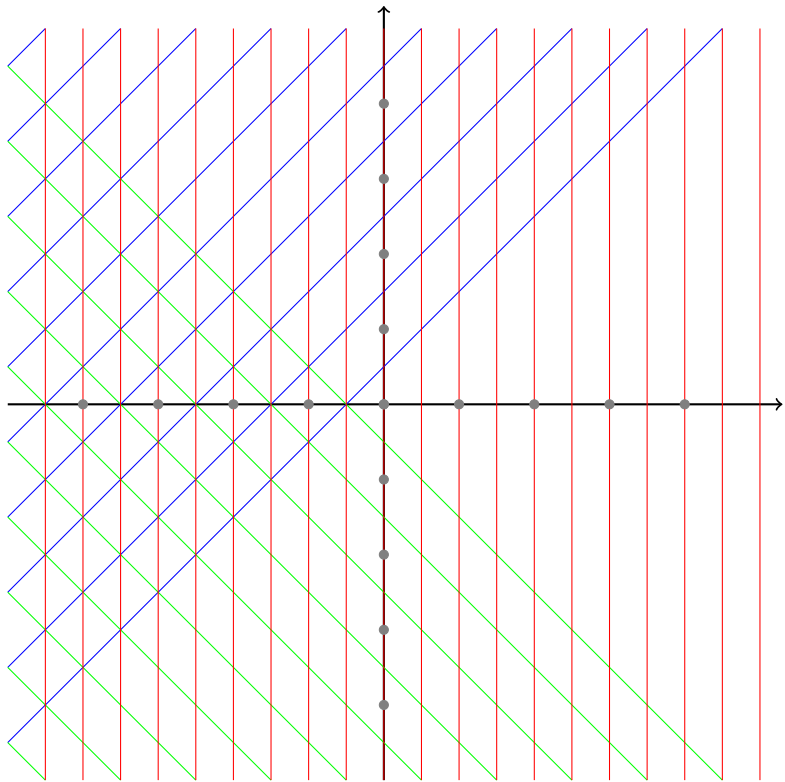

\subsection{The doubly degenerate case}\label{sec_dd}

We shall now consider the region
\begin{equation}\label{double}
\Big\{(m,\beta)\mid\beta\in\tfrac12\Z,\
m \in\tfrac12\Z,\ \beta+m+\tfrac12\in\Z\Big\}.
\end{equation}
In other words, we consider
$m\in \Z$, $\beta \in\Z+ \frac{1}{2}$, or  $m\in \Z+\frac{1}{2}$, $\beta \in \Z$.
This situation will be called the {\em doubly degenerate case}.
We will again set $m=\frac{p}{2}$ with $p\in \Z$.
Note that for $(m,\beta)$ in
\eqref{double} we have the identity
\begin{equation}\label{eq_gen3}
\cI_{\beta,m}=
\frac{  (-1)^{\beta+m+\frac32+p}}
{\Gamma\big(\frac12+m+\beta\big)}\cK_{\beta,m}+
\frac{  (-1)^{\beta+m+\frac12}}
{\Gamma\big(\frac12+m-\beta\big)}\cX_{\beta,m},
\end{equation}
which is a special case of \eqref{eq_lin_comb}.
In this case we also have
\begin{equation}\label{eq_Wr_2d}
\Wr(\cK_{\beta,m} , \cX_{\beta,m};x)=(-1)^{m+\beta+\frac12}.
\end{equation}
Hence $\cK_{\beta,m}$ and $\cX_{\beta,m}$ always span the space of solutions
in the doubly degenerate case.

In order to analyze  the doubly degenerate case more precisely, let us divide  \eqref{double} into 4 distinct regions (see Figure \ref{fig:M2}).

\medskip
\noindent{\bf Region $\mathrm{I}_-$.} $\beta+m\in-\big(\N+\frac12\big),\quad -\beta+m\in-\big(\N+\frac12\big).$

\medskip
\noindent We have
\begin{equation*}
\cI_{\beta,m}=0,
\end{equation*}
which follows for example from \eqref{eq_gen3}.
By setting $n_1:=\beta-m-\frac12\in\N$ and $n_2=-\beta-m-\frac12\in\N$, then
$\cK_{\beta,m}=\cK_{\frac{1+p}{2}+n_1,\frac{p}{2}}$ is the decaying Laguerre solution, see \eqref{eq_bb},
and $\cX_{\beta,m}=\cX_{-\frac{1+p}{2}-n_2,\frac{p}{2}}$ is the exploding Laguerre solution, see  \eqref{eq_magic_3}.

\medskip
\noindent{\bf Region $\mathrm{I}_+$.} $\beta+m\in\N+\frac12,\quad -\beta+m\in\N+\frac12.$

\medskip
\noindent
First note that $(m,\beta)\in \mathrm{I}_-$ if and only if $(-m,\beta)\in \mathrm{I}_+$.
By setting $n_1:=\beta+m-\frac12\in\N$ and $n_2:=-\beta+m-\frac12\in\N$, one has
$\beta=\frac{n_1-n_2}{2}$, $m=\frac{n_1+n_2+1}{2}$, and the equality
\eqref{eq_gen3} can be rewritten as
\begin{equation*}
\cI_{\beta,m}= \frac{  (-1)^{n_2+1}}
{n_1!}\cK_{\beta,m}+\frac{  (-1)^{n_1+1}}{n_2!}\cX_{\beta,m}.
\end{equation*}
Note then that $\cK_{\beta,m}=\cK_{\frac{1-p}{2}+n_1,\frac{p}{2}}=\cK_{\frac{1-p}{2}+n_1,\frac{-p}{2}}$
corresponds to the decaying Laguerre solution, while
$\cX_{\beta,m}=\cX_{-\frac{1-p}{2}-n_2,\frac{p}{2}}
=(-1)^p\cX_{-\frac{1-p}{2}-n_2,-\frac{p}{2}}
=(-1)^p\cX_{-\frac{1-p}{2}-n_2,\frac{-p}{2}}$ corresponds to
the exploding Laguerre solution.
In this region, the space of solutions can also be spanned by the pair
$\cK_{\beta,m}$ and $\cI_{\beta,m}$, or by
the pair $\cI_{\beta,m}$ and $\cX_{\beta,m}$.

\medskip
\noindent{\bf Region $\mathrm{II}_-$.} $\beta+m\in-\big( \N+\frac12\big),\quad -\beta+m\in\N+\frac12.$

\medskip
\noindent
By setting $n:=-\beta-m-\frac12\in\N$, then the equality \eqref{eq_aa}
reduces to
\begin{equation*}
\cI_{-\frac{p+1}{2}-n,\frac{p}{2}}=
\frac{  (-1)^n}{(p+n)!}\cX_{-\frac{p+1}{2}-n,\frac{p}{2}}.
\end{equation*}
Thus  $\cI_{\beta,m}$ is proportional to
$\cX_{\beta,m}$
and  corresponds to the exploding Laguerre case.
The second solution is  $\cK_{\beta,m}$. It decays exponentially
and has a logarithmic singularity at zero, therefore we
call this function the {\em decaying logarithmic solution}.

\medskip
\noindent{\bf Region $\mathrm{II}_+$.} $\beta+m\in \N+\frac12,\quad -\beta+m\in-\big(\N+\frac12\big).$

\medskip
\noindent
By setting $n:=\beta-m-\frac12\in\N$, then the equality \eqref{eq_bbb}
reduces to
\begin{equation*}
  \cI_{\frac{p+1}{2}+n,\frac{p}{2}}=
\frac{  (-1)^n}{(p+n)!}\cK_{\frac{p+1}{2}+n,\frac{p}{2}}.
\end{equation*}
Thus  $\cI_{\beta,m}$ is proportional to
 $\cK_{\beta,m}$
 and  corresponds to the decaying Laguerre case.
The second solution is  $\cX_{\beta,m}$. It explodes exponentially
and has a logarithmic singularity at zero, therefore
we call this function the {\em exploding logarithmic solution}.

The results of this section are summarized in Figure \ref{fig:M2}.

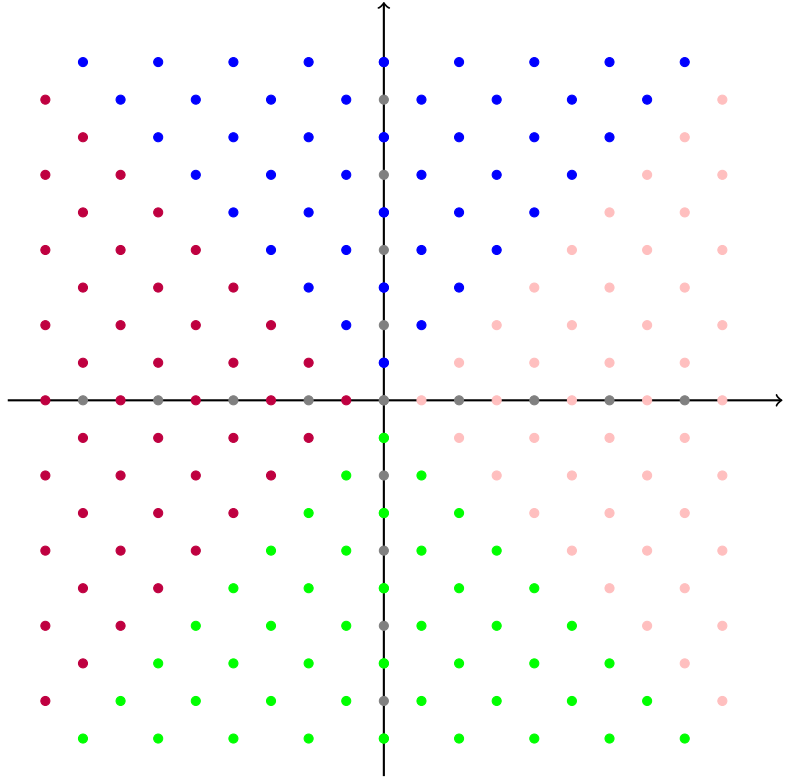
\begin{figure}[H]
\centering
\begin{tikzpicture}
\setlength{\unitlength}{1cm}

\draw[thick, ->] (-5,0)--(5.3,0);
\draw[thick, ->] (0,-5)--(0,5.3);

\put(0.1,5){$\beta$};
\put(5,0.1){$m$};
\put(-5.1,1.5){$\mathrm{I}_-$}
\put(4.9,1.5){$\mathrm{I}_+$}
\put(1.3,5){$\mathrm{II}_+$}
\put(1.3,-5){$\mathrm{II}_-$}

\filldraw[purple](-4.5,4)circle(0.06);
\filldraw[purple](-4.5,3)circle(0.06);
\filldraw[purple](-4.5,2)circle(0.06);
\filldraw[purple](-4.5,1)circle(0.06);
\filldraw[purple](-4.5,0)circle(0.06);
\filldraw[purple](-4.5,-1)circle(0.06);
\filldraw[purple](-4.5,-2)circle(0.06);
\filldraw[purple](-4.5,-3)circle(0.06);
\filldraw[purple](-4.5,-4)circle(0.06);
\filldraw[blue](-3.5,4)circle(0.06);
\filldraw[purple](-3.5,3)circle(0.06);
\filldraw[purple](-3.5,2)circle(0.06);
\filldraw[purple](-3.5,1)circle(0.06);
\filldraw[purple](-3.5,0)circle(0.06);
\filldraw[purple](-3.5,-1)circle(0.06);
\filldraw[purple](-3.5,-2)circle(0.06);
\filldraw[purple](-3.5,-3)circle(0.06);
\filldraw[green](-3.5,-4)circle(0.06);
\filldraw[blue](-2.5,4)circle(0.06);
\filldraw[blue](-2.5,3)circle(0.06);
\filldraw[purple](-2.5,2)circle(0.06);
\filldraw[purple](-2.5,1)circle(0.06);
\filldraw[purple](-2.5,0)circle(0.06);
\filldraw[purple](-2.5,-1)circle(0.06);
\filldraw[purple](-2.5,-2)circle(0.06);
\filldraw[green](-2.5,-3)circle(0.06);
\filldraw[green](-2.5,-4)circle(0.06);
\filldraw[blue](-1.5,4)circle(0.06);
\filldraw[blue](-1.5,3)circle(0.06);
\filldraw[blue](-1.5,2)circle(0.06);
\filldraw[purple](-1.5,1)circle(0.06);
\filldraw[purple](-1.5,0)circle(0.06);
\filldraw[purple](-1.5,-1)circle(0.06);
\filldraw[green](-1.5,-2)circle(0.06);
\filldraw[green](-1.5,-3)circle(0.06);
\filldraw[green](-1.5,-4)circle(0.06);
\filldraw[blue](-0.5,4)circle(0.06);
\filldraw[blue](-0.5,3)circle(0.06);
\filldraw[blue](-0.5,2)circle(0.06);
\filldraw[blue](-0.5,1)circle(0.06);
\filldraw[purple](-0.5,0)circle(0.06);
\filldraw[green](-0.5,-1)circle(0.06);
\filldraw[green](-0.5,-2)circle(0.06);
\filldraw[green](-0.5,-3)circle(0.06);
\filldraw[green](-0.5,-4)circle(0.06);
\filldraw[pink](4.5,4)circle(0.06);
\filldraw[pink](4.5,3)circle(0.06);
\filldraw[pink](4.5,2)circle(0.06);
\filldraw[pink](4.5,1)circle(0.06);
\filldraw[pink](4.5,0)circle(0.06);
\filldraw[pink](4.5,-1)circle(0.06);
\filldraw[pink](4.5,-2)circle(0.06);
\filldraw[pink](4.5,-3)circle(0.06);
\filldraw[pink](4.5,-4)circle(0.06);
\filldraw[blue](3.5,4)circle(0.06);
\filldraw[pink](3.5,3)circle(0.06);
\filldraw[pink](3.5,2)circle(0.06);
\filldraw[pink](3.5,1)circle(0.06);
\filldraw[pink](3.5,0)circle(0.06);
\filldraw[pink](3.5,-1)circle(0.06);
\filldraw[pink](3.5,-2)circle(0.06);
\filldraw[pink](3.5,-3)circle(0.06);
\filldraw[green](3.5,-4)circle(0.06);
\filldraw[blue](2.5,4)circle(0.06);
\filldraw[blue](2.5,3)circle(0.06);
\filldraw[pink](2.5,2)circle(0.06);
\filldraw[pink](2.5,1)circle(0.06);
\filldraw[pink](2.5,0)circle(0.06);
\filldraw[pink](2.5,-1)circle(0.06);
\filldraw[pink](2.5,-2)circle(0.06);
\filldraw[green](2.5,-3)circle(0.06);
\filldraw[green](2.5,-4)circle(0.06);
\filldraw[blue](1.5,4)circle(0.06);
\filldraw[blue](1.5,3)circle(0.06);
\filldraw[blue](1.5,2)circle(0.06);
\filldraw[pink](1.5,1)circle(0.06);
\filldraw[pink](1.5,0)circle(0.06);
\filldraw[pink](1.5,-1)circle(0.06);
\filldraw[green](1.5,-2)circle(0.06);
\filldraw[green](1.5,-3)circle(0.06);
\filldraw[green](1.5,-4)circle(0.06);
\filldraw[blue](0.5,4)circle(0.06);
\filldraw[blue](0.5,3)circle(0.06);
\filldraw[blue](0.5,2)circle(0.06);
\filldraw[blue](0.5,1)circle(0.06);
\filldraw[pink](0.5,0)circle(0.06);
\filldraw[green](0.5,-1)circle(0.06);
\filldraw[green](0.5,-2)circle(0.06);
\filldraw[green](0.5,-3)circle(0.06);
\filldraw[green](0.5,-4)circle(0.06);

\filldraw[blue](-4,4.5)circle(0.06);
\filldraw[purple](-4,3.5)circle(0.06);
\filldraw[purple](-4,2.5)circle(0.06);
\filldraw[purple](-4,1.5)circle(0.06);
\filldraw[purple](-4,0.5)circle(0.06);
\filldraw[purple](-4,-0.5)circle(0.06);
\filldraw[purple](-4,-1.5)circle(0.06);
\filldraw[purple](-4,-2.5)circle(0.06);
\filldraw[purple](-4,-3.5)circle(0.06);
\filldraw[green](-4,-4.5)circle(0.06);
\filldraw[blue](-3,4.5)circle(0.06);
\filldraw[blue](-3,3.5)circle(0.06);
\filldraw[purple](-3,2.5)circle(0.06);
\filldraw[purple](-3,1.5)circle(0.06);
\filldraw[purple](-3,0.5)circle(0.06);
\filldraw[purple](-3,-0.5)circle(0.06);
\filldraw[purple](-3,-1.5)circle(0.06);
\filldraw[purple](-3,-2.5)circle(0.06);
\filldraw[green](-3,-3.5)circle(0.06);
\filldraw[green](-3,-4.5)circle(0.06);
\filldraw[blue](-2,4.5)circle(0.06);
\filldraw[blue](-2,3.5)circle(0.06);
\filldraw[blue](-2,2.5)circle(0.06);
\filldraw[purple](-2,1.5)circle(0.06);
\filldraw[purple](-2,0.5)circle(0.06);
\filldraw[purple](-2,-0.5)circle(0.06);
\filldraw[purple](-2,-1.5)circle(0.06);
\filldraw[green](-2,-2.5)circle(0.06);
\filldraw[green](-2,-3.5)circle(0.06);
\filldraw[green](-2,-4.5)circle(0.06);
\filldraw[blue](-1,4.5)circle(0.06);
\filldraw[blue](-1,3.5)circle(0.06);
\filldraw[blue](-1,2.5)circle(0.06);
\filldraw[blue](-1,1.5)circle(0.06);
\filldraw[purple](-1,0.5)circle(0.06);
\filldraw[purple](-1,-0.5)circle(0.06);
\filldraw[green](-1,-1.5)circle(0.06);
\filldraw[green](-1,-2.5)circle(0.06);
\filldraw[green](-1,-3.5)circle(0.06);
\filldraw[green](-1,-4.5)circle(0.06);
\filldraw[blue](-0,4.5)circle(0.06);
\filldraw[blue](-0,3.5)circle(0.06);
\filldraw[blue](-0,2.5)circle(0.06);
\filldraw[blue](-0,1.5)circle(0.06);
\filldraw[blue](-0,0.5)circle(0.06);
\filldraw[green](-0,-0.5)circle(0.06);
\filldraw[green](-0,-1.5)circle(0.06);
\filldraw[green](-0,-2.5)circle(0.06);
\filldraw[green](-0,-3.5)circle(0.06);
\filldraw[green](-0,-4.5)circle(0.06);
\filldraw[blue](4,4.5)circle(0.06);
\filldraw[pink](4,3.5)circle(0.06);
\filldraw[pink](4,2.5)circle(0.06);
\filldraw[pink](4,1.5)circle(0.06);
\filldraw[pink](4,0.5)circle(0.06);
\filldraw[pink](4,-0.5)circle(0.06);
\filldraw[pink](4,-1.5)circle(0.06);
\filldraw[pink](4,-2.5)circle(0.06);
\filldraw[pink](4,-3.5)circle(0.06);
\filldraw[green](4,-4.5)circle(0.06);
\filldraw[blue](3,4.5)circle(0.06);
\filldraw[blue](3,3.5)circle(0.06);
\filldraw[pink](3,2.5)circle(0.06);
\filldraw[pink](3,1.5)circle(0.06);
\filldraw[pink](3,0.5)circle(0.06);
\filldraw[pink](3,-0.5)circle(0.06);
\filldraw[pink](3,-1.5)circle(0.06);
\filldraw[pink](3,-2.5)circle(0.06);
\filldraw[green](3,-3.5)circle(0.06);
\filldraw[green](3,-4.5)circle(0.06);
\filldraw[blue](2,4.5)circle(0.06);
\filldraw[blue](2,3.5)circle(0.06);
\filldraw[blue](2,2.5)circle(0.06);
\filldraw[pink](2,1.5)circle(0.06);
\filldraw[pink](2,0.5)circle(0.06);
\filldraw[pink](2,-0.5)circle(0.06);
\filldraw[pink](2,-1.5)circle(0.06);
\filldraw[green](2,-2.5)circle(0.06);
\filldraw[green](2,-3.5)circle(0.06);
\filldraw[green](2,-4.5)circle(0.06);
\filldraw[blue](1,4.5)circle(0.06);
\filldraw[blue](1,3.5)circle(0.06);
\filldraw[blue](1,2.5)circle(0.06);
\filldraw[blue](1,1.5)circle(0.06);
\filldraw[pink](1,0.5)circle(0.06);
\filldraw[pink](1,-0.5)circle(0.06);
\filldraw[green](1,-1.5)circle(0.06);
\filldraw[green](1,-2.5)circle(0.06);
\filldraw[green](1,-3.5)circle(0.06);
\filldraw[green](1,-4.5)circle(0.06);
\filldraw[blue](0,4.5)circle(0.06);
\filldraw[blue](0,3.5)circle(0.06);
\filldraw[blue](0,2.5)circle(0.06);
\filldraw[blue](0,1.5)circle(0.06);
\filldraw[blue](0,0.5)circle(0.06);
\filldraw[green](0,-0.5)circle(0.06);
\filldraw[green](0,-1.5)circle(0.06);
\filldraw[green](0,-2.5)circle(0.06);
\filldraw[green](0,-3.5)circle(0.06);
\filldraw[green](0,-4.5)circle(0.06);

\put(-4.1,0.1){$\scriptstyle{-}4$};
\put(-3.1,0.1){$\scriptstyle{-}3$};
\put(-2.1,0.1){$\scriptstyle{-}2$};
\put(-1.1,0.1){$\scriptstyle{-}1$};

\put(0.1,0.1){$\scriptstyle0$};
\put(1.1,0.1){$\scriptstyle1$};
\put(2.1,0.1){$\scriptstyle2$};
\put(3.1,0.1){$\scriptstyle3$};
\put(4.1,0.1){$\scriptstyle4$};

\filldraw[gray](-4,0)circle(0.06);
\filldraw[gray](-3,0)circle(0.06);
\filldraw[gray](-2,0)circle(0.06);
\filldraw[gray](-1,0)circle(0.06);
\filldraw[gray](0,0)circle(0.06);
\filldraw[gray](1,0)circle(0.06);
\filldraw[gray](2,0)circle(0.06);
\filldraw[gray](3,0)circle(0.06);
\filldraw[gray](4,0)circle(0.06);

\put(0.1,-4.07){$\scriptstyle{-}4$};
\put(0.1,-3.07){$\scriptstyle{-}3$};
\put(0.1,-2.07){$\scriptstyle{-}2$};
\put(0.1,-1.07){$\scriptstyle{-}1$};
\put(0.1,3.92){$\scriptstyle4$};
\put(0.1,2.92){$\scriptstyle3$};
\put(0.1,1.92){$\scriptstyle2$};
\put(0.1,0.92){$\scriptstyle1$};

\put(1.1,0.1){$\scriptstyle1$};
\put(2.1,0.1){$\scriptstyle2$};
\put(3.1,0.1){$\scriptstyle3$};
\put(4.1,0.1){$\scriptstyle4$};

\filldraw[gray](0,-4)circle(0.06);
\filldraw[gray](0,-3)circle(0.06);
\filldraw[gray](0,-2)circle(0.06);
\filldraw[gray](0,-1)circle(0.06);
\filldraw[gray](0,0)circle(0.06);
\filldraw[gray](0,1)circle(0.06);
\filldraw[gray](0,2)circle(0.06);
\filldraw[gray](0,3)circle(0.06);
\filldraw[gray](0,4)circle(0.06);

\end{tikzpicture}
\caption{Solutions for the doubly degenerate case:
Region $\mathrm{I}_-$:  the decaying Laguerre and the exploding Laguerre solutions.
Region $\mathrm{I}_+$: any of the three solutions.
Region $\mathrm{II}_+$: the decaying Laguerre and the exploding logarithmic solutions.
Region $\mathrm{II}_-$: the exploding Laguerre and the decaying logarithmic solutions.
}
\label{fig:M2}
\end{figure}

\subsection{Recurrence relations}

Solutions of the Whittaker equation satisfy interesting recurrence relations.
These relations can be checked by using the series provided in \eqref{Taylor_2}.
The computations are straightforward, but rather lengthy.
These relations read
\begin{align*}
\Big(\sqrt{z}\partial_z+\frac{-\frac12-m}{\sqrt{z}}-\frac {\sqrt{z}}{2}\Big)\cI_{\beta,m}(z)
&= \Big(-\frac12-m-\beta\Big)\cI_{\beta+\frac12,m+\frac12}(z),\\
\Big(\sqrt{z}\partial_z+\frac{-\frac12+m}{\sqrt{z}}+\frac {\sqrt{z}}{2}\Big)\cI_{\beta,m}(z)
&=\cI_{\beta-\frac12,m-\frac12}(z),\\
\Big(\sqrt{z}\partial_z+\frac{-\frac12+m}{\sqrt{z}}-\frac {\sqrt{z}}{2}\Big)\cI_{\beta,m}(z)
&=\cI_{\beta+\frac12,m-\frac12}(z),\\
\Big(\sqrt{z}\partial_z+\frac{-\frac12-m}{\sqrt{z}}+\frac {\sqrt{z}}{2}\Big)\cI_{\beta,m}(z)
&=\Big(\frac12+m-\beta\Big)\cI_{\beta-\frac12,m+\frac12}(z),\\
\Big(z\partial_z+\beta-\frac{z}{2}\Big)\cI_{\beta,m}(z)
&=\Big(\frac12+m+\beta\Big)\cI_{\beta+1,m}(z),\\
\Big(z\partial_z-\beta+\frac{z}{2}\Big)\cI_{\beta,m}(z)
&=\Big(\frac12+m-\beta\Big)\cI_{\beta-1,m}(z).
\end{align*}
By using the relation between the functions $\cK_{\beta,m}$ and the functions
$\cI_{\beta,m}$ provided in \eqref{generic}, one infers from the above
relations the following ones:
\begin{align*}
\Big(\sqrt{z}\partial_z+\frac{-\frac12-m}{\sqrt{z}}-\frac {\sqrt{z}}{2}\Big)
\cK_{\beta,m}(z)
&= -\cK_{\beta+\frac12,m+\frac12}(z),\\
\Big(\sqrt{z}\partial_z+\frac{-\frac12+m}{\sqrt{z}}+\frac {\sqrt{z}}{2}\Big)
\cK_{\beta,m}(z)
&= \Big(-\frac12+m+\beta\Big)    \cK_{\beta-\frac12,m-\frac12}(z),\\
\Big(\sqrt{z}\partial_z+\frac{-\frac12+m}{\sqrt{z}}-\frac {\sqrt{z}}{2}\Big)
\cK_{\beta,m}(z)
&= -\cK_{\beta+\frac12,m-\frac12}(z),\\
\Big(\sqrt{z}\partial_z+\frac{-\frac12-m}{\sqrt{z}}+\frac {\sqrt{z}}{2}\Big)
\cK_{\beta,m}(z)
&= \Big(-\frac12-m+\beta\Big)\cK_{\beta-\frac12,m+\frac12}(z),\\
\Big(z\partial_z+\beta-\frac{z}{2}\Big)\cK_{\beta,m}(z)
&= -\cK_{\beta+1,m}(z),\\
\Big(z\partial_z-\beta+\frac{z}{2}\Big)\cK_{\beta,m}(z)
&= \Big(\frac12+m-\beta\Big)    \Big(\frac12-m-\beta\Big)\cK_{\beta-1,m}(z).
\end{align*}

\subsection{Integral identities}\label{sec_int_id}

Let us start with a general fact about 1-dimensional Schr\"odinger operators,
see for example \cite[Eq.~(3.24)]{DG2019}.

\begin{lemma}\label{lem_int_id}
For $i\in \{1,2\}$, suppose that $v_i\in \Dom(L_{\beta,\alpha}^\max)$ satisfies $L_{\beta,\alpha} v_i=\lambda_i v_i$ for some $\lambda_i\in \C$. Then, for all $a,b \in ]0,\infty[$,
\begin{equation}\label{lagrange}
(\lambda_1-\lambda_2)\int_a^b v_1(x)v_2(x)\d x= \Wr(v_1,v_2;b)-  \Wr(v_1,v_2;a),
\end{equation}
where $\Wr$ is the Wronskian introduced in \eqref{wronskian}.
\end{lemma}

As a consequence of this lemma one has:

\begin{proposition}\label{prop-lag1}
Let $k,p\in \C$ with $\Re(k)>0$ and $\Re(p)>0$.
\begin{enumerate}\item[(i)] If $-1<\Re(m)<1$, $m\not\in\big\{-\frac12,0,\frac12\big\}$, then
\begin{align*}
&(k^2-p^2)\int_0^\infty
\cK_{\frac{\beta}{2k},m}(2kx)
\cK_{\frac{\beta}{2p},m}(2px)\d x\\
&=\frac{\pi}{\sin(2\pi m)}\sqrt{4kp}
\Bigg(
\frac{k^{m}p^{-m}}{\Gamma\big(\frac12+m-\frac{\beta}{2p}\big)\Gamma\big(\frac12-m-\frac{\beta}{2k}\big)}
-\frac{p^{m}k^{-m}}{\Gamma\big(\frac12+m-\frac{\beta}{2k}\big)
\Gamma\big(\frac12-m-\frac{\beta}{2p}\big)}\Bigg).
\end{align*}

\item[(ii)]
If $m=0$, then
\begin{align*}
&(k^2-p^2)\int_0^\infty
\cK_{\frac{\beta}{2k},0}(2kx)
\cK_{\frac{\beta}{2p},0}(2px)\d x\\
&= \sqrt{4kp}\;\! \frac{\psi\big(\frac12-\frac{\beta}{2k}\big)-\psi\big(\frac12-\frac{\beta}{2p}\big) + \ln(k)-\ln(p)}{
\Gamma\big(\frac12-\frac{\beta}{2p}\big)\Gamma\big(\frac12-\frac{\beta}{2k}\big)}.
\end{align*}

\item[(iii)]
If $m=\pm\frac12$, then
\begin{align*}
&(k^2-p^2)\int_0^\infty
\cK_{\frac{\beta}{2k},\frac12}(2kx)
\cK_{\frac{\beta}{2p},\frac12}(2px)\d x\\
&= \beta\;\!\frac{\frac12\psi\big(1-\frac{\beta}{2k}\big) + \frac12\psi\big(-\frac{\beta}{2k}\big)-\frac12 \psi\big(1-\frac{\beta}{2p}\big) - \frac12 \psi\big(-\frac{\beta}{2p}\big) + \ln(k)-\ln(p)}{
\Gamma\big(1-\frac{\beta}{2p}\big)\Gamma\big(1-\frac{\beta}{2k}\big)}.
\end{align*}
\end{enumerate}
\end{proposition}

\begin{proof}
The proof consists in an application of Lemma \ref{lem_int_id}.
Consider $k,p\in \C$ with $\Re(k)>0$, $\Re(p)>0$ and set $\lambda_1=-k^2$ and $\lambda_2=-p^2$. As shown in the proof of Theorem \ref{spectrum} the functions
$v_i$ defined by
\begin{equation*}
v_1(x)=\cK_{\frac{\beta}{2k},m}(2kx) \quad \hbox{and} \quad
v_2(x)=\cK_{\frac{\beta}{2p},m}(2px)
\end{equation*}
belong to $\Dom(L_{\beta,m^2}^\max)$ and are eigenfunctions of $L_{\beta,m^2}$ associated with the eigenvalues $\lambda_i$.
Let us then set $\Wr(v_1,v_2;0):=
\lim\limits_{x\searrow0}\Wr(v_1,v_2;x)$
and observe that $
\lim\limits_{x\to+\infty}\Wr(v_1,v_2;x)=0$, as a consequence of Proposition \ref{lem_properties}.
This yields directly
\begin{equation}\label{lagrange1}
(k^2-p^2)  \int_0^\infty v_1(x)v_2(x)\d x= \Wr(v_1,v_2;0).
\end{equation}

Let us now set
\begin{equation*}
u_{1,\pm}(x)=\cI_{\frac{\beta}{2k},\pm m}(2kx)
\quad \hbox{and} \quad u_{2,\pm}(x)=\cI_{\frac{\beta}{2p},\pm m}(2px).
\end{equation*}
Then, the identity \eqref{generic} leads to
\begin{align*}
v_1(x)&
= \frac{\pi}{\sin(2\pi m)}\Big(-\frac{u_{1,+}(x)}{\Gamma\big(\frac{1}{2}-m-\frac{\beta}{2k}\big)} + \frac{u_{1,-}(x)}{\Gamma\big(\frac{1}{2}+m-\frac{\beta}{2k}\big)}\Big),\\
v_2(x)&
= \frac{\pi}{\sin(2\pi m)}\Big(-\frac{u_{2,+}(x)}{\Gamma\big(\frac{1}{2}-m-\frac{\beta}{2p}\big)} + \frac{u_{2,-}(x)}{\Gamma\big(\frac{1}{2}+m-\frac{\beta}{2p}\big)}\Big),
\end{align*}
and with the expansion provided in \ref{Taylor_2} one directly infers that
\begin{align*}
\Wr(u_{1,+},u_{2,+};0)&=  \Wr(u_{1,-},u_{2,-};0)\,=\,0,\\
\Wr(u_{1,+},u_{2,-};0)&=  -\frac{4m k^{\frac12+m}p^{\frac12-m}}{\Gamma(1+2m)\Gamma(1-2m)}=-\frac{2\sin(2\pi m)}{\pi}
k^{\frac12+m}p^{\frac12-m},\\
\Wr(u_{1,-},u_{2,+};0)&=  \frac{4m k^{\frac12-m}p^{\frac12+m}}{\Gamma(1+2m)\Gamma(1-2m)}= \frac{2\sin(2\pi m)}{\pi} k^{\frac12-m}p^{\frac12+m}.
\end{align*}
As a consequence of these equalities one gets
\begin{align*}
& \Wr(v_1,v_2;0) \\
& =\frac{\pi}{\sin(2\pi m)}\bigg(\frac{2k^{\frac12+m}p^{\frac12-m}}{\Gamma\big(\frac12-m-\frac{\beta}{2k}\big)
\Gamma\big(\frac12+m-\frac{\beta}{2p}\big)}
-\frac{2k^{\frac12-m}p^{\frac12+m}}{\Gamma\big(\frac12+m-\frac{\beta}{2k}\big)
\Gamma\big(\frac12-m-\frac{\beta}{2p}\big)}\bigg).
\end{align*}
This proves $(i)$.
The equalities $(ii)$ and $(iii)$ can be proved similarly by
using \eqref{eq_a12} and \eqref{eq_a0}.
\end{proof}

By using the L'Hospital's rule one directly obtains:

\begin{corollary}\label{prop-lag2}
Let $\Re(k)>0$.
\begin{enumerate}
\item[(i)]
For $-1<\Re(m)<1$, $m\not \in \big\{-\frac{1}{2},0,\frac{1}{2}\big\}$ one has
\begin{equation*}
\int_0^\infty\cK_{\frac{\beta}{2k},m}(2kx)^2\d x=
\frac{\pi}{\sin(2\pi m)} \frac{2m+\frac{\beta}{2k}\psi\big(\frac12+m-\frac{\beta}{2k}\big)-\frac{\beta}{2k}\psi\big(\frac12-m-\frac{\beta}{2k}\big)}
{k\Gamma\big(\frac12+m-\frac{\beta}{2k}\big) \Gamma\big(\frac12-m-\frac{\beta}{2k}\big)}.
\end{equation*}
\item[(ii)]
For $m=0$,
\begin{equation*}
\int_0^\infty\cK_{\frac{\beta}{2k},0}(2kx)^2\d x=
\frac{1+\frac{\beta}{2k}\psi'\big(\frac12-\frac{\beta}{2k}\big)}{k\Gamma\big(\frac12-\frac{\beta}{2k}\big)^2}.
\end{equation*}
\item[(iii)]
For $m=\frac12$,
\begin{equation*}
\int_0^\infty\cK_{\frac{\beta}{2k},\frac12}(2kx)^2\d x=
-\frac{1+\frac{\beta}{4k}\psi'\big(-\frac{\beta}{2k}\big)+\frac{\beta}{4k}\psi'\big(1-\frac{\beta}{2k}\big)}{k\Gamma\big(-\frac{\beta}{2k}\big)\Gamma\big(1-\frac{\beta}{2k}\big)}.
\end{equation*}
\end{enumerate}
\end{corollary}

\subsection{The trigonometric type Whittaker equation}

Along with the standard Whittaker equation \eqref{whit}, sometimes called {\em hyperbolic type}, it is natural to consider the {\em trigonometric type}
Whittaker equation
\begin{equation}\label{Whittaker-trig}
\Big(L_{\beta,m^2}-\frac14\Big)f = \bigg(-\pder_z^2 +\Big(m^2 - \frac{1}{4}\Big)\frac{1}{z^2} - \frac{\beta}{z}-\frac{1}{4}\bigg)f = 0.
\end{equation}
In \cite[Sec.~2.6 \& 2.7]{DR2} we introduced the functions
\begin{equation}\label{eq:jbm}
\J{\beta}{m}(z) = \e^{\mp\i\frac{\pi}{2}(\frac{1}{2}+m)}\I{\mp\i\beta}{m}\big(\e^{\pm\i\frac{\pi}{2}}z\big)
\end{equation}
and
\begin{align}\label{Hpm-definition}
\begin{split}
\Hpm{\beta}{m}(z) =& \e^{\mp \i\frac{\pi}{2}\naw{\frac{1}{2}+m}}\K{\pm\i\beta}{m}(\e^{\mp\i\frac{\pi}{2}}z)\\
=& \frac{\pm \i\pi}{\sin(2\pi m)}\Big(\frac{\e^{\mp\i\pi m}\J{\beta}{m}(z)}{\Gamma\naw{\frac{1}{2}-m\mp\i\beta}} - \frac{\J{\beta}{-m}(z)}{\Gamma\naw{\frac{1}{2}+m\mp\i\beta}}\Big),
\end{split}
\end{align}
which solve \eqref{Whittaker-trig}.
Note that the function $\Hpm{\beta}{m}$ has been used in the proof of Theorem
\ref{spectrum} when dealing with positive eigenvalues of the Whittaker operators.

\subsection{Integral identities in the trigonometric case}

Here are the analogues of Proposition \ref{prop-lag1} and Corollary  \ref{prop-lag2}
in the trigonometric case. The approach can be mimicked from Section \ref{sec_int_id} because of the identity
$$
L_{\beta,m^2} \;\!\Hpm{\frac{\beta}{2\mu}}{m}(2\mu x) = \mu^2 \;\!\Hpm{\frac{\beta}{2\mu}}{m}(2\mu x)
$$
valid for any $\mu>0$.

\begin{proposition}
Let $\mu,\eta>0$ with $\mu<\pm\Im\big(\beta)$ and $\eta< \pm\Im\big(\beta)$.
\begin{enumerate}

\item[(i)]
If $-1<\Re(m)<1$, $m\not\in\big\{-\frac12,0,\frac12\big\}$, then
\begin{align*}
& (\mu^2-\eta^2)\int_0^\infty
\cH_{\frac{\beta}{2\mu},m}^\pm(2\mu x)\cH_{\frac{\beta}{2\eta},m}^\pm(2\eta x)\d x\\
&=\frac{\pi\e^{\mp\i\pi m}}{\sin(2\pi m)}\sqrt{4\mu\eta}
\Bigg(\frac{\mu^{m}\eta^{-m}}{\Gamma\big(\frac12+m\mp \i\frac{\beta}{2\eta}\big)\Gamma\big(\frac12-m\mp \i\frac{\beta}{2\mu}\big)}
- \frac{\eta^{m}\mu^{-m}}{\Gamma\big(\frac12+m\mp \i\frac{\beta}{2\mu}\big)
\Gamma\big(\frac12-m\mp \i\frac{\beta}{2\eta}\big)}
\Bigg).
\end{align*}

\item[(ii)]
If $m=0$, then
\begin{align*}
& (\mu^2-\eta^2)\int_0^\infty
\cH^\pm_{\frac{\beta}{2\mu},0}(2\mu x)
\cH^\pm_{\frac{\beta}{2\eta},0}(2\eta x)\d x\\
&= \sqrt{4\mu\eta} \;\!\frac{\psi\big(\frac12\mp \i\frac{\beta}{2\mu}\big)-\psi\big(\frac12\mp \i\frac{\beta}{2\eta}\big)+\ln(\mu)-\ln(\eta)}{
\Gamma\big(\frac12\mp \i\frac{\beta}{2\mu}\big)\Gamma\big(\frac12\mp \i\frac{\beta}{2\eta}\big)}.
\end{align*}

\item[(iii)]
If $m=\frac12$, then
\begin{align*}
&(\mu^2-\eta^2)\int_0^\infty
\cH^\pm_{\frac{\beta}{2\mu},\frac12}(2\mu x)
\cH^\pm_{\frac{\beta}{2\eta},\frac12}(2\eta x)\d x\\
&= \beta\frac{\frac12 \psi\big(1\mp \i \frac{\beta}{2\mu}\big)
+ \frac12 \psi\big(\mp \i \frac{\beta}{2\mu}\big)
-\frac12\psi\big(1\mp \i \frac{\beta}{2\eta}\big)
- \frac12\psi\big(\mp \i \frac{\beta}{2\eta}\big)
+\ln(\mu)-\ln(\eta)}{
\Gamma\big(1\mp \i \frac{\beta}{2\eta}\big)\Gamma\big(1\mp \i\frac{\beta}{2\mu}\big)}.
\end{align*}
\end{enumerate}
\end{proposition}

\begin{corollary}\label{cor:lag2}
Let $0<\mu<\pm\Im(\beta)$.
\begin{enumerate}

\item[(i)]
For $-1<\Re(m)<1$, $m\not \in \big\{-\frac{1}{2},0,\frac{1}{2}\big\}$ one has
\begin{equation*}
\int_0^\infty\cH_{\frac{\beta}{2\mu},m}^\pm(2\mu x)^2\d x=
\frac{\pi \e^{\mp\i\pi m}}{\sin(2\pi m)}\Bigg(\frac{2m\pm\i \frac{\beta}{2\mu}\psi\big(\frac12+m\mp \i \frac{\beta}{2\mu}\big)\mp\i \frac{\beta}{2k}\psi\big(\frac12-m\mp \i \frac{\beta}{2\mu}\big)}
{\mu \Gamma\big(\frac12+m\mp \i\frac{\beta}{2\mu}\big) \Gamma\big(\frac12-m\mp \i\frac{\beta}{2\mu}\big)}\Bigg).
\end{equation*}

\item[(ii)]
For $m=0$,
\begin{equation*}
\int_0^\infty\cH_{\frac{\beta}{2\mu},0}^\pm(2\mu x)^2\d x=
\frac{1\pm \i \frac{\beta}{2\mu}\psi'\big(\frac12\mp \i\frac{\beta}{2\mu}\big)}{\mu\Gamma\big(\frac12\mp \i\frac{\beta}{2\mu}\big)^2}.
\end{equation*}

\item[(iii)]
For $m=\frac12$,
\begin{equation*}
\int_0^\infty\cH_{\frac{\beta}{2\mu},\frac12}^\pm(2\mu x)^2\d x=
\frac{\pm \i -\frac{\beta}{4\mu}\psi'\big(\mp\i\frac{\beta}{2\mu}\big) - \frac{\beta}{4\mu}\psi'\big(1\mp\i\frac{\beta}{2\mu}\big)}{\mu\Gamma\big(\mp\i\frac{\beta}{2\mu}\big)\Gamma\big(1\mp\i\frac{\beta}{2\mu}\big)}.
\end{equation*}
\end{enumerate}
\end{corollary}

\section{The Bessel equation}
\setcounter{equation}{0}
\renewcommand{\theequation}{B.\arabic{equation}}

\subsection{The modified Bessel equation}

The {\em modified (or hyperbolic type) Bessel equation for dimension $1$}
\begin{equation}\label{lap7}
\bigg(-\partial_z^2+\Big(m^2-\frac14\Big)\frac{1}{z^2}+1\bigg)f = 0,
\end{equation}
is up to a trivial rescaling, a special case of the Whittaker equation with $\beta=0$.
Its theory was discussed at length in \cite[App.~A]{DR1}.
Nevertheless, we briefly discuss some of its elements here, explaining the parallel elements to the theory of the Whittaker equation, as well as the differences.

Let the \emph{modified Bessel function for dimension $1$} be
\begin{align}\label{asym2}
\begin{split}
\Ia_m(z)& = \sum_{n=0}^\infty\frac{\sqrt\pi\left(\frac{z}{2}\right)^{2n+m+\frac12}}{n!\Gamma(m+n+1)}\\
&=\frac{\sqrt\pi}{\Gamma(m+1)}\Big(\frac{z}{2}\Big)^{m+\frac12}
{}_{0}F_{1}\Big(m+1;\Big(\frac{z}{2}\Big)^2\Big).
\end{split}
\end{align}
The equation \eqref{lap7} is invariant with respect to $m\to-m$.
At the level of the function \eqref{asym2} this property is reflected by
\begin{equation*}
\Ia_m(z)=\e^{\mp\i \pi (\frac12+m)}\Ia_m(\e^{\pm\i\pi}z).
\end{equation*}
For the Wronskian we have
\begin{equation*}
\Wr(\Ia_m,\Ia_{-m};z)=-\sin(\pi m).
\end{equation*}

The function $\cK_m$ can be introduced
for $m\not \in \Z$ by
\begin{equation*}
\Ka_m(z)= \frac{1}{\sin(\pi m)}\big(-\Ia_{m}(z)+\Ia_{-m}(z)\big).
\end{equation*}
For $m\in\Z$ the definition is extended by continuity.
Note that the relation $\cK_m(z)=\cK_{-m}(z)$ holds, and that
\begin{align*}
\Wr(\Ka_m,\Ia_m;z)&=1.
\end{align*}

To make our presentation of the hyperbolic Bessel equation as much parallel to that
of the Whittaker equation as possible, we introduce the function
\begin{equation*}
\cX_{m}(z):=\frac12\Big(
\e^{-\i\pi(\frac{1}{2}+m)}\cK_{m}\big(\e^{\i\pi}z\big)
+ \e^{\i\pi(\frac{1}{2}+m)}\cK_{m}\big(\e^{-\i\pi}z\big)\Big).
\end{equation*}
Then the following relations hold:
\begin{align}
\nonumber \cX_m&=-
\frac{1}{\sin(\pi m)}\big(\cI_{m}-\cos(2m\pi)\cI_{-m}\big),\\
\label{uyu} \cI_m&=
\frac{1}{2\sin(m\pi)}\big(\cos( 2m\pi)\cK_m-\cX_m\big).
\end{align}
The precise relations between the Whittaker functions for $\beta=0$ and Bessel-type functions are of the form
\begin{align}\label{eqI_{0,m}}
\cI_{0,m}(z)&
= \frac{2}{\Gamma\big(\frac{1}{2}+m\big)}\cI_m\Big(\frac{z}{2}\Big),\\
\nonumber \cK_{0,m}(z)&
= \cK_m\Big(\frac{z}{2}\Big),\\
\nonumber \cX_{0,m}(z)&
= \cX_m\Big(\frac{z}{2}\Big).
\end{align}

\subsection{Recurrence relations}\label{sec_recurrence_r}

For the functions $\cI_m$ and $\cK_m$, the following recurrence relations hold:
\begin{align*}
\left(\partial_z+\Big(m-\frac12\Big)\frac{1}{z}\right)\cI_m(z)& =\cI_{m-1}(z), \\
\left(\partial_z+\Big(-m-\frac12\Big)\frac{1}{z}\right)\cI_m(z)&= \cI_{m+1}(z),\\
\left(\partial_z+\Big(m-\frac12\Big)\frac{1}{z}\right)\cK_m(z)& =-\cK_{m-1}(z), \\
\left(\partial_z+\Big(-m-\frac12\Big)\frac{1}{z}\right)\cK_m(z)&= -\cK_{m+1}(z).
\end{align*}

\subsection{Integral identities}

It is proved for example in \cite[Sec.~A.8]{DR1} that for $|\Re (m)|<1$, $m\neq 0$
and for $\Re(a+b)>0$ one has
\begin{equation}\label{int0}
\int_0^\infty \Ka_m(ax)\Ka_m(bx)\d x
= \frac{(a^{2m}-b^{2m})a^{\frac12-m}b^{\frac12-m}}{\sin(\pi m)(a^2-b^2)}.
\end{equation}
Observe that the r.h.s.~of \eqref{int0} can be extended by continuity
to $a=b$ and $m=0$ since
\begin{eqnarray*}
\lim_{m\to 0}\frac{(a^{2m}-b^{2m})a^{\frac12-m}b^{\frac12-m}}{\sin(\pi m)(a^2-b^2)}
&=&\frac{2}{\pi}\frac{\big(\ln(a)-\ln(b)\big)a^{\frac12}b^{\frac12}}{a^2-b^2},\\
\lim_{b\to a}\frac{(a^{2m}-b^{2m})a^{\frac12-m}b^{\frac12-m}}{\sin(\pi m)(a^2-b^2)}
&=&\frac{ m}{\sin (\pi m) a}, \\
\lim_{m\to 0}\lim_{b\to a}\frac{(a^{2m}-b^{2m})a^{\frac12-m}b^{\frac12-m}}{\sin(\pi m)(a^2-b^2)}&=&\frac{1}{\pi a}.
\end{eqnarray*}

We shall need another integral identity:

\begin{proposition}
For $|\Re (m)|<1$ and $\Re(a+b)>0$ one has
\begin{align}\label{quad}
& \int_0^\infty x^2 \Ka_m(ax)\Ka_m(bx)\d x \\
& =\frac{4a^{\frac12-m}b^{\frac12-m}\big\{(m-1)(a^{2m+2}-b^{2m+2})
+(m+1)a^2b^2(b^{2m-2}-a^{2m-2})\big\}}
{\sin(\pi m)(b^2-a^2)^3}.
\end{align}
In addition, the following limiting cases hold:
\begin{eqnarray*}
\int_0^\infty x^2 \Ka_0(ax)\Ka_0(bx)\d x
& = & -\frac{8a^{\frac12}b^{\frac12}}{\pi(b^2-a^2)^2}
+ \frac{8a^{\frac12}b^{\frac12}(a^2+b^2)}{\pi(b^2-a^2)^3}
\big(\ln(b)-\ln(a)\big), \\
\int_0^\infty x^2 \Ka_m(ax)^2\d x
& = & \frac{2m(1-m^2)}{3 a^3 \sin(\pi m)}, \\
\int_0^\infty x^2 \Ka_0(ax)\Ka_0(ax)\d x
& =  & \frac{2}{3\pi  a^3}.
\end{eqnarray*}
\end{proposition}

\begin{proof}
Assume first that $-1<\Re(m)<0$.
By using twice the recurrence relations of Section \ref{sec_recurrence_r} one gets
\begin{equation*}
\int_0^\infty x^2 \Ka_m(ax)\Ka_m(bx)\d x
= -\Big(\partial_b{+}\frac{m{+}\frac12}{b}\Big)\int_0^\infty x\Ka_m(ax)\Ka_{m+1}(bx)\d x,
\end{equation*}
and
\begin{equation*}
\int_0^\infty x \Ka_m(ax)\Ka_{m+1}(bx)\d x
= -\Big(\partial_a{+}\frac{m{+}\frac12}{a}\Big)\int_0^\infty \Ka_{m+1}(ax)\Ka_{m+1}(bx)\d x.
\end{equation*}
Then, we infer that
\begin{align*}
&  \int_0^\infty x^2 \Ka_m(ax)\Ka_m(bx)\d x \\
&=
\Big(\partial_a{+}\frac{m{+}\frac12}{a}\Big)  \Big(\partial_b{+}\frac{m{+}\frac12}{b}\Big)\int_0^\infty \Ka_{m+1}(ax)\Ka_{m+1}(bx)\d x \\
&=\Big(\partial_a{+}\frac{m{+}\frac12}{a}\Big)  \Big(\partial_b{+}\frac{m{+}\frac12}{b}\Big) \bigg(
 \frac{(a^{2m+2}-b^{2m+2})a^{-\frac12-m}b^{-\frac12-m}}{\sin(\pi m)(b^2-a^2)}
\bigg) \\
&=a^{-\frac12-m}b^{-\frac12-m} \partial_a \partial_b
\bigg(\frac{a^{2m+2}-b^{2m+2}}{\sin(\pi m)(b^2-a^2)}\bigg)
\end{align*}
where we have used \eqref{int0} with $m+1$ instead of $m$,
and the fact that $\partial_a (a^{-\frac12-m})= -\frac{m+\frac{1}{2}}{a}a^{-\frac12-m}$.
Clearly, a similar relation holds for $a$ replaced by $b$.
By computing the derivatives, one gets the expressions provided in the statement.
This proves \eqref{quad} for $-1<\Re(m)<0$. We then extend the equality to $|\Re(m)|<1$ by analytic continuation.
Finally, the limiting cases are obtained by taking the limit $m\to 0$ in the first case,
the limit $b\to a$ in the second case, and from this result the limit $m\to 0$.
Note that the same result is obtained if we take the limits in the reverse order.
\end{proof}

\subsection{The degenerate case}

For $m\in\Z$ the following relation holds:
\begin{equation*}
\cI_{-m}(z)=\cI_m(z).
\end{equation*}
Assuming that $m\in \N$, we also have
\begin{equation*}
\Ia_m(z) = \Big(\frac{z}{2}\Big)^{m+\frac12}
\sum_{k=0}^\infty\frac{\sqrt\pi}{k!(m+k)!}
\Big(\frac{z}{2}\Big)^{2k},
\end{equation*}
and
\begin{align*}
\nonumber \Ka_m(z)&=(-1)^{m+1} \frac{2}{\pi}\;\!\ln\Big(\frac{z}{2}\Big)\;\!\Ia_m(z)
\\
\nonumber &\quad +\frac{(-1)^m}{\sqrt\pi}\Big(\frac{z}{2}\Big)^{m+\frac12}\sum_{k=0}^\infty
\frac{\psi(k+1)+\psi(m+k+1)}{k!(m+k)!}\Big(\frac{z}{2}\Big)^{2k}
\\
&\quad +\frac{(-1)^{m}}{\sqrt\pi}\Big(\frac{z}{2}\Big)^{m+\frac12}\sum_{j=1}^{m} (-1)^{j} \frac{(j-1)!}{(m-j)!} \Big(\frac{z}{2}\Big)^{-2j}.
\end{align*}

\subsection{The half-integer case}

The half-integer case of the hyperbolic Bessel equation is a special case
of the doubly degenerate case of the Whittaker equation.
However, it is worthwhile to discuss it separately. In particular, for $n\in \N$ the function $\cI_{-\frac12-n}$ is not proportional to
the function $\cI_{0,-\frac12-n}$, which is identically $0$ by \eqref{eqI_{0,m}}.

By analogy of the presentation of Section \ref{sec_dd} we can divide the half-integer case into two
regions, namely
{\bf Region $\mathrm{I}_-$} with $m\in-\frac12-\N$, and
{\bf Region $\mathrm{I}_+$} with $m\in\frac12+\N$.
The following schematic diagram of various special cases for the Bessel equation is an analog of Fig. 2.
\bigskip

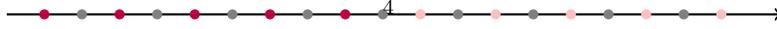
\begin{figure}[H]
\begin{center}
\begin{tikzpicture}
\setlength{\unitlength}{1cm}

\draw[thick, ->] (-5,0)--(5.3,0);
\put(5,0.1){$m$};

\filldraw[purple](-4.5,0)circle(0.06);
\filldraw[purple](-3.5,0)circle(0.06);
\filldraw[purple](-2.5,0)circle(0.06);
\filldraw[purple](-1.5,0)circle(0.06);
\filldraw[purple](-0.5,0)circle(0.06);
\filldraw[pink](4.5,0)circle(0.06);
\filldraw[pink](3.5,0)circle(0.06);
\filldraw[pink](2.5,0)circle(0.06);
\filldraw[pink](1.5,0)circle(0.06);
\filldraw[pink](0.5,0)circle(0.06);
\put(-4.1,0.1){$\scriptstyle{-}4$};
\put(-3.1,0.1){$\scriptstyle{-}3$};
\put(-2.1,0.1){$\scriptstyle{-}2$};
\put(-1.1,0.1){$\scriptstyle{-}1$};
\put(0.1,0.1){$\scriptstyle0$};
\put(1.1,0.1){$\scriptstyle1$};
\put(2.1,0.1){$\scriptstyle2$};
\put(3.1,0.1){$\scriptstyle3$};
\put(4.1,0.1){$\scriptstyle4$};

\filldraw[gray](-4,0)circle(0.06);
\filldraw[gray](-3,0)circle(0.06);
\filldraw[gray](-2,0)circle(0.06);
\filldraw[gray](-1,0)circle(0.06);
\filldraw[gray](0,0)circle(0.06);
\filldraw[gray](1,0)circle(0.06);
\filldraw[gray](2,0)circle(0.06);
\filldraw[gray](3,0)circle(0.06);
\filldraw[gray](4,0)circle(0.06);

\put(1.1,0.1){$\scriptstyle1$};
\put(2.1,0.1){$\scriptstyle2$};
\put(3.1,0.1){$\scriptstyle3$};
\put(4.1,0.1){$\scriptstyle4$};

\end{tikzpicture}
\caption{The two regions in the half-integer case}
\end{center}
\label{fig:M3}
\end{figure}

Note that unlike for the Whittaker equation, in both regions
$\mathrm{I}_-$ and  $\mathrm{I}_+$ the functions
$\cI_m$, $\cI_{-m}$ and $\cK_m$ are well defined and distinct,
and any two of them form a basis of solutions of \eqref{lap7}.
In this case all solutions are elementary functions:
For $n\in\N$ and $m=\pm(\frac12+n)$ one has
\begin{align}
\nonumber \cK_{\pm(\frac12+n)}(z)
&= (-1)^nn!(2z)^{-n}\e^{-z}L_n^{(-1-2n)}(2z),\\
\nonumber \cX_{\pm(\frac12+n)}(z)
&=\pm(-1)^nn!(2z)^{-n}\e^{z}L_n^{(-1-2n)}(-2z),\\
\cI_{\frac12+n}(z)
&=-\frac{1}{2}n!(2z)^{-n}\Big(\e^{-z}L_n^{(-1-2n)}(2z)-\e^{z}L_n^{(-1-2n)}(-2z)\Big),
\label{poi3}\\
\cI_{-\frac12-n}(z)
&=\frac{1}{2} n!(2z)^{-n}\Big(\e^{-z}L_n^{(-1-2n)}(2z)+\e^{z}L_n^{(-1-2n)}(-2z)\Big).  \label{poi4}
\end{align}
Note also that \eqref{poi3} and  \eqref{poi4} are special cases of
\eqref{uyu}, namely
\begin{align*}
\cI_{\frac12+n}(z)
&= \frac{(-1)^{n+1}}{2}\big(\cK_{\frac12+n}(z)-\cX_{\frac12+n}(z)\big), \\
\cI_{-\frac12-n}(z)
&=  \frac{(-1)^n}{2}\big(\cK_{\frac12+n}(z)+\cX_{\frac12+n}(z)\big).
\end{align*}

\subsection{The standard Bessel equation}\label{sec:standard_bessel}

The {\em standard (or trigonometric-type) Bessel equation for dimension $1$}
\begin{equation}\label{lap71}
\bigg(-\partial_z^2+\big(m^2-\frac14\big)\frac{1}{z^2}-1\bigg)f = 0,
\end{equation}
is up to a trivial rescaling, a special case of the trigonometric-type Whittaker equation with $\beta=0$.
One can introduce the following functions which solve this equation (see \cite[App.~A]{DR1} for more information)\;:
\begin{equation*}
\Ja_m(z) = \e^{\pm\i\frac{\pi}{2}(m+\frac12)}\Ia_m(\e^{\mp\i\frac{\pi}{2}}z)
= \sum_{n=0}^\infty\frac{(-1)^n\sqrt\pi\left(\frac{z}{2}\right)^{2n+m+\frac12}}{n!\Gamma(m+n+1)},
\end{equation*}
\begin{equation*}
\Ha_m^{\pm}(z) = \e^{\mp\i\frac{\pi}{2}(m+\frac12)}
\Ka_m(\e^{\mp\i\frac{\pi}{2}} z)
=\pm\i\frac{\e^{\mp\i \pi m}\Ja_m(z)- \Ja_{-m}(z)}{\sin (\pi m)},
\end{equation*}
and
\begin{equation*}
\cY_m(z):=\frac{\cos(\pi m)\cJ_m(z)-\cJ_{-m}(z)}{\sin(\pi m)}.
\end{equation*}

\subsection{The zero eigenvalue Whittaker equation}\label{B5}

The zero eigenvalue Whittaker equation is provided by the equation
\begin{equation}\label{whit0}
L_{\beta,m^2}f := \bigg(-\partial_z^2+\Big(m^2-\frac14\Big)\frac{1}{z^2}-\frac{\beta}{z}\bigg)f=0.
\end{equation}
It is easy to see that if $v$ solves the trigonometric Bessel equation of dimension 1 \eqref{lap7} with parameter $2m$, then the function $f$ defined by $f(x):= (\beta x)^{\frac14}v(2\sqrt{\beta x})$ solves the equation \eqref{whit0}.

One can also obtain solutions of \eqref{whit0} by rescaling solutions of the hyperbolic-type or trigonometric-type Whittaker equation:

\begin{proposition}\label{propB1}
For any fixed $x\in \R_+$, $m \in \Pi$ and $\beta \in \C^{\times}$, one has
\begin{align}
\label{zero.1} \lim_{k\to0}\Big(\frac{1}{2k}\Big)^{\frac12+m}
\cI_{\frac{\beta}{2k},m}(2kx)&= \beta^{-m-\frac12}
\frac{(\beta x)^{\frac14}}{\sqrt{\pi}}\cJ_{2m}(2\sqrt{\beta x}),\\
\label{zero.2}
\lim_{k \to0}\Big(\frac{1}{2k}\Big)^{\frac12+m}
\cJ_{\frac{\beta}{2k},m}(2k x)&= \beta^{-m-\frac12}
\frac{(\beta x)^{\frac14}}{\sqrt{\pi}}\cJ_{2m}(2\sqrt{\beta x}).
\end{align}
For any fixed $x\in \R_+$, any $m \in \Pi$ and $\beta \in \C^\times$, one has
\begin{align}
\label{zero3}
\lim_{k\to0} \mp \i \frac{\Gamma\big(\frac12+m-\frac{\beta}{2k}\big)}{\sqrt{\pi}}
\Big(\frac{\beta}{2k}\Big)^{\frac12-m}
\cK_{\frac{\beta}{2k},m}(2kx)&=
(\beta x)^{\frac14}\cH_{2m}^\pm(2\sqrt{\beta x}),\\
\label{zero4}
\lim_{\mu \to0}\frac{\Gamma\big(\frac12+m\mp\i\frac{\beta}{2\mu }\big)}{\sqrt{\pi}}
\Big(\frac{\beta}{2\mu }\Big)^{\frac12-m}
\cH_{\frac{\beta}{2\mu },m}^\pm(2\mu x)&
=(\beta x)^{\frac14}\cH_{2m}^\pm(2\sqrt{\beta x}).
\end{align}
where the first limit is taken such that $\pm \big(\arg(\beta)-\arg(k)\big)\in ]\varepsilon, \pi-\varepsilon[$ with $\varepsilon > 0$,
and the second limit is taken with $\mu>0$ and is valid if $\Re(\beta)>0$.
\end{proposition}

\begin{proof}
Using the definition of Pochhammer's symbol recalled in Section \ref{subsec:general}, one infers that
\begin{equation*}
\lim_{k \to 0}\Big(\frac12+m\mp\frac{\beta}{2k}\Big)_j(\pm2k)^j=(-\beta)^j.
\end{equation*}
In addition, for all $k \in \C$ with $|k|<1$, one has
\begin{equation*}
\Big|\Big(\frac12+m\mp\frac{\beta}{2k}\Big)_j(\pm2k)^j\Big|\leq c^j j!
\end{equation*}
for some constant $c$ independent of $k$ and $j$. Hence, by an application of the version of the Lebesgue dominated convergence theorem for series, one gets
\begin{equation*}
\lim_{k\to0}  \sum_{j=0}^\infty
\frac{\big(\frac12+m\mp\frac{\beta}{2k}\big)_j(\pm2kx)^j}
{\Gamma(1+2m+j)j!}= \sum_{j=0}^\infty\frac{(-\beta x)^j}{\Gamma(1+2m+j)j!},
\end{equation*}
which leads directly to the equality \eqref{zero.1}.
The equality \eqref{zero.2} can then be deduced from \eqref{zero.1}
by using the relation \eqref{eq:jbm} between the functions $\cI_{\beta,m}$ and $\cJ_{\beta,m}$.

For \eqref{zero3}, by using successively \eqref{generic}, \eqref{gamma},
\eqref{zero.1}, and \cite[App.~A.5]{DR1} one gets
\begin{align*}
\nonumber &\mp\i\frac{\Gamma\big(\frac12+m-\frac{\beta}{2k}\big)}{\sqrt{\pi}}
\Big(\frac{\beta}{2k}\Big)^{\frac12-m}
\cK_{\frac{\beta}{2k},m}(2kx)\\
\nonumber &= \frac{\mp\i\sqrt{\pi}}{\sin(2\pi m)}
\Big(\frac{\beta}{2k}\Big)^{\frac12-m}
\Big(-\frac{\Gamma\big(\frac12+m-\frac{\beta}{2k}\big)}{\Gamma\big(\frac{1}{2}-m-\frac{\beta}{2k}\big)}
\cI_{\frac{\beta}{2k},m}(2kx)
+ \cI_{\frac{\beta}{2k},-m}(2kx)\Big) \\
 & =
\frac{\mp\i\sqrt{\pi}}{\sin(2\pi m)}
\Big(-\e^{\mp\i\pi 2m}\Big(\frac{\beta}{2k}\Big)^{\frac12+m}
\cI_{\frac{\beta}{2k},m}(2kx)
+ \Big(\frac{\beta}{2k}\Big)^{\frac12-m}\cI_{\frac{\beta}{2k},-m}(2kx)\Big) + o( 1 ) \\
 & =
\frac{\mp\i}{\sin(2\pi m)}
\Big(-\e^{\mp\i\pi 2m}(\beta x)^{\frac14}\cJ_{2m}(2\sqrt{\beta x})
+ (\beta x)^{\frac14}\cJ_{-2m}(2\sqrt{\beta x})\Big) + o( 1 ) \\
\nonumber &= (\beta x)^{\frac14}\cH_{2m}^\pm(2\sqrt{\beta x}) + o (1 ),
\end{align*}
where we have used that $\pm \arg\big(\frac{\beta}{2k}\big) \in ]0,\pi]$  and that $\big|\arg\big(-\frac{\beta}{2k}\big)\big|<\pi-\varepsilon$ for $\varepsilon>0$. The equality \eqref{zero4} can then be deduced from \eqref{zero3} by using the relation \eqref{Hpm-definition} between the functions  $\cK_{\beta,m}$ and $\cH_{\beta,m}^\pm$.
\end{proof}

The following lemma plays a key role in the above proof.

\begin{lemma}\label{lem_key}
Let $a,b \in \C$. For $|z|\to\infty$ with $|\arg(z)|<\pi-\varepsilon$ and $\varepsilon>0$ one has
\begin{equation}\label{gamma}
\lim_{z\to \infty} \frac{\Gamma(a+z)}{\Gamma(b+z)}z^{b-a} =1.
\end{equation}
\end{lemma}

\begin{proof}
Recall first the logarithmic version of Stirling formula \cite[Eq.~6.1.41]{AS}\;\!:
\begin{equation*}
\ln\big(\Gamma(z)\big)= z\ln(z)-z+\frac12\ln(2\pi)-\frac12\ln(z)+
O\Big(\frac{1}{z}\Big).
\end{equation*}
This readily implies that
\begin{equation*}
\ln\big(\Gamma(a+z)\big)-\ln\big(\Gamma(b+z)\big)
+ (b-a)\ln(z) \underset{z\to\infty}\to0.
\end{equation*}
After exponentiation it leads to the statement.
\end{proof}

\subsection{Integrals for zero eigenvalue solutions of the Whittaker equation}

Based on the results of the previous sections and on Lemma \ref{lem_int_id}, one easily gets:

\begin{proposition} \label{prop_last}
Let $k\in \C$ with $\Re(k)>0$ and let $\beta\in \C$ with $\pm \Im(\sqrt{\beta})>0$.
If $m\in \C$ with $|\Re(m)|<1$, one has
\begin{align}\notag
& \int_0^\infty(\beta x)^{\frac14}\cH_{2m}^{\pm}(2\sqrt{\beta x})
\cK_{\frac{\beta}{2k},m}(2k x)
\d x\\ \label{into1}
&=\mp \i
\frac{(2\pi k\beta)^{\frac{1}{2}}}{\sin(2\pi m)}\Bigg(
\frac{\big(\frac{\beta}{2k}\big)^{-m}}{\Gamma\big(\frac12-m-\frac{\beta}{2k}\big)}
-\e^{\mp\i2\pi m}
\frac{\big(\frac{\beta}{2k}\big)^{m}}{\Gamma\big(\frac12+m-\frac{\beta}{2k}\big)}\Bigg).
\end{align}
Observe that the r.h.s.~of \eqref{into1} can be extended by continuity
to $m\in\big\{0,\frac12\big\}$ with
\begin{align*}
\eqref{into1}\big|_{m=0}=&
\mp \i \frac{1}{\sqrt{\pi}}\frac{(2k \beta)^{\frac{1}{2}}}{\Gamma\big(\frac{1}{2}-\frac{\beta}{2k}\big)}
\bigg(\psi\Big(\frac{1}{2}-\frac{\beta}{2k}\Big)-\ln\Big(\frac{\beta}{2k}\Big)\pm \i \pi\bigg),\\
\eqref{into1}\big|_{m=\frac12}=&
\pm\i\frac{1}{\sqrt{\pi}}\frac{2k}{\Gamma\big(-\frac{\beta}{2k}\big)}
\bigg(\frac12\psi\Big(-\frac{\beta}{2k}\Big)
+\frac12\psi\Big(1-\frac{\beta}{2k}\Big)-\ln\Big(\frac{\beta}{2k}\Big)\pm \i \pi\bigg).
\end{align*}
\end{proposition}

In the next proposition, we consider the integral of $\big((\beta x)^{\frac14}\cH_{2m}^{\pm}(2\sqrt{\beta x})\big)^2$ which cannot be computed by the same means.

\begin{proposition}\label{propB6}
Let $\beta\in \C$ with $\pm \Im(\sqrt{\beta})>0$. For all $-1 < \Re( m ) < 1$, one has
\begin{align}
& \int_0^\infty\Big((\beta x)^{\frac14}\cH_{2m}^{\pm}(2\sqrt{\beta x})\Big)^2\d x = \frac{m \big(4m^2-1\big) \e^{\mp \i\pi 2m}}{3\beta\sin(2\pi m)}. \label{eq:last_3}
\end{align}
\end{proposition}

\begin{proof}
Let us consider for $|\Re(m)|<2$ the integral $\int_0^\infty y^2\cK_{2m}(y)^2\d y$.
After a change of variable and by taking into account the relation between the MacDonald function for dimension $1$ and the usual MacDonald function one infers from \cite{Wolfram} that
\begin{align}
\nonumber \int_0^\infty y^2\cK_{2m}(y)^2\d y
& = \frac{2}{3 \pi}\Gamma(2-2m)\Gamma(2+2m) \\
\nonumber & = \frac{4m}{3 \pi}(1-2m)(1+2m)\Gamma(1-2m)\Gamma(2m) \\
& = \frac{4m}{3 \sin(2\pi m)}\big(1-4m^2\big) .
\end{align}
Note that this result can also be obtained by an analytic continuation of the
result obtained in \ref{quad}.
By a contour integration with a vanishing contribution at infinity, one gets that
for $\pm \Im(\sqrt{\beta})>0$,
\begin{align*}
& \int_0^\infty\Big((\beta x)^{\frac14}\cH_{2m}^{\pm}(2\sqrt{\beta x})\Big)^2\d x\\
&=\frac{1}{2}\int_0^\infty2\sqrt{\beta x}\e^{\mp\i\pi( 2 m+\frac12)}
\cK_{2m}\big(\e^{\mp\i\frac{\pi}{2}}2\sqrt{\beta x}\big)^2\d x\\
&=-\frac{\e^{\mp\i \pi 2 m}}{4\beta}\int_0^\infty y^2\cK_{2m}(y)^2\d y .
\end{align*}
This leads to the statement of the proposition.
\end{proof}

\begin{remark}\label{curious}
Curiously, a naive computation suggests incorrectly that
\begin{equation*}
\int_0^\infty\Big( (\beta x)^{\frac14}\cH_{2m}^{\pm}(2\sqrt{\beta x})\Big)^2\d x=0.
\end{equation*}
Indeed, for $m\not\in\big\{-\frac12,0,\frac12\big\}$ and $k\in \C$ with $\Re(k)>0$,
and such that $\pm \big(\arg(\beta)-\arg(k)\big)\in ]\varepsilon, \pi-\varepsilon[$ with $\varepsilon > 0$,
one has
\begin{align}
\label{eq_last_one}
&\int_0^\infty(\beta x)^{\frac14}\cH_{2m}^{\pm}(2\sqrt{\beta x})
\bigg[ \mp \i \frac{\Gamma\big(\frac12+m-\frac{\beta}{2k}\big)}{\sqrt{\pi}}
\Big(\frac{\beta}{2k}\Big)^{\frac12-m} \cK_{\frac{\beta}{2k},m}(2k x)\bigg]
\d x\\
\nonumber &= \mp \i \frac{\Gamma\big(\frac12+m-\frac{\beta}{2k}\big)}{\sqrt{\pi}}
\Big(\frac{\beta}{2k}\Big)^{\frac12-m} \int_0^\infty(\beta x)^{\frac14}\cH_{2m}^{\pm}(2\sqrt{\beta x})
\cK_{\frac{\beta}{2k},m}(2k x)
\d x\\
\nonumber &=-
\frac{\beta}{\sin(2\pi m)}
\Bigg(
\frac{\Gamma\big(\frac12+m-\frac{\beta}{2k}\big)}{\Gamma\big(\frac12-m-\frac{\beta}{2k}\big)}\Big(\frac{\beta}{2k}\Big)^{-2m}
-\e^{\mp\i\pi 2m}
\Bigg).
\end{align}
By taking a limit as $k\to 0$, one obtains from Lemma \ref{lem_key} that
\begin{equation}\label{eq_last_two}
\lim_{k\to 0} \Bigg(
\frac{\Gamma\big(\frac12+m-\frac{\beta}{2k}\big)}{\Gamma\big(\frac12-m-\frac{\beta}{2k}\big)}\Big(\frac{\beta}{2k}\Big)^{-2m}
-\e^{\mp\i\pi 2m}
\Bigg) = 0.
\end{equation}
Although by \eqref{zero3} the term in the square braket of \eqref{eq_last_one} converges
pointwise to $(\beta x)^{\frac14}\cH_{2m}^\pm(2\sqrt{\beta x})$,
a limit $\lim_{k\to 0}$ and the integral in \eqref{eq_last_one} can certainly not be exchanged, since otherwise it would lead to a contradiction.
\end{remark}

To conclude, we give a lemma which was used in the proof of Proposition \ref{prop:not_continuous}.

\begin{lemma}\label{lemB5}
For $|z|\to\infty$ with $|\arg(z)|<\pi-\varepsilon$ and $\varepsilon>0$ one has
\begin{align*}
\psi(b+z)-\psi(c+z)=&\frac{b-c}{z}+\frac{(b-c)(1-b-c)}{2z^2} \\
&+\frac{(b-c)[ 1 - 3 (b + c) + 2(  b^2+bc+c^2 ) ]}{6z^3} + O\Big(\frac1{z^4}\Big).
\end{align*}
\end{lemma}

\begin{proof}
The asymptotic expansion of the $\psi$ function is provided in \cite[Eq.~6.3.18]{AS}
and reads as $|z|\to\infty$ with $|\arg(z)|<\pi-\varepsilon$ and $\varepsilon>0$\;\!:
\begin{equation*}
\psi(z)=\ln(z)-\frac{1}{2z}-\frac{1}{12z^2}+O\Big(\frac{1}{z^4}\Big).
\end{equation*}
Hence
\begin{align*}
&\psi(b+z)-\psi(c+z) \\
& = \ln(b+z)-\frac{1}{2(b+z)} -\frac{1}{12(b+z)^2}
-\ln(c+z)+\frac{1}{2(c+z)}+\frac{1}{12(c+z)^2}+O\Big(\frac{1}{z^4}\Big)\\
&=
\ln\Big(1+\frac{b}{z}\Big)-\frac{1}{2z(1+\frac{b}{z})}-\frac{1}{12z^2(1+\frac{b}{z})^2} \\
&\quad - \ln\Big(1+\frac{c}{z}\Big)+\frac{1}{2z(1+\frac{c}{z})} + \frac{1}{12z^2(1+\frac{c}{z})^2} +O\Big(\frac{1}{z^4}\Big) \\
&= \frac{b-c}{z}+\frac{c^2-b^2 +b-c}{2z^2} + \frac{b-c - 3b^2 + 3c^2 + 2 b^3 - 2 c^3 }{6z^3} + O\Big(\frac{1}{z^4}\Big)
\end{align*}
which leads directly to the statement.
\end{proof}

\begin{acknowledgements}
S.~R. was supported
by the grant\emph{Topological invariants
through scattering theory and noncommutative geometry} from Nagoya University,
and by JSPS Grant-in-Aid for scientific research (C) no 18K03328.
\end{acknowledgements}

\end{document}